\documentclass[11pt]{article}
\usepackage{geometry}
\usepackage{xr-hyper}
\usepackage[colorlinks=true, allcolors=blue]{hyperref}
\usepackage{amsmath,amssymb,amsfonts,amsthm,mathtools,physics}
\usepackage{setspace,algorithm}
\usepackage[capitalize,noabbrev,nosort]{cleveref}
\crefname{assumption}{Assumption}{Assumptions}

\usepackage{natbib}
\usepackage{graphicx}
\usepackage{comment,xcolor,soul}
\usepackage{placeins}
\usepackage{enumerate}
\usepackage{enumitem}
\usepackage{booktabs}
\usepackage{adjustbox}
\usepackage{pdflscape}

\newtheorem{theorem}{Theorem}
\newtheorem{assumption}{Assumption}
\newtheorem{corollary}{Corollary}
\newtheorem{lemma}{Lemma}
\newtheorem{remark}{Remark}

\newcommand{\ind}{\perp\!\!\!\!\perp}

\renewcommand{\P}{\mathbb{P}}

\newcommand{\dto}{\leadsto}

\newcommand{\G}{\mathbb{G}}
\DeclareMathOperator{\Cov}{Cov}
\DeclareMathOperator{\logit}{logit}
\newcommand{\Dcar}{D_\mathrm{CAR}}
\newcommand{\Dcart}{D_{\mathrm{CAR},t}}
\newcommand{\trans}{^{\top}}
\DeclareMathOperator{\expit}{expit}

\DeclareMathOperator*{\argmax}{arg\,max} %

\usepackage{accents}
\newcommand\munderbar[1]{\underaccent{\bar}{#1}}

\usepackage{tikz}

\newcommand{\mathperiod}{\,\text{.}}
\newcommand{\mathcomma}{\,\text{,}}
\newcommand{\mathsemicolon}{\,\text{;}}

\crefalias{enumi}{assumption}   %
\crefalias{enumii}{assumption}  %

\allowdisplaybreaks

\title{Nonparametric Estimation of  Optimal Stochastic Just-In-Time Adaptive Interventions for Distal Outcomes}
\author{Jack M.\@ Wolf, Nandita Mitra, \& Ashkan Ertefaie}
\date{
	Division of Biostatistics, University of Pennsylvania\\[2ex]
	\today
}

\makeatletter
\newcommand*{\addFileDependency}[1]{%
	\typeout{(#1)}
	\@addtofilelist{#1}
	\IfFileExists{#1}{}{\typeout{No file #1.}}
}
\makeatother

\begin{document}
	
	\maketitle
	\onehalfspacing 
	
	\begin{abstract}
		Mobile and wearable technologies enable the delivery of just-in-time adaptive interventions (JITAIs)—interventions that adapt treatment delivery to an individual’s rapidly changing internal state and context in real-time, real-world settings. Estimating optimal JITAIs, however, remains challenging because these studies often involve dozens of decision points per individual, and existing methods can produce unstable and irregular estimators with substantial bias and slow convergence rates. 
		Advanced reinforcement learning approaches may be difficult to interpret and often target proximal, discounted outcomes rather than the distal end-of-study outcomes that define long-term success in many behavioral and clinical studies.  
		To address these challenges, we develop a nonparametrically efficient estimator of the regimen–response curve for distal outcomes under a class of stochastic policies and introduce a data-adaptive tilting procedure to stabilize estimation in settings with many decision points. 
		We show that the estimated regimen–response curve converges weakly to a Gaussian process, enabling simultaneous confidence bands, and we derive asymptotic theory for the optimizer of the curve, thereby enabling inference for the learned optimal stochastic policy. 
		These developments provide a unified framework for estimation, inference, and optimization of stochastic JITAIs for distal outcomes.
	\end{abstract}
	
	
	\newpage
	\section{Introduction}
	\label{sec:introduction}

	Modern mobile and wearable technologies provide unprecedented opportunities to deliver just-in-time adaptive interventions (JITAIs; \citealp{nahum-shani_just--time_2017})---intervention designs that adapt to an individual’s rapidly changing internal state in real-time, real-world settings.
	These interventions are especially attractive in behavioral settings, where treatment effectiveness may depend strongly on momentary vulnerability, receptivity, and burden.
	Recent advances in micro-randomized trials (MRTs; \citealp{klasnja_microrandomized_2015}) have made it possible to collect intensive longitudinal data that can empirically inform the construction of JITAIs by repeatedly randomizing treatment assignment across many decision points \citep{liao_off-policy_2021,qian_microrandomized_2022}.
	A motivating example is Mobile Assistance for Regulating Smoking (MARS; \citealp{nahum-shani_mobile_2021}), a ten-day MRT conducted among people attempting to quit smoking that was designed to inform the development of a mobile health intervention for promoting real-time, real-world engagement in evidence-based self-regulatory strategies.
	There, participants could be randomized up to six times per day to receive no prompt, a prompt recommending a low-effort self-regulatory strategy, or a prompt recommending a higher-effort strategy. 
	This type of high-frequency randomization yields the intensive longitudinal data needed to develop empirically grounded JITAIs.
	
	Despite this promise, much of the existing literature analyzing MRT data has focused on proximal outcomes measured shortly after each decision point, such as short-term engagement in a target behavior or immediate response to treatment \citep{boruvka_assessing_2018, qian_estimating_2021}.
	However, in many behavioral health applications, outcomes measured farther in the future are also of substantive interest---particularly when the relevant processes involve habit formation or delayed treatment effects. 
	In such settings, proximal responses may not fully capture long-term success, even when they remain important intermediate targets.

	More broadly, using MRT data to construct JITAIs is a special case of learning an optimal treatment regime from longitudinal data.
	Classical work in longitudinal causal inference formalized counterfactual outcomes under time-varying interventions and established identification strategies for regime-specific mean outcomes from observed data \citep{robins_new_1986,robins_marginal_2000,murphy_marginal_2001}.
	Subsequent developments based on marginal structural models (MSMs) and related dynamic regime formulations extended these ideas to estimation and comparison of treatment rules over classes of candidate regimes \citep{joffe_model_2004, orellana_dynamic_2010}.
	These approaches provide a natural starting point for constructing JITAIs from MRT data, but the large number of decision points in many MRTs creates methodological challenges that are not well handled by standard longitudinal methods: inverse-probability weights can become highly variable, estimators may become irregular, and finite-sample bias can be substantial \citep{benkeser_doubly_2017}. 
	
	Beyond MSMs and related approaches, the policy-learning literature seeks to directly estimate an optimal rule rather than first estimating the expected outcome under each candidate rule, including outcome- and residual-weighted learning \citep{zhao_reinforcement_2009,zhou_residual_2017} and observational policy-learning methods \citep{athey_policy_2021}.
	While these approaches are appealing alternatives to the previous estimation-reliant methods, they were not developed for intensive longitudinal settings in which individuals repeatedly face treatment decisions over time.
	Reinforcement-learning methods are explicitly designed for sequential decision problems over long horizons, but they are commonly formulated to optimize cumulative or discounted rewards constructed from proximal outcomes observed after each decision point \citep{ertefaie_constructing_2018, luckett_estimating_2020}. 
	Such objectives are not well suited to our setting for two important reasons. 
	First, from scientific and clinical perspectives, distal outcomes, such as substance use at six months, are often more relevant than proximal outcomes, such as next-day substance use. 
	Second, choosing an appropriate discount rate, $\gamma$, which governs the relative importance of near-term versus long-term outcomes, can be challenging in mobile health settings because of limited empirical or practical guidance for this choice \citep{liao_off-policy_2021}.
	
	Motivated by these limitations, we propose new methods to construct an optimal JITAI with respect to a clinically meaningful distal outcome, without the need to pre-specify candidate decision rules.
	We focus on a class of stochastic policies parameterized by a finite-dimensional vector. 
	From a practical perspective, stochastic policies provide a useful alternative to dichotomized deterministic rules by allowing treatment delivery to remain flexible and by permitting decision-makers to incorporate additional considerations not explicitly modeled, such as side effects or insurance constraints.
	In behavioral settings, stochastic policies may allow occasional intervention in states where a deterministic rule would not intervene, helping to sustain participant engagement, reinforce healthy behaviors, and prevent disengagement over time.
	Moreover, stochastic policies help mitigate finite-sample positivity problems that arise when many decision points make deterministic or near-deterministic treatment paths poorly supported in the observed data \citep{van_der_laan_causal_2007,munoz_population_2012,haneuse_estimation_2013, kennedy_nonparametric_2019, mcclean_longitudinal_2025}. 
	However, with many decision points, they may still lead to unstable estimation. 
	To address this, we introduce a tilting approach that forms a weighted combination of the estimated optimal stochastic policy and the behavioral policy, that is, the policy generating the data. 
	
	This paper makes four main contributions.
	First, we develop a nonparametrically efficient estimator for the regimen--response curve, which characterizes the expected distal outcome under a class of stochastic policies.
	Our estimator is constructed to avoid estimating the large collection of augmentation terms that would otherwise arise in efficient estimation, thereby extending \cite{ertefaie_nonparametric_ipw_2023} to time-varying settings with stochastic policy learning. 
	This construction also avoids explicit derivation of the efficient influence function and mitigates the irregularity issues that can arise with doubly robust estimators \citep{benkeser_doubly_2017}.
	Second, we introduce a data-adaptive tilting procedure that addresses finite-sample instability by balancing bias and variance, and we establish conditions under which tilting is asymptotically negligible.
	Third, we show that the estimated regimen-response curve converges weakly to a Gaussian process, which enables the creation of simultaneous confidence bands for the curve.
	Finally, we derive the asymptotic properties of the optimizer of the regimen-response curve, corresponding to the optimal stochastic policy. 
	
	\section{Notation and Preliminaries}
	
	\subsection{Observed Data and Potential Outcomes}
	
	Suppose that we observe $n$ i.i.d.\@ observations $O=(S_{1},A_{1},\ldots,S_{T},A_{T},Y)\sim \P_0\in \mathcal{M}$, where $\mathcal{M}$ is some nonparametric model and $T$ denotes the number of decision points.
	At time $t \in \{1,\ldots,T\}$, $S_t\in \mathbb{R}^{d}$ denotes time-varying covariates, $A_t\in\{0,1\}$ is a binary treatment indicator, and $Y$ is the distal (end-of-study) outcome.
	We use overbar notation to denote complete histories: $\bar{S}_t\coloneq (S_1,\ldots,S_t)$ and $\bar{A}_t\coloneq(A_1,\ldots,A_t)$; additionally, let $\bar{A}\coloneq\bar{A}_T$ and $\bar{S}\coloneq\bar{S}_T$ be the full covariate and treatment histories.
	Let $ H_t\coloneq(\bar{S}_t,\bar{A}_{t-1})$ denote all information observed prior to assigning $A_t$.
	
	Let $Y^{\bar{a}}$ denote the potential outcome that would have been observed if, possibly contrary to reality, treatment were assigned as $\bar{a}$.
	Similarly, let $S_t^{\bar{a}_{t-1}}$ represent the covariates at time $t$ that would have been observed if treatment was assigned as $\bar{a}_{t-1}$ up until time $t-1$.
	Implicit in this notation, we assume that these covariates are not affected by future treatment decisions $a_t,\ldots,a_T$; i.e., $S_t^{\bar{a}_{t-1}}=S_t^{\bar{a}_T}$ for any $\bar{a}_T\in\{0,1\}^T$.
	We define full data as 
	$X \coloneq (S_1,\,\{S_t^{\bar a_{t-1}}: \bar a_{t-1}\in\{0,1\}^{t-1},\ t=2,\ldots,T\},\,
	\{Y^{\bar a}: \bar a\in\{0,1\}^T\}) \sim \P_X \in \mathcal{M}^F$ 
	where $\mathcal{M}^F$ is the nonparametric model for the full data $X$. 
	
	Denote the empirical distribution function as $\P_n$, and let $\P f = \int f(o)\dd\P(o)$ for any measurable function $f$ and distribution $\P$. 
	For any random variable, $X$, we let the calligraphic $\mathcal X$ denote its support.

	\subsection{Stochastic Adaptive Interventions and Our Causal Estimand}
	
	A stochastic adaptive intervention assigns $A_t$, at each time $t$, by drawing from a multinomial distribution in which the probabilities are functions of the current history $ H_t$.
	To operationalize policy optimization and avoid the curse of dimensionality, we focus on policies that only depend on the current state $S_t$.  
	Henceforth, we consider the parametric class $\{q_t^\theta:\theta\in\Theta\}$ for some compact index set $\Theta\subset \mathbb{R}^p$, in which, for $a\in\{0,1\}$
	\begin{equation*}
		\label{equation:q-theta}
		q_t^\theta(a,s)\coloneq 1 - \lvert \expit\{b(s)\trans\theta\}- a\rvert\mathcomma
	\end{equation*}
	for some fixed basis function $b: \mathbb{R}^d \to \mathbb{R}^p$.
	This gives the familiar logistic model: $q_t^\theta(1,s)=\expit\{b(s)\trans\theta\}$ and $q_t^\theta(0,s)=1-\expit\{b(s)\trans\theta\}$.
	Other smooth parametric forms are also admissible provided that they satisfy suitable regularity conditions.
	We restrict attention to genuinely stochastic policies whose treatment probabilities are bounded away from $0$ and $1$; \cref{sec:index-set} provides a practical construction of $\Theta$ from investigator-specified probability bounds.
	
	For each \(\theta\in\Theta\), let \(Y^\theta\) denote the potential distal outcome that would have been observed had the policy \(\{q_t^\theta\}_{t=1}^T\) been followed throughout. Our primary causal estimand is the regimen-response curve
	\[
	\psi_0(\theta)=\mathbb P_0(Y^\theta),
	\qquad \theta\in\Theta.
	\]
	This curve has two roles. 
	For a fixed \(\theta\), it gives the mean distal outcome under the corresponding stochastic intervention. 
	As a function of \(\theta\), it provides the basis for policy learning: an optimal stochastic policy can be defined as any maximizer of \(\psi_0(\theta)\) over \(\Theta\), or of a penalized version of this criterion when additional regularization is imposed. Thus, policy learning is reduced to estimation and maximization of a smooth finite-dimensional regimen-response curve. 
	Identification of \(\psi_0(\theta)\) from the observed data relies on the standard longitudinal causal assumptions stated next.
	
	\begin{assumption}
		\label{assumption:identification}
		Suppose that the following causal assumptions hold:
		\begin{enumerate}[label=(\alph*),ref=\theassumption(\alph*)]
			\item Consistency: For all $t = 2,\ldots,T$, if $\bar{A}_{t-1}=\bar{a}_{t-1}$, then $S_t=S_t^{\bar{a}_{t-1}}$; if $\bar{A}=\bar{a}$, then $Y=Y^{\bar{a}}$.    
			\item Ignorability/Sequential Randomization: $\{Y^{\bar{a}}, S_k^{\bar{a}_{k-1}} : k = t,\ldots,T\}  \ind A_t \mid  H_t$ for all $t=1,\ldots,T$.
			\item \label{assumption:identification:positivity} Strong Positivity: There exists some $\epsilon_\pi>0$ such that $\pi_{t,0}(A_t =  a \mid H_t) > \epsilon_\pi$ almost surely for all $a\in\{0,1\}$ and $t\in\{1,\ldots,T\}$.
		\end{enumerate}
	\end{assumption}
	
	\cref{assumption:identification} consists of standard conditions required to identify counterfactual outcomes under longitudinal interventions \citep{robins_new_1986,robins_marginal_2000,murphy_marginal_2001,orellana_dynamic_2010}. 
	Consistency links the observed data to the corresponding potential outcomes under the realized treatment history.
	This assumption may be violated, even in MRT settings, if one individual's treatment assignments affect the outcomes of another individual.
	Sequential randomization asserts that, conditional on the observed history $H_t$, treatment assignment at $A_t$ is effectively randomized and does not depend on future potential outcomes; this is satisfied by design in MRT settings.
	Positivity ensures that each treatment option has a nonzero probability at every decision point and guarantees absolute continuity of the counterfactual law under the target policy with respect to the observed data law.
	We note that only weak positivity is required for identification; however, later theoretical results will leverage strong positivity.
	Although strong positivity holds by design in MRTs, the large number of decision points can still induce finite-sample instability in inverse probability weights, motivating the use of stochastic interventions and the policy-tilting approach introduced in \cref{sec:tilting}.
	We note that protocols that deterministically restrict treatment at some decision points are handled in \cref{sec:ineligible}.
	Together, these assumptions ensure that the mean counterfactual outcome under any stochastic policy is nonparametrically identifiable using a change-of-measure argument.
	
	\section{Proposed Approach}
	
	We first introduce an $\alpha$-tilting device that stabilizes longitudinal products of inverse probability weights by shrinking early decision rules toward the observed treatment mechanism (\cref{sec:tilting}).
	We then formalize identification and the semiparametric structure of the resulting (tilted) regimen-response curve, highlighting the key nuisance quantities that arise in efficient estimation (\cref{sec:ipw-mapping}).
	Building on this structure, we propose an undersmoothed inverse probability weighted estimator that attains nonparametric efficiency while avoiding explicit estimation of the large collection of time- and policy-indexed outcome regressions (\cref{sec:uipw}).
	We conclude with practical implementation components: data-adaptive selection of the tilting level $\alpha_n$, investigator-driven elicitation of the feasible policy set $\Theta$, and a modification to accommodate protocol-driven randomization ineligibility (\cref{sec:adaptive-tilting,sec:index-set,sec:ineligible}).
	
	\subsection{Tilted Policies}
	\label{sec:tilting}

	Common inverse probability weighted (IPW) or doubly robust estimators of the regimen-response curve typically involve products of time-specific weights: $\prod_{t=1}^T q_t^{\theta,\alpha}/\pi_{t,n}$.
	When $T$ is large and the policy of interest substantially differs from the observed treatment mechanism, these products can be numerically unstable, resulting in bias and large variances in finite samples even when strong positivity holds by design.
	To mitigate the effect of such finite-sample positivity violations, we propose a policy-tilting device that shrinks earlier weights towards $1$ by interpolating between the target policy and the treatment assignment process observed in the data.
	For $\alpha\in[0,\infty)$, define the $\alpha$-tilted policy
	\begin{equation*}
		\label{equation:q-theta-alpha}
		q_t^{\theta,\alpha}(A_t, H_t) \coloneq (t/T)^\alpha q_t^\theta(A_t,S_t) + \{1-(t/T)^\alpha\}\pi_{t,0}(A_t\mid H_t)\mathperiod
	\end{equation*}
	For $\alpha=0$, the tilted policy reduces to the original policy of interest, $q_t^{\theta,0}=q_t^\theta$, and the corresponding estimand coincides with the scientific target.
	For $\alpha>0$, tilting defines a distinct stochastic intervention and hence a different counterfactual mean; in \cref{sec:adaptive-tilting} we select $\alpha$ to balance this estimand drift against the finite-sample variance reduction gained by stabilizing weights, and later results show that sufficiently small $\alpha$ yields asymptotically negligible drift.
	We note that although the scientific policy $q_t^\theta$ is restricted to depend only on the current state $S_t$, the tilted rule $q_t^{\theta,\alpha}$ depends on the full history $H_t$ through the observed assignment mechanism $\pi_{t,0}(\cdot\mid H_t)$.
	
	To see the stabilization effect, we can write the time-$t$ weight contribution for a generic observation as 
	$$W_{t,0}^{\theta,\alpha}(O)\coloneq \frac{q_t^{\theta,\alpha}(A_t, H_t)}{\pi_{t,0}(A_t\mid H_t)} = (t/T)^\alpha\frac{q_t^\theta(A_t,S_t)}{\pi_{t,0}(A_t\mid H_t)} + \{1-(t/T)^\alpha\}\mathcomma$$
	which is a convex combination of the usual IPW factor $q_t^\theta/\pi_{t,0}$ and $1$.
	Because $(t/T)^\alpha\to0$ as $\alpha\to\infty$ for each $t<T$, we have $W_{t,0}^{\theta,\alpha}(O)\to1$ for all early decision points $t<T$.
	Hence, the product weight $\prod_{t=1}^T W_{t,0}^{\theta,\alpha}(O)$ shrinks toward a product with fewer volatile factors. 
	
	Let $Y^{\theta,\alpha}$ denote the potential outcome under regimen $\{q_t^{\theta,\alpha}\}_{t=1}^T$ and 
	denote the $\alpha$-tilted regimen-response curve by the mapping $(\theta,\alpha)\mapsto \P_0(Y^{\theta,\alpha})$.
	By the previous logic, we have $\P_0(Y^{\theta,\alpha=0})=\P_0(Y^{\theta})$.
	
	\begin{remark}
		In the motivating context of MRTs, strong positivity holds by design, and the un-tilted estimand $\P_0(Y^{\theta,\alpha=0})$ is well defined and identified.
		We introduce $\alpha>0$ solely as a finite-sample stabilization device to mitigate numerical instability in products of inverse probability weights when the number of decision points is large, and our asymptotic results explicitly consider behavior as $\alpha\to0$.
		Importantly, policy tilting does not repair violations of strong positivity: if a target policy assigns a positive probability to treatment values outside the support of the observed treatment mechanism, the corresponding counterfactual mean remains unidentified for any finite $\alpha$.
		Formally, taking $\alpha=\infty$ collapses the policy to the observed treatment mechanism for all $t<T$, thereby avoiding support violations at earlier decision points.
		Even in this case, strong positivity is required at $T$.
	\end{remark}
	
	\subsection{Inverse Probability Weighting Mapping}
	\label{sec:ipw-mapping}
	
	Under the full data $X$, the mean counterfactual outcome under the tilted policy $q_t^{\theta,\alpha}$ is defined as $\Psi^F_{\theta,\alpha}(P_X) = P_X (Y^{\theta,\alpha})$ where $\Psi^F: \mathcal{M}^F \rightarrow \mathbb{R}$. 
	For a fixed policy parameter $\theta\in\Theta$ and tilting parameter $\alpha\ge 0$, define  the corresponding full data canonical gradient $D^F_{\theta,\alpha}(X,\Psi^F) = Y^{\theta,\alpha} - \Psi^F_{\theta,\alpha}(P_X)$.  
	A longitudinal inverse probability weighting mapping of $D^F_{\theta,\alpha}(X,\Psi^F)$ to the observed data is given by $U(O;\theta,\alpha,\Psi)=W^{\theta,\alpha}(P)(O)Y-\Psi(P)$ where
	\begin{equation*}
		W^{\theta,\alpha}(P)(O) = \prod_{t=1}^T
		\frac{(t/T)^\alpha q_t^\theta(A_t,S_t)+\{1-(t/T)^\alpha\}P(A_t\mid H_t)}
		{P(A_t\mid H_t)}\mathperiod    
	\end{equation*}
	Under  Assumption \ref{assumption:identification}, the target parameter $\Psi^F_{\theta,\alpha}(\P_X)$ is identified by 
	$$\Psi_{\theta,\alpha}(P)\coloneq P\{W^{\theta,\alpha}(P)(O)Y\}\mathcomma$$ 
	and $P\{U(O;\theta,\alpha,\Psi) \mid X \}=  D^F_{\theta,\alpha}(X,\Psi^F)$. 
	Let $\mathcal T_{\mathrm{CAR}}$ denote the nuisance tangent space corresponding to the sequential treatment mechanism. 
	Because the treatment mechanism factorizes over time, this space decomposes as $\mathcal T_{\mathrm{CAR}} = \mathcal T_{\mathrm{CAR},1} \oplus \cdots \oplus \mathcal T_{\mathrm{CAR},T}$ where $T_{\mathrm{CAR},t}$ is the nuisance tangent space generated by the treatment assignment mechanism at time $t$, for $t=1,2,\ldots,T$.
	The canonical gradient of $\Psi(P)$ is given by $D^\star(P)(O;\theta,\alpha)=U_\pi(O;\theta,\alpha,\psi_0)-D_{\mathrm{CAR}}(P)$,
	where $D_{\mathrm{CAR}}(P)$ is the projection of $U_\pi$ onto $\mathcal T_{\mathrm{CAR}}$. 
	Owing to the above decomposition,
	$D_{\mathrm{CAR}}(P)=\sum_{t=1}^T D_{\mathrm{CAR},t}(O;\eta_0,\theta,\alpha)$ for nuisance parameters $\eta_0$. 
	As shown in \cref{lemma:dcar-expectation}, the  nuisance parameter $\eta_0$ includes the regression functions 
	$\P_0\{(\prod_{j>t}W^{\theta,\alpha}_{j,0})Y \mid  A_t,H_t\}$.
	This representation yields the efficient influence function for the regimen-response curve in the longitudinal stochastic intervention model. 
	
	In observational settings under a nonparametric model, both the propensity scores and the regression functions generally must be estimated. 
	When these nuisance functions are estimated nonparametrically, inconsistency in any of them can lead to an irregular estimator with substantial bias. 
	This difficulty becomes especially severe when the number of decision points is large \citep{benkeser_doubly_2017,pham_nonparametric_2025}. 
	This is particularly challenging for the regression functions, which must be estimated for each $\theta\in\Theta$ whenever $t<T$.
	In MRTs, the propensity scores are known by design, so the irregularity issue does not arise. 
	However, inconsistent estimation of the regression functions can still lead to inefficiency.
	This motivates our proposed undersmoothed inverse probability weighted estimator, which only requires estimation of $T$ propensity scores but still is $\sqrt{n}$-consistent and nonparametric efficient.
	
	\subsection{Undersmoothed Inverse Probability Weighting Estimator}
	\label{sec:uipw}
	
	The efficient influence function characterization in \cref{sec:ipw-mapping} suggests that nonparametric efficient estimation of the regimen-response curve generally involves a large family of nuisance functions indexed by both time and policy, $\{\mu_{t,0}^{\theta,\alpha}: t<T, \theta\in\Theta\}$, which becomes impractical when $T$ is large or when $\Theta$ is not discrete.  
	We therefore construct an undersmoothed inverse probability weighted (UIPW) estimator that achieves nonparametric efficiency without explicitly estimating these policy-indexed outcome regressions, thereby offering a practically feasible approach for intensive longitudinal settings.
	The key idea is to estimate  the time-specific propensity scores flexibly via the highly adaptive lasso (HAL; \citealp{benkeser_highly_2016}) with deliberate undersmoothing so that the resulting fits approximately solve a rich collection of empirical score equations.
	This renders the bias asymptotically negligible while preserving the simplicity of a weighted-mean estimator.
	We proceed by summarizing the estimator and then describe an implementable procedure to select the undersmoothing level.
	
	The HAL is a nonparametric regression technique that estimates functions via a sparse linear combination of indicator basis functions indexed by observed covariate values.
	The estimator minimizes an empirical risk while constraining the $L_1$-norm of its coefficients by $\lambda$, which controls the total sectional variational norm of the fitted function.
	Under mild conditions, HAL estimators achieve fast convergence rates (roughly $n^{-1/3}$) for a wide class of c\`adl\`ag functions without requiring smoothness assumptions.
	Full technical details, including the precise definition of the basis and norm, are provided in \cref{sec:hal-details}.
	
	When the regimen response curve is estimated using IPW, the asymptotic expansion admits a remainder of the form
	\begin{equation}
		\label{eqn:ipw-remainder-generic}
		\sum_{t=1}^T \P_n\big[f_{t}(O;\theta,\alpha,\bar\pi_{t,n},\mu_{t,0}^{\theta,\alpha})\{A_t-\pi_{t,n}(1\mid H_t)\}\big]\mathcomma  
	\end{equation}
	where the function $f_{t}$ involves the estimated weights from previous decision points $i<t$ and the true outcome regression function looking forward following a particular stochastic intervention at $j\ge t$.
	Each term in the sum resembles $\P_n \Dcart$ with the estimated propensity score plugged in place of the true propensity score for all $s\le t$. 
	If the treatment model is estimated using a correctly specified parametric model, this term is $O_p(n^{-1/2})$ and is absorbed into the estimator's influence function under the semiparametric model.
	However, if nonparametric methods are used, this term is generally not even $O_p(n^{-1/2})$, preventing valid inference.
	
	\cite{ertefaie_nonparametric_ipw_2023} showed that undersmoothing the HAL---selecting a sufficiently large $L_1$ bound that is typically greater than the cross-validated bound---ensures that the fitted model is rich enough to approximately solve these empirical score equations, rendering the term asymptotically negligible.
	Intuitively, this occurs because the HAL approximately solves score functions of the form $\P_n[\varphi(O)\{A_{t,n}-\pi_{t,n}(1\mid H_t)\}]$ for the included basis functions $\varphi$.
	If the model is sufficiently undersmoothed, a rich enough collection of basis functions is retained.
	Hence, if $f_{t,n}$ can be approximated by these basis functions within an $n^{-1/4}$ neighborhood, the HAL will approximately solve the bias term, rendering it asymptotically negligible. 
	\cref{sec:undersmooth} contains full details regarding undersmoothing and formal statements of the previous conditions.
	
	While the undersmoothing condition provides a theoretical target for eliminating the leading remainder term, it leaves open the practical choice of tuning parameters, which must balance bias control against weight instability.
	Selecting too large of a penalty can induce oversmoothing and induce bias, whereas too small of a penalty can inflate finite-sample variance, cause numeric instability, and prevent convergence at the $n^{-1/4}$ rate.
	
	One approach to selection is to directly target a cross-validated estimate of the leading IPW remainder (\cref{eqn:ipw-remainder-generic}). 
	However, this would require estimating the large set of nuisance outcome regression functions $\{\mu_{t,0}^{\theta,\alpha}:t<T,\theta\in\Theta\}$, and doing so fully nonparametrically is generally impractical.  
	To retain this remainder-targeting motivation while avoiding this computational burden, we consider a pragmatic 
	selector that plugs in a low-dimensional working approximation, $\mu_{t,n}^{\theta,\alpha}$.
	The bias is then approximated by 
	\begin{equation*} 
		\label{eqn:dcar-bias-cv-approx}
		R_{t,n,v}(\lambda;\theta)\coloneq \P_{n,v}^1  f_t(O;\theta,\alpha,\bar\pi^{(v)}_{t-1,n,\lambda_{t,n,\text{CV}}},\mu_{t,n}^{(v)},\pi_{t,n,\lambda_{t,v,\text{CV}}}^{(v)}) \{A_t-\pi_{t,n,\lambda}^{(v)}(A_t\mid H_t)\}\mathcomma    
	\end{equation*}
	where $v=1,\ldots,V$ indexes folds, $\P_{n,v}^1$ denotes the empirical measure on the testing data from the $v$th fold.
	We then consider two selectors that minimize different norms of this bias approximation:
	\begin{equation}
		\label{eqn:linfty-dcar-criterion-body}
		\lambda_{t,n,L_\infty}\coloneq \arg\min_{\lambda} V^{-1} \sum_{v=1}^V \max_{\theta\in\Theta_\text{test}} \lvert R_{t,n,v}(\lambda;\theta)\rvert\mathcomma
	\end{equation}
	and
	\begin{equation}
		\label{eqn:l2-dcar-criterion-body}
		\lambda_{t,n,L_2}\coloneq \arg\min_{\lambda} V^{-1} \lvert\Theta_\text{test}\rvert^{-1} \sum_{v=1}^V \sum_{\theta\in\Theta_\text{test}} R_{t,n,v} (\lambda;\theta)^2\mathcomma
	\end{equation}
	where $\Theta_\text{test}\subset\Theta$ is a holdout set of $\theta$ not used to train the working model $\mu_{t,n}$.
	We refer to both $\lambda_{t,n,L_\infty}$ and $\lambda_{t,n,L_2}$ as $\Dcar$-based selectors because they minimize the norm of a term that resembles $\Dcar$. 
	The $L_\infty$ criterion is motivated by a desire to uniformly control \cref{eqn:ipw-remainder-generic} over $\Theta$ to support curve-level inference and the construction of simultaneous confidence bands. 
	The $L_2$ criterion averages the remainder over $\Theta$ and is typically less sensitive to isolated, high-variance points, often yielding improved finite-sample stability. 
	
	Additionally, we consider a score-based selector that depends only on the propensity score.
	That is, we pick $\lambda$ by minimizing a cross-validated aggregated standardized score residual over the set of active HAL basis functions.
	For each $t$ and $v$, let 
	$\mathcal J_t^{(v)}(\lambda)$ denote the set of HAL basis functions retained (i.e., with nonzero basis coefficients) in the fitted $\pi_{t,n,\lambda}^{(v)}$.
	We select $\lambda$ by
	\begin{equation}
		\label{eqn:score-criterion-body}
		\lambda_{t,n,\text{score}} = \arg\min_{\lambda} \frac1V \sum_{v=1}^V \frac{1}{\lVert \beta^{(v)}_{t,n,\lambda}\rVert_1} \sum_{\varphi\in\mathcal J_t^{(v)}(\lambda)}\Bigg\lvert \P_{n,v}^{1} \Bigg\{\varphi(H_t) \frac{A_t-\pi_{t,n,\lambda}^{(v)}(1\mid H_t)}{\pi_{t,n,\lambda}^{(v)}(1\mid H_t)} \Bigg\} \Bigg\rvert\mathcomma
	\end{equation}
	where $\lVert \beta_{t,n,\lambda}\rVert_1$ is the fitted HAL coefficient $L_1$ norm.
	This criterion is a computationally tractable proxy to targeting the remainder term that avoids estimating $\mu_{t,n}^{\theta,\alpha}$.
	It leverages the fact that the propensity scores enter the remainder terms through inverse-weight factors; hence, minimizing aggregated standardized score residuals balances remainder control against weight instability.
	
	Finally, following \cite{pham_nonparametric_2025}, we additionally impose a mild cap on oversmoothing to control the rate at which the basis grows and ensure that the HAL converges at the $n^{-1/4}$ rate.
	Full details of the $\Dcar$ and score based undersmoothing selectors are provided in \cref{sec:undersmooth-practice}.
	
	Under this framework, let $\{\pi_{t,n}\}_{t=1}^T$ be the undersmoothed propensity score estimates and let 
	\begin{equation*}
		W^{\theta,\alpha}_n=\prod_{t=1}^T\frac{(t/T)^\alpha q_t^\theta(A_t,S_t)+\{1-(t/T)^\alpha\}\pi_{t,n}(A_t\mid H_t)}{\pi_{t,n}(A_t\mid H_t)}.    
	\end{equation*}
	The UIPW regimen-response curve estimator is then $\psi_n(\theta;\alpha) \coloneq \P_n\{ (\prod_{t=1}^T W^{\theta,\alpha}_{t,n})Y\}$.
	Let $\psi_n(\theta)\coloneq\psi_n(\theta;\alpha=0)$ denote the untilted estimator.
	
	Estimation of $\psi_0(\theta;\alpha)$ enables the construction of policies by searching over $\Theta$ to identify a $\theta$ that maximizes the expected distal outcome.
	However, we note that global  maximization over the unrestricted policy class $\mathbb R^p$ typically favors near-deterministic interventions that push treatment probabilities toward $0$ or $1$ in regions of the state space where treatment is beneficial or harmful.
	Consequently, under a bounded-probability feasible set $\Theta$, the optimizer can behave like a constrained projection: along many directions in the parameter space the estimated mean function is approximately monotone, so the maximizer is attained at the feasible point ``closest'' to the unconstrained maximizer, which generally lies on $\partial\Theta$.
	Because boundary solutions can be numerically unstable and complicate optimizer-based inference, we maximize a penalized value function \citep{gao_asymptotic_2025, luckett_estimating_2020}:
	\begin{equation*}
		V_n(\theta;\alpha) \coloneq \psi_n(\theta;\alpha)-\mathcal P(\theta)\mathperiod
	\end{equation*}
	In particular, we will focus on an inverse-distance barrier-weighted ridge penalty: 
	\begin{equation*}
		\mathcal P(\theta)=  \frac{\gamma}{\inf_{\theta' \in \partial\Theta} d(\theta,\theta')}\|\tilde\theta\|_2^2\mathcomma
	\end{equation*}
	for some fixed ridge penalty $\gamma>0$ and distance function $d$.
	Here, $\theta$ is partitioned as $(\theta_0,\tilde\theta)$ when $b(s)$ includes an intercept in the first coordinate (and $\tilde\theta=\theta$ otherwise) so as to avoid penalizing the intercept, and $\partial\Theta$ denotes the boundary of $\Theta$.
	This penalty diverges to $+\infty$ as $\theta$ approaches $\partial\Theta$, so the penalized value function assigns value $-\infty$ on the boundary, thereby ensuring an interior maximizer.
	We note that this penalty works in parallel with the index set $\Theta$: careful elicitation of $\Theta$ ensures a truly stochastic set of admissible policies, $\mathcal P$ prevents a boundary maximizer within this set, and supports inference.
	
	The corresponding population value function is $V_0(\theta;\alpha)\coloneq \psi_0(\theta;\alpha)-\mathcal P(\theta)$.
	We denote the optimal policy by $\theta_0^\star(\alpha)\coloneq\argmax_{\theta\in\Theta} V_0(\theta;\alpha)$, and consider the estimator $\theta_n^\star(\alpha)\coloneq\argmax_{\theta\in\Theta}V_n(\theta;\alpha)$.
	Let $\theta_0^\star=\theta_0^\star(\alpha=0)$ and $\theta_n^\star=\theta_n^\star(\alpha=0)$ denote the untilted optimizers. 
	
	\subsection{Data-Adaptive Tilting}
	\label{sec:adaptive-tilting}
	
	The optimal amount of tilting will depend on a variety of conditions including the sample size, number of decisions, and distance from the optimal policy to the observed policy.
	If $T$ is small, $n$ is large, and the optimal policy $q_t^{\theta_0^\star(0)}(A_t,S_t)$ approximates the observational policy $\pi_{t,0}(A_t\mid H_t)$, the weights will be relatively stable near $\theta_0^\star(0)$ without any tilting.
	Instead, if the optimal policy is far from the observed policy, the untilted weights may be volatile at $\theta_0^\star(0)$.

	As $\alpha$ varies, tilting trades off two competing goals.
	For $\alpha$ near zero, the tilted policy remains close to the untilted scientific intervention, so the corresponding target value $V_0(\cdot;\alpha)$ is expected to be close to $V_0(\cdot;0)$.
	This limits estimand drift and helps keep the learned optimizer $\theta_n^\star(\alpha)$ close to the policy that is optimal for the untilted target.
	On the other hand, increasing $\alpha$ shrinks early weight factors toward one, reducing the variability and numerical instability induced by longitudinal products of inverse probability weights and thereby stabilizes $\psi_n(\cdot;\alpha)$ and, in turn, $V_n(\cdot;\alpha)$.
	
	We operationalize this tradeoff by selecting $\alpha$ to minimize the corresponding mean squared error (MSE) of the value function at the optimal policy. By setting the derivative of the asymptotic MSE to zero, we obtain
	\[2\P_0\bigl\{\bar V_n(\alpha) - \bar V_0(\alpha=0)\bigr\} \P_0\left\{ \pdv\alpha \bar V_n(\alpha)\right\} + 2n^{-1}\bar\sigma_0(\alpha)\pdv\alpha \bar\sigma_0(\alpha)=0\mathcomma  \]
	where $\bar V_\cdot(\alpha)\coloneq V_\cdot\{\theta_\cdot^\star(\alpha);\alpha\}$, is the $\theta$-profiled value function for $\cdot\in\{0,n\}$, and $\bar\sigma^2_0(\alpha)\coloneq \P_0[\phi_\psi\{O;\eta_0, \theta_0^\star(\alpha),\alpha\}^2]$.
	We know that the optimal (up to a constant) bias-variance trade-off implies  $\P_0\{\bar V_n(\alpha) - \bar V_0(\alpha=0)\} \approx n^{-1/2}\bar\sigma_0(\alpha)$ \citep{van_der_laan_cv-tmle_2018}. 
	Hence, the optimal $\alpha$ will also solve one of the following equations
	$\P_0\left\{\pdv\alpha \bar V_n(\alpha)\right\}\pm c n^{-1/2} \pdv\alpha  \bar\sigma_0(\alpha) =0$, 
	where $c$ is a positive constant. 
	We approximate these expressions via $\pdv\alpha \bar V_n(\alpha)\pm c n^{-1/2}\pdv\alpha  \bar\sigma_n(\alpha)$ where $\bar\sigma_n(\alpha)$ is a consistent estimator for $\bar\sigma_0(\alpha)$.
	The choice of which solution yields the minimizer depends on the sign of $\pdv\alpha \bar V_n(\alpha)$; because $\bar V_0(\alpha)$ is generally decreasing in $\alpha$, 
	we propose selecting
	\begin{equation}
		\label{eqn:alpha-n-def}
		\alpha_n \coloneq \argmax_{\alpha\ge0} \bar{V}_n(\alpha)-\kappa n^{-1/2}\bar\sigma_n(\alpha)\mathcomma
	\end{equation}
	for some fixed $\kappa>0$. 
	This is equivalent to selecting the value at which $\bar V_n(\alpha)$ attains the highest lower confidence interval bound at level $\Phi(\kappa)-\Phi(-\kappa)$ where $\Phi$ is the standard normal distribution function.
	In practice, the lower-bound criterion can occasionally return very small $\alpha_n$ corresponding to negligible stabilization relative to $\alpha=0$.
	To avoid introducing an unnecessarily tilted estimand in such cases---and to align with asymptotic comparison to the untilted target developed later---we apply a simple clipping rule that sets $\alpha_n$ to zero whenever it falls below a vanishing threshold:
	\begin{equation*}
		\label{eqn:alpha-n-clipped}
		\tilde\alpha_n = \begin{cases}0 & \alpha_n < n^{-1/4} \\ \alpha_n & \alpha_n\ge n^{-1/4}\end{cases}\mathperiod
	\end{equation*}
	Specifically, in \cref{sec:theory} we show that, under the condition that the population profiled value $\bar V_0(\alpha)$ is maximized at $\alpha=0$, using $\tilde\alpha_n$ yields the asymptotic limit as the untilted estimator.
	
	\begin{remark}
		\label{remark:v-bias}
		The lower-bound rule $\bar V_n(\alpha)-\kappa n^{-1/2}\bar\sigma_n(\alpha)$ is designed to select a tilt level that yields a high and stable estimated value.
		In asymptotic results developed later, we adopt the working assumption that the population profiled value $\bar{V}_0(\alpha)\coloneq V_0\{\theta_0^\star(\alpha);\alpha\}$ is maximized at $\alpha=0$.
		Under this condition, $\bar V_n(\alpha)$ can be viewed heuristically as a proxy for the bias incurred by tilting away from the untilted estimand.
		This condition is not automatic and can fail for two conceptually different reasons:
		\begin{enumerate}
			\item \textbf{Limited policy class.} If the observational policy is near-optimal but the class $\{q^\theta:\theta\in\Theta\}$ cannot approximate it well---either because $\Theta$ is too restrictive or the basis does not include enough covariates---then increasing $\alpha$ moves the estimand towards a policy that may be better than anything attainable at $\alpha=0$ within the class.
			\item \textbf{Tilt-penalty interaction.} Even when the best untilted policy lies in the class, the assumption may fail due to an interaction between tilting and penalization. 
			As $\alpha$ increases, $\psi_0(\theta;\alpha)$ becomes less sensitive to $\theta$ because the observed component, invariant to $\theta$, receives more weight. 
			As a result, many $\theta$ can yield similar tilted outcomes while differing substantially in penalty.
			The optimizer may then move toward smaller-normed coefficients to reduce the penalty, thereby increasing $\bar V_0(\alpha)$ even though there is no improvement in $\psi_0\{\theta_0^\star(\alpha);\alpha\}$.
		\end{enumerate}
		When either phenomenon occurs, $\bar V_n(\alpha)$ should be interpreted primarily as an empirical criterion for stable policy learning rather than as a monotone proxy for drift relative to $\alpha=0$.
		In MRT settings with constant treatment probabilities, these issues are largely mitigated when the policy class includes an intercept and the penalty excludes it. 
		There, the observational policy is representable by an intercept-only policy and incurs no penalty.
	\end{remark}

	\section{Theoretical Properties}
	\label{sec:theory}
	
	We proceed from efficiency characterization to inference for both the regimen-response curve and the induced optimizer.
	We first establish identification of the tilted counterfactual mean as an observed-data functional expressed through longitudinal inverse probability weights (\cref{lemma:identification-expectation}).
	We then derive the canonical gradient of $\psi_0(\theta;\alpha)$ under the nonparametric model (\cref{lemma:dcar-expectation}).
	This canonical gradient implies the nonparametric efficiency bound and the nuisance structure associated with efficient estimation.
	
	Building on this characterization, we show that the proposed UIPW estimator is regular and asymptotically linear, and we establish both pointwise and uniform (over $\Theta$) limiting results for the estimated regimen-response curve (\cref{theorem:uipw-asym}).
	We then show that, for appropriately vanishing tilting sequence $\alpha_n$, the resulting estimator is asymptotically equivalent to the untilted estimator targeting $\psi_0(\theta;0)$ (\cref{corollary:uipw-asymp-alpha}).
	We then present analogous results for the learned optimizer $\theta_n^\star(\alpha)$ and its counterpart under vanishing tilting (\cref{theorem:opt-theta-asym} and \cref{corollary:opt-theta-asym-alpha}).
	Finally, we study the behavior of the data-adaptive tilting selector and provide sufficient conditions for it to attain the requisite rate to support the previous corollaries (\cref{lemma:alpha-clip-op}).
	
	\subsection{Identification and Canonical Gradient}
	
	This subsection formalizes the observed-data parameter underlying our proposed estimator and characterizes its nonparametric efficiency bound. 
	Under \cref{assumption:identification}, the tilted counterfactual mean $\P_0(Y^{\theta,\alpha})$ is identifiable from the observed data law. 
	While identification can be expressed via either outcome regression or inverse probability weighting, we focus on an IPW change-of-measure representation because it yields a convenient parameter mapping for deriving the canonical gradient and motivates the efficient IPW construction used throughout. 
	
	\begin{lemma}
		\label{lemma:identification-expectation}
		Under \cref{assumption:identification} the causal parameter $\P_0(Y^{\theta,\alpha})$ is identified by
		$\psi_{0}(\theta;\alpha)\coloneq\P_0(W_0^{\theta,\alpha}Y)$, where $W_0^{\theta,\alpha}=\prod_{t=1}^T W_{t,0}^{\theta,\alpha}$ is the product of time-specific weights of the form 
		$$W_{t,0}^{\theta,\alpha}=\frac{(t/T)^\alpha q_t^\theta(A_t, H_t)+\{1-(t/T)^\alpha\}\pi_{t,0}(A_t\mid H_t)}{\pi_{t,0}(A_t\mid H_t)}
		\mathperiod$$
	\end{lemma}

	From this identification, we can obtain the canonical gradient, which characterizes the nonparametric efficiency bound for estimating $\psi_0(\theta;\alpha)$.
	Let 
	$$\mu_{t,0}^{\theta,\alpha}(A_t,H_t)\coloneq \P_0\Bigl[ \Bigl\{\prod_{j>t}W_{j,0}^{\theta,\alpha}(A_j,S_j)\Bigr\}Y\mid A_t,H_t\Bigr]$$ 
	be the conditional expected outcome if an individual with covariate history $\{A_t,H_t\}$ were to follow regimen $q_j^{\theta,\alpha}$ at decision points $j\in\{t+1,\ldots,T\}$.
	These regressions admit a standard backward recursion; we record it here for completeness:
	$\mu_T^{\theta,\alpha}(A_T,H_T)=\P_0(Y|A_T,H_T)$, and
	\begin{align*}
		\mu_t^{\theta,\alpha}(A_t,H_t) = \P_0\left\{\sum_{a\in\{0,1\}} q^{\theta,\alpha}_{t+1}(A_{t+1}=a, S_{t+1})\mu^{\theta,\alpha}_{t+1}(A_{t+1}=a,H_{t+1}) \mid A_t,H_t \right\}\mathcomma
	\end{align*}
	for $t=1,\ldots,T-1$.
	Let $\eta_0\coloneq\{\pi_{1,0},\ldots,\pi_{T,0},\mu_{1,0}^{\theta,\alpha},\ldots,\mu_{T,0}^{\theta,\alpha}\}$ denote the set of nuisance functions corresponding to both sequential treatment assignment and outcome regressions.
	
	\begin{lemma}
		\label{lemma:dcar-expectation}
		The canonical gradient of $\psi_{0}(\theta;\alpha)$ is 
		$$\phi_\psi(O;\eta_0,\theta,\alpha)= W_0^{\theta,\alpha}Y - \sum_{t=1}^T \Dcart(O;\eta_0,\theta,\alpha)-\psi_0(\theta;\alpha)\mathcomma$$
		where 
		\begin{equation}
			\begin{split}
				\label{eqn:dcar-def}
				&\Dcart(O;\eta_0,\theta,\alpha) \coloneq \\
				&\Bigl\{\prod_{i<t} W_{i,0}^{\theta,\alpha}\Bigr\}
				(t/T)^\alpha\, \sum_{a\in\{0,1\}}\mu_{t,0}^{\theta,\alpha}(A_t=a,H_t)\, q_t^\theta(a,S_t)
				\frac{I(A_t=a)-\pi_{t,0}(a\mid H_t)}{\pi_{t,0}(a\mid H_t)}\mathcomma
			\end{split}
		\end{equation}
		is the projection of $\Psi(\P_0)$ onto the nuisance tangent space corresponding to the $t$th treatment mechanism, $\mathcal{T}_{\mathrm{CAR},t}$.
	\end{lemma}
	
	\cref{lemma:dcar-expectation} shows that the canonical gradient decomposes into a na\"ive term, $W_0^{\theta,\alpha}Y-\psi_0(\theta;\alpha)$ minus an augmentation term $\sum_{t=1}^T \Dcart(O;\eta_0,\theta,\alpha)$ that projects the IPW mapping onto the nuisance tangent space associated with the sequential treatment mechanism.
	Each $\Dcart$ term can be read as a weighted, scaled IPW residual at time $t$, providing a connection to the HAL score equations. 
	This canonical gradient determines the nonparametric efficiency bound for $\psi_0(\theta;\alpha)$ and will serve as the reference influence function in our asymptotic linearity results.
	
	\subsection{Efficient Inference for the Regimen Response Curve}
	
	We state a set of regularity conditions used to control both the pointwise asymptotic linear expansion of $\psi_n(\theta;\alpha)$ for fixed $(\theta,\alpha)$ and the stronger uniform convergence needed for functional inference over $\Theta$. 
	\cref{assumption:bounded} is a convenient sufficient condition that ensures the relevant influence functions and remainder terms are uniformly integrable and that the empirical process arguments used for $\ell^\infty(\Theta)$ convergence apply without additional tail conditions. 
	\cref{assumption:q-smooth} enforces stochasticity and smoothness of the policy class; \cref{assumption:q-smooth:derivs} is only needed for downstream results on optimizer inference and can be relaxed if one is only interested in curve-level inference.
	
	\begin{assumption}
		\label{assumption:bounded}
		Suppose that the support $\mathcal H_t$ is bounded for all $t=1,\ldots,T$ and that $\mathcal Y$ is bounded.
	\end{assumption}
	
	\begin{assumption}
		\label{assumption:q-smooth}
		Suppose that the parametric stochastic intervention $q_t^\theta(A_t,S_t)$ satisfies the following.
		\begin{enumerate}[label=(\alph*),ref=\theassumption(\alph*)]
			\item There exist constants $0<\munderbar q < \bar q < 1$ such that $\munderbar q \le q_t^\theta(a,s) \le \bar q$ for all $(a,s,t,\theta)\in\{0,1\}\times \bigcup_{t=1}^T\mathcal S_t \times \{1,\ldots,T\}\times\Theta$.
			\item For each $a \in \{0,1\}$, $q_t^\theta(a,\cdot)$ is c\`adl\`ag with sectional variational norm uniformly bounded over $\Theta$.
			\item \label{assumption:q-smooth:derivs}
			For each $a\in\{0,1\}$ and almost every $S_t$, the map $\theta\mapsto q_t^\theta(a,S_t)$ is thrice continuously differentiable.
			Furthermore, there exist constants $C_1,C_2,C_3$ such that
			$$\sup_{\theta\in\Theta}\|\nabla_\theta q_t^\theta(a,S_t)\|\le C_1\mathcomma \quad
			\sup_{\theta\in\Theta}\|\nabla_\theta^2 q_t^\theta(a,S_t)\|\le C_2\mathcomma \quad \text{and} \quad\sup_{\theta\in\Theta}\|\nabla_\theta^3 q_t^\theta(a,S_t)\|\le C_3\mathcomma$$
			almost surely.
		\end{enumerate}
	\end{assumption}

	These conditions are readily verifiable for the logistic model paired with a compact index set $\Theta$.
	However, deterministic interventions, given by $q_t^\theta(1,s)=I\{b(s)\trans\theta>0\}$, do not satisfy these conditions.
	We note that the indicator function can be retrieved from the logistic model by letting $\lVert\theta\rVert \to\infty$; this is in part why we do not consider the general index set $\Theta=\mathbb{R}^p$ and require a compact subset.
	
	We now combine the canonical gradient representation from \cref{lemma:dcar-expectation} with the undersmoothing condition to obtain an asymptotic linear expansion for $\psi_n(\theta;\alpha)$.
	
	\begin{theorem}
		\label{theorem:uipw-asym}
		Suppose that \cref{assumption:identification:positivity,,assumption:bounded} hold and that $\pi_{t,n}$ is sufficiently undersmoothed such that \cref{eqn:hal-undersmooth-condition} holds.
		Consider a fixed $(\theta,\alpha)\in\Theta\times [0,\infty)$. 
		Suppose that \cref{assumption:cadlag,,assumption:hal-space} hold locally at $(\theta,\alpha)$.
		Then, 
		$$\psi_n(\theta; \alpha)-\psi_0(\theta; \alpha) = (\P_n-\P_0)\phi_\psi(O;\eta_0,\theta,\alpha)+R_n(\theta,\alpha)\mathcomma$$
		where $R_n(\theta,\alpha)=o_p(n^{-1/2})$. 
		Consequently,
		$$\sqrt{n}\{\psi_n(\theta,\alpha)-\psi_0(\theta,\alpha)\} \dto N\big[0, \P_0\{\phi_\psi^2(O;\eta_0,\theta,\alpha)\}\big]\mathcomma$$
		where $X_n\dto X$ denotes weak convergence of $X_n$ to $X$.
		Moreover, if \cref{assumption:q-smooth} holds and \cref{assumption:cadlag,,assumption:hal-space} hold uniformly for all $\theta\in \Theta$, 
		$$\sqrt{n}\{\psi_{n}(\theta;\alpha)-\psi_0(\theta;\alpha)\}\dto \G(\theta;\alpha)$$
		in $\ell^\infty(\Theta)$,
		where
		$\G$ denotes a tight Gaussian process indexed by $\Theta$ with covariance
		$$\Cov\{\G(\theta_1; \alpha),\G(\theta_2;\alpha)\}=\Cov\{ \phi_\psi(O,\eta_0,\theta_1,\alpha),\phi_\psi(O;\eta_0,\theta_2,\alpha)\}\mathcomma$$
		and $\ell^\infty(T)$ denotes the Banach space of all bounded functions $f:T\to\mathbb{R}$ equipped with the uniform norm.
	\end{theorem}
	
	\cref{theorem:uipw-asym} states that, under regularity conditions and sufficient undersmoothing, $\psi_n(\theta;\alpha)$ is regular and asymptotically linear with influence function equal to the canonical gradient $\phi_\psi$, and therefore attains the nonparametric efficiency bound for each fixed $(\theta,\alpha)$.
	The uniform weak convergence result strengthens this to support functional inference over $\Theta$, providing the basis for simultaneous inference for the regimen-response curve.
	
	\begin{remark}
		Although these results show that one can obtain an asymptotically efficient estimator without estimating the outcome regression functions $\mu_{t,0}^{\theta,\alpha}$, the plug-in variance estimator $\P_n\{\phi(O;\eta_n,\theta,\alpha)^2\}$ requires a consistent estimator for these functions for every $t$ and $\theta$.
		To avoid estimating the outcome regression functions, the variance can be estimated using the bootstrap, or conservatively by treating the propensity score and weights as known.
	\end{remark}
	
	The next corollary shows that, for a vanishing tilting sequence $\alpha_n$, the tilted estimator has the same first-order asymptotic linear representation---and hence limiting law---as the untilted estimator $\psi_n(\theta;0)$.
	
	\begin{corollary}
		\label{corollary:uipw-asymp-alpha}
		Suppose that \cref{assumption:identification:positivity,,assumption:bounded} hold, that $\pi_{t,n}$ is sufficiently undersmoothed.
		Fix $\theta\in\Theta$ and suppose that \cref{assumption:cadlag,,assumption:hal-space} hold for $\theta$ uniformly for all $\alpha\in[0,\alpha_\epsilon]$ for some $\alpha_\epsilon>0$.
		If $\alpha_n=o_p(n^{-1/2})$, then 
		$$\psi_n(\theta; \alpha_n)-\psi_0(\theta) = (\P_n-\P_0)\phi_\psi(O;\eta_0,\theta,\alpha=0)+o_p(n^{-1/2})\mathcomma$$ and
		$$\sqrt{n}\{\psi_n(\theta;\alpha_n)-\psi_0(\theta)\} \dto N\big[0, \P_0\{\phi_\psi^2(O;\eta_0,\theta,\alpha=0)\}\big]\mathperiod$$
		Likewise if \cref{assumption:q-smooth} holds and \cref{assumption:cadlag,,assumption:hal-space} hold uniformly for all $(\theta,\alpha)\in\Theta\times [0,\alpha_\epsilon]$, then 
		$$\sqrt{n}\{\psi_{n}(\theta;\alpha_n)-\psi_0(\theta)\}\dto \G(\theta;\alpha=0)\mathcomma$$ with $\G$ as defined in \cref{theorem:uipw-asym}.
	\end{corollary}

	\subsection{Efficient Policy Optimization}
	
	We now study the learned optimal policy parameter $\theta_n^\star(\alpha)$, defined as a maximizer of the penalized value function $V_n(\theta;\alpha)$ over $\Theta$. 
	Our argument follows a standard argmax/M-estimation strategy: under suitable curvature, the optimizer can be characterized through a first-order condition, and its limiting behavior is governed by the corresponding gradient process. 
	
	\begin{assumption}
		\label{assumption:v0-regularity}
		Suppose that 
		$\theta_0^\star(\alpha)$ is a unique interior point of $\Theta$ 
		and that $H_0(\alpha)\coloneq\nabla^2_\theta V_0\{\theta_0^\star(\alpha);\alpha\}$ is negative definite in a neighborhood of $\theta_0^\star(\alpha)$.
	\end{assumption}
	
	Assumption~\ref{assumption:v0-regularity} imposes standard regularity for argmax functionals: uniqueness and interiority ensure that $\theta_0^\star(\alpha)$ is characterized by the first-order condition $\nabla_\theta V_0\{\theta_0^\star(\alpha);\alpha\}=0$, while negative definiteness of $H_0(\alpha)$ provides local curvature needed to invert the linearization and obtain a $\sqrt{n}$-rate expansion.
	Under this assumption, expanding the first-order condition,
	\begin{equation*}
		\theta_n^\star(\alpha)-\theta_0^\star(\alpha)=-H_0(\alpha)^{-1}\nabla_\theta V_n\{\theta_0^\star(\alpha);\alpha\} + o_p(n^{-1/2})\mathperiod
	\end{equation*}
	Consequently, inference for $\theta_n^\star(\alpha)$ reduces to establishing regular asymptotic linearity of the gradient $\nabla_\theta V_n(\theta;\alpha)$ evaluated at $\theta=\theta_0^\star(\alpha)$.
	
	To that end, we show that, under one additional HAL condition regarding approximation of the relevant gradient multipliers (\cref{assumption:hal-space-deriv}), the gradient process is asymptotically linear with influence function equal to the canonical gradient of $\nabla_\theta V_0(\theta;\alpha)$ (\cref{lemma:deriv-gradient,,lemma:psi-deriv-asymp}).
	It follows that the optimizer inherits an asymptotic linear expansion determined by the inverse Hessian, $H_0(\alpha)^{-1}$, and the canonical gradient of $\nabla_\theta\psi_0(\theta;\alpha)$. This is formalized in the following theorem.
	
	\begin{theorem}
		\label{theorem:opt-theta-asym}
		Let $\alpha\ge0$ and suppose that \cref{assumption:identification:positivity,,assumption:bounded,,assumption:q-smooth,,assumption:v0-regularity} hold, that $\pi_{t,n}$ is sufficiently undersmoothed, and that \cref{assumption:cadlag,,assumption:hal-space,,assumption:hal-space-deriv} hold uniformly in a neighborhood of $\theta_0^\star(\alpha)$.
		Then,
		\begin{equation*}
			\theta_n^\star(\alpha)-\theta_0^\star(\alpha) = (\P_n-\P_0)\left[ -H_0^{-1}(\alpha)\phi_{\nabla \psi}\{O;\eta_0,\eta_0',\theta_0^\star(\alpha),\alpha\}\right] + o_p(n^{-1/2})\mathcomma
		\end{equation*}
		and
		\begin{equation*}
			\sqrt{n}\{\theta_n^\star(\alpha)-\theta_0^\star(\alpha)\} \dto N\left( 0, H_0^{-1}(\alpha)\P_0[\phi_{\nabla \psi}^{\otimes 2}\{O;\eta_0,\eta_0',\theta_0^\star(\alpha),\alpha\}] H_0^{-1}(\alpha)\right)\mathcomma
		\end{equation*}
		where $f^{\otimes2}=ff\trans$ for any measurable function $f$, and $\phi_{\nabla\psi}$ is the canonical gradient of $\nabla_\theta \psi_0(\theta;\alpha)$ as defined in \cref{lemma:deriv-gradient}.
	\end{theorem}
	
	Theorem~\ref{theorem:opt-theta-asym} shows that, under curvature and undersmoothing conditions, $\theta_n^\star(\alpha)$ is a regular asymptotically linear estimator of $\theta_0^\star(\alpha)$ with a $\sqrt{n}$ limit distribution. 
	Relative to curve-level inference, the optimizer expansion introduces the additional nuisance component 
	$\eta_0'\coloneq\{\nabla_\theta \mu_{t,0}^{\theta,\alpha} \mid_{\theta=\theta_0^\star(\alpha)}\}_{t=1}^T$,
	reflecting that inference for an argmax depends on the behavior of the gradient process and therefore on derivatives of the forward regressions with respect to $\theta$. 
	
	Analogous to \cref{corollary:uipw-asymp-alpha}, if $\alpha_n=o_p(n^{-1/2})$ then $\theta_n^\star(\alpha_n)$ admits the same first-order asymptotic linear representation as the untilted optimizer estimator $\theta_n^\star(0)$.
	\begin{corollary}
		\label{corollary:opt-theta-asym-alpha}
		Suppose that
		\cref{assumption:identification:positivity,,assumption:bounded,,assumption:q-smooth,,assumption:v0-regularity-uniform} hold,
		that $\pi_{t,n}$ is sufficiently undersmoothed such that \cref{eqn:hal-undersmooth-condition} holds, and that
		\cref{assumption:cadlag,,assumption:hal-space,,assumption:hal-space-deriv} hold at $\theta_0^\star(\alpha)$ uniformly over $\alpha\in[0,\alpha_\epsilon]$ for some $\alpha_\epsilon>0$.
		Then if $\alpha_n=o_p(n^{-1/2})$,
		\begin{equation*}
			\theta_n^\star(\alpha_n)-\theta_0^\star = (\P_n-\P_0)\left\{ -H_0^{-1}\phi_{\nabla \psi}(O;\eta_0,\eta_0',\theta_0^\star,\alpha=0)\right\} + o_p(n^{-1/2})\mathcomma
		\end{equation*}
		and
		\begin{equation*}
			\sqrt{n}\{\theta_n^\star(\alpha_n)-\theta_0^\star\} \dto N\left[ 0, H_0^{-1}\P_0\{\phi_{\nabla \psi}^{\otimes 2}(O;\eta_0,\eta_0',\theta_0^\star,\alpha=0)\} H_0^{-1}\right]\mathperiod
		\end{equation*}
	\end{corollary}
	
	\begin{remark}
		Similar to variance estimation for the regimen-response curve, the plug-in variance estimator $ H_n(\alpha)^{-1}\P_n[\phi^{\otimes2}_{\nabla\psi}\{O; \eta_n, \eta_n',\theta_n^\star(\alpha),\alpha)\}]  H_n(\alpha)^{-1}$ requires consistent estimation of all outcome regression functions embedded in $\eta_0$; however, this additionally requires consistent estimation of the derivatives of all outcome regression functions, as expressed through nuisance parameter $\eta_n'$.
		Again, variance estimation can instead be performed using the bootstrap or conservatively by treating the propensity score as fixed.
	\end{remark}
	
	\subsection{Data-Adaptive Tilting}

	The asymptotic results in \cref{corollary:uipw-asymp-alpha,,corollary:opt-theta-asym-alpha} require $\alpha_n=o_p(n^{-1/2})$ so that the tilted estimators share the same first-order limit as the untilted analogs.
	The following lemma provides sufficient conditions for the trimmed lower-confidence-bound selector $\tilde\alpha_n$ to attain the necessary rate.
	
	\begin{lemma}
		\label{lemma:alpha-clip-op}
		Let $\delta>0$.
		Suppose that \cref{assumption:identification:positivity,,assumption:bounded,,assumption:q-smooth} hold, that $\pi_{t,n}$ is sufficiently undersmoothed such that \cref{eqn:hal-undersmooth-condition} holds, that \cref{assumption:cadlag,,assumption:hal-space,,assumption:hal-space-deriv} hold uniformly at $\theta_0^\star(\alpha)$ uniformly over $\alpha\in[0,\delta]$, and that \cref{assumption:v0-regularity-uniform} holds.
		Moreover, suppose that
		\begin{enumerate}[label=(\alph*),ref=\thelemma(\alph*)]
			\item The population profiled value function satisfies $\alpha_0\coloneq\argmax_{\alpha\ge0} \bar{V}_0(\alpha)=0$, $\bar{V}_0(0)-\bar{V}_0(\alpha)\ge c\alpha$ for all $\alpha\in[0,\delta]$, and $\inf_{\alpha>\delta}\{\bar V_0(0) - \bar V_0(\alpha)\} \ge \epsilon_\delta >0$. 
			\item The profiled standard error estimator satisfies $\sup_{\alpha\in[0,\delta]} n^{-1/2}\bar\sigma_n(\alpha)=O_p(n^{-1/2})$
		\end{enumerate}
		Then, $\alpha_n = O_p(n^{-1/2})$ and $\tilde\alpha_n=o_p(n^{-1/2})$ for any $\kappa\ge0$.
	\end{lemma}
	
	The first condition of \cref{lemma:alpha-clip-op} requires that the untilted regimen response curve attains the highest possible value function (as discussed in \cref{remark:v-bias}) and that there is linear separation near $\alpha=0$. 
	The second condition ensures that $\sup_{\alpha\in[0,\delta]}\lvert\{\bar V_n(\alpha)-\kappa n^{-1/2}\bar \sigma_n(\alpha)\}-\bar V_0(\alpha)\rvert=O_p(n^{-1/2})$.
	We note that this is a weak assumption that generally holds under \cref{assumption:identification:positivity,,assumption:bounded}.
	Together, these conditions enable an argmax argument for $\alpha_n$ that accounts for the boundary maximizer.  
	
	\section{Numerical Simulations}
	
	\subsection{Simulation Setup}
	
	The simulation study was designed to generate longitudinal data with features that could plausibly arise in studies of mobile prompts intended to support cigarette quit attempts.
	In this setting, prompts may be useful when an individual is experiencing elevated negative affect, particularly when cigarettes are available, but may be burdensome when delivered at low-risk moments.
	We therefore generated repeated measurements of two time-varying state variables $S_{t,1}$ and $S_{t,2}$, intended to respectively represent negative affect and cigarette availability at time $t$.
	For participant $i=1,\ldots,n$ we generated the trajectory $O \coloneq (S_{1},A_{1}, \ldots, S_{T},A_{T},Y)$ where $S_{t}\coloneq (S_{t,1},S_{t,2})$ contains both state variables.
	The data generation process is described in \cref{sec:simulation-details}.
	
	We considered the one-dimensional policy class $q_t^\theta(A_t,S_t)=\expit(\theta S_{t,1})$ with $\Theta=[-0.73,0.73]$.
	This index set, constructed following \cref{sec:index-set}, ensures that $q_t^\theta(a,S_t)\in[0.05,0.95]$ for all $a\in\{0,1\}$, $S_t\in\mathcal S_t$. 
	Policy optimization was performed using a penalty weight of $\lambda=0.01$.
	
	We considered five IPW estimators, which differed only in how the treatment mechanism was handled.
	The oracle estimator used the true treatment probabilities $\pi_{t,0}$.
	All other estimators used pooled propensity score models fit to observations from all $T$ decision points, with decision time included as a covariate.
	The parametric estimator used a correctly specified logistic regression model with linear terms for $S_{t,1}$, $S_{t,2}$, $t$, and $S_{t,1}S_{t,2}$.
	The remaining three estimators used the HAL to estimate the pooled propensity score as a function of $S_{t,1}$, $S_{t,2}$, and $t$, and were cross-fit over two external folds.
	The CV HAL estimator used the standard HAL fit with the $L_1$-norm selected by five-fold cross-validation.
	The UIPW-$\Dcar$ and UIPW-Score estimators selected undersmoothed HAL fits by minimizing, respectively, aggregated $L_2$ $\Dcar$ (\cref{eqn:l2-dcar-criterion-body}) and score (\cref{eqn:score-criterion-body}) criteria.
	Because a pooled propensity model was used, both UIPW selectors targeted sums of the corresponding time-specific criteria across $t=1,\ldots,T$, matching the decomposition of the longitudinal IPW remainder.
	For each estimator, we considered both untilted ($\alpha=0$) and data-adaptive (\cref{sec:adaptive-tilting}) implementations.
	We implemented truncated product-weights for all estimators, forcing $\prod_{t=1}^TW_{t,n}^{\theta,\alpha}\le20$.
	
	We evaluated estimator performance under the following factorial design: $a\in\{0,0.25\}$ (where $\pi_{t,0}(A_t=1\mid H_t) = \expit(a S_{t,1})$, thereby considering MRT and observational designs), $n\in\{100,500,1000\}$, and $T\in\{10,25,50\}$.
	We report relevant metrics, averaged across $200$ Monte Carlo replications, as follows. 
	For the regimen response curve, we summarize the pointwise bias, mean squared error (MSE), and coverage as a function of $\theta$. 
	For the optimal policy, we summarize the square root mean squared error (RMSE)
	and coverage.%
	
	\subsection{Simulation Results}
	
	We summarize performance for both regimen-response curve estimation and policy learning under the MRT ($a=0$) and observational ($a=0.25$) assignment mechanisms.
	
	\cref{fig:psi-bias} presents the scaled bias of all estimators without tilting (i.e., $\alpha=0$) and shows that the bias was generally small in the MRT setting across all estimators.
	We note that in this setting, CV HAL often selected very small bases of $1$--$2$ functions, and therefore behaved similarly to a low-dimensional or near-parametric estimator.
	In contrast, in the observational setting, CV HAL shows more pronounced scaled bias as $n$ increases, especially when $T=10$.
	In contrast, the UIPW estimators generally controlled the bias.
	This pattern supports the role of undersmoothing in reducing the leading IPW bias term when the treatment mechanism must be learned nonparametrically.
	
	\begin{figure}
		\centering
		\includegraphics[width=\linewidth]{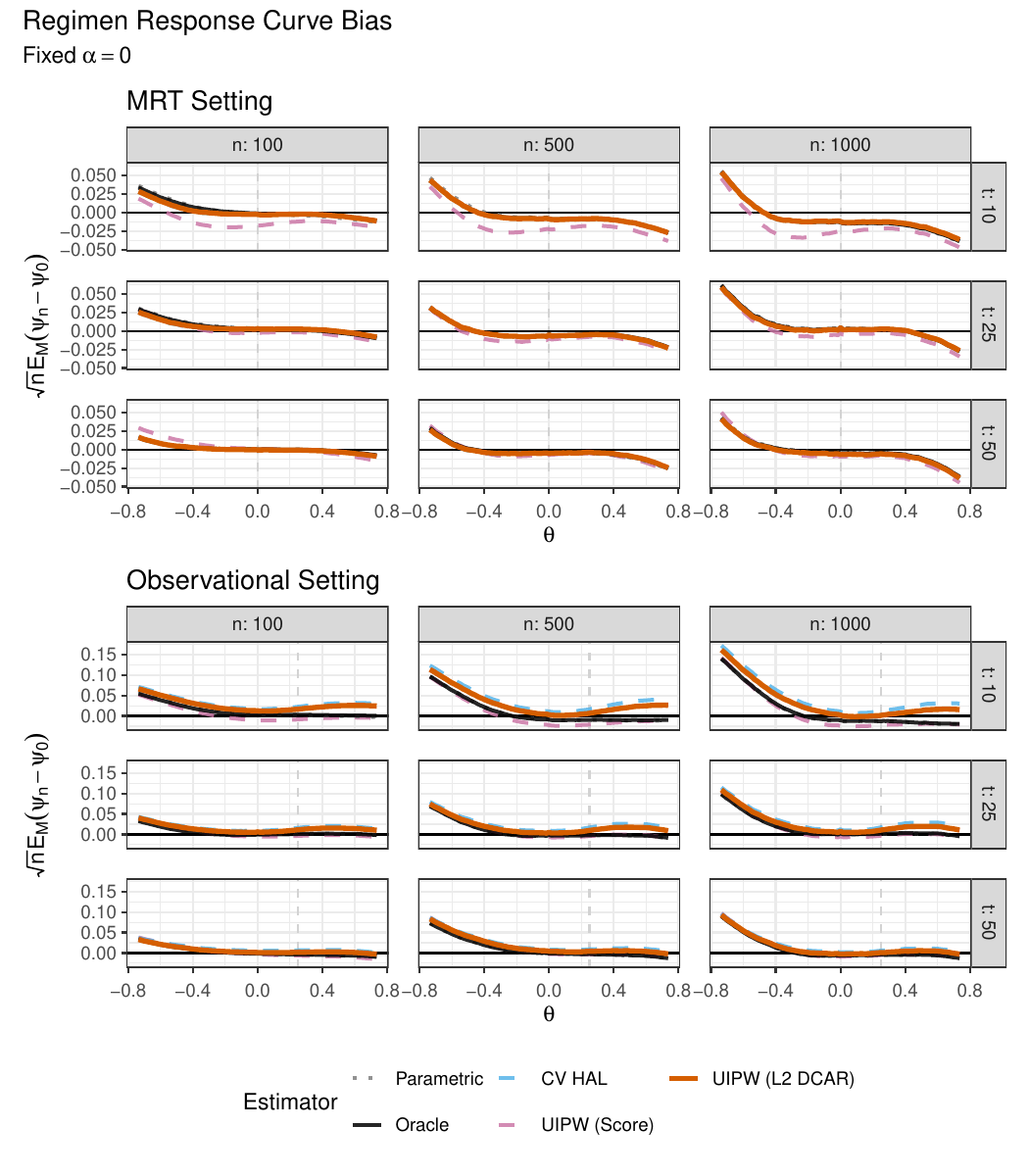}
		\caption{Scaled pointwise bias for the regimen response curve.}
		\label{fig:psi-bias}
	\end{figure}
	
	\cref{fig:psi-coverage} demonstrates broadly similar pointwise coverage patterns across all untilted estimators, including the oracle IPW estimator.
	Coverage was generally close to the nominal rate for moderate values of $\theta$, while undercoverage was observed when the target policy differed more from the observational policy. 
	The fact that undercoverage appears in similar regions for all estimators---including the oracle IPW estimator---suggests that the main driver was not estimator-specific bias control or the variance estimator, but numerical positivity issues associated with the longitudinal weights. 
	\cref{fig:psi-coverage-with-alpha} compares coverage for the oracle and UIPW estimators using data-adaptive $\alpha_n$ tilting.
	There, tilting resulted in decreased coverage for $\theta$ far from the observational policy; this was more pronounced in the MRT setting, where larger $\alpha_n$ were generally selected, resulting in larger bias away from the observational policy.
	
	\begin{figure}
		\centering
		\includegraphics[width=\linewidth]{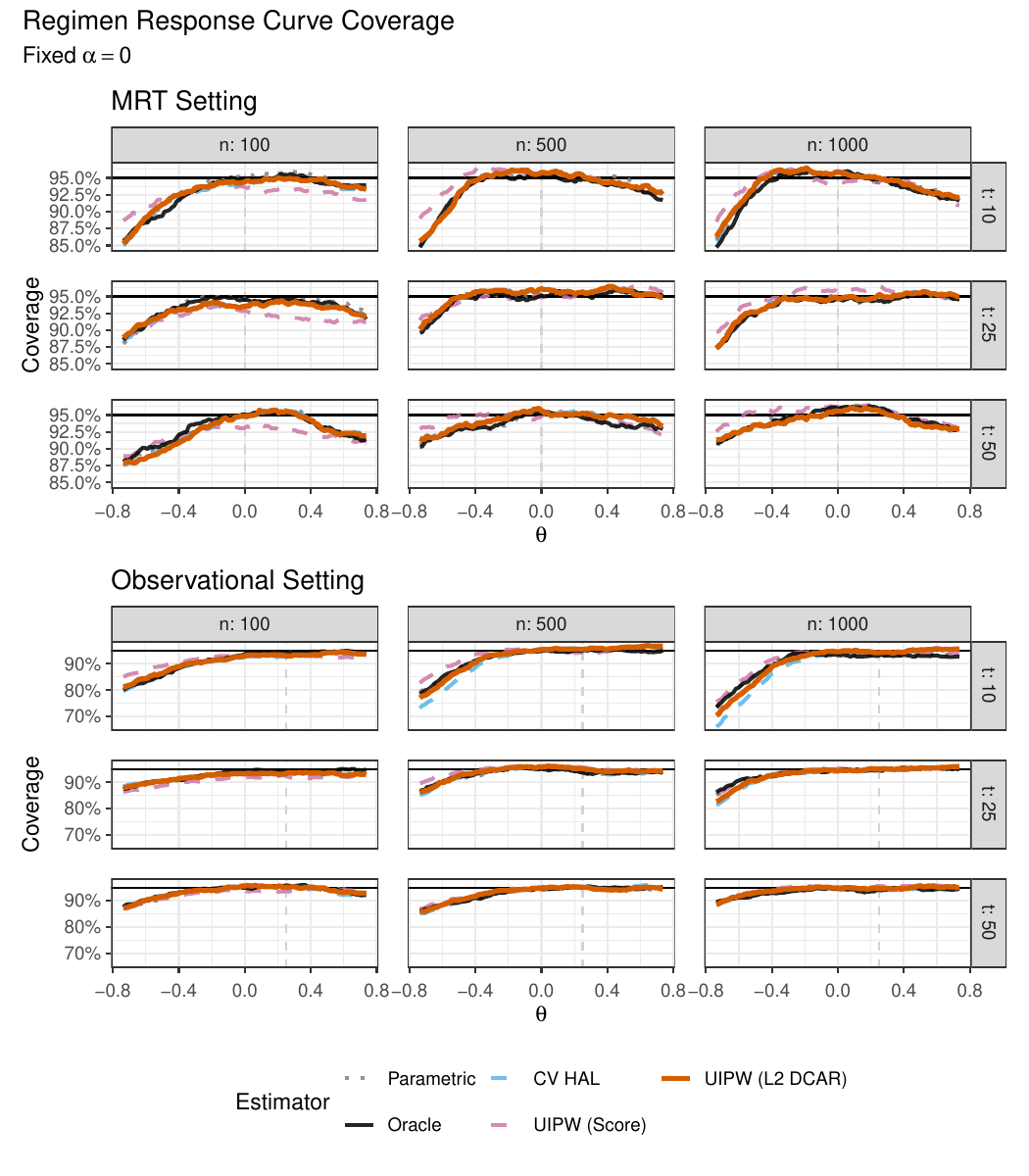}
		\caption{Pointwise coverage for the regimen response curve.}
		\label{fig:psi-coverage}
	\end{figure}
	
	\cref{fig:psi-rel-eff} presents the relative efficiency of each estimator to the oracle IPW estimator in the absence of tilting.
	In the MRT setting, most estimators behaved similarly to the oracle IPW estimator. 
	Score-based undersmoothing resulted in larger variances than all other estimators nearly uniformly across settings; this is likely the result of the score-based selector often selecting the maximum amount of undersmoothing allowed, resulting in potential numerical instability.
	In the observational setting, the relative efficiencies of all nonparametric estimators varied more strongly in $\theta$.
	This heterogeneity reflects that the outcome regression functions $\mu_{t,0}^{\theta,\alpha}$ that appear in the leading bias term, \cref{eqn:ipw-remainder-generic}, depend on $\theta$. 
	Hence, some regions of the policy class appear to induce more stable or easier-to-approximate empirical score functions than others; this is reflected in efficiency gains where these functions are adequately solved, and potential inefficiency otherwise.
	
	\begin{figure}
		\centering
		\includegraphics[width=\linewidth]{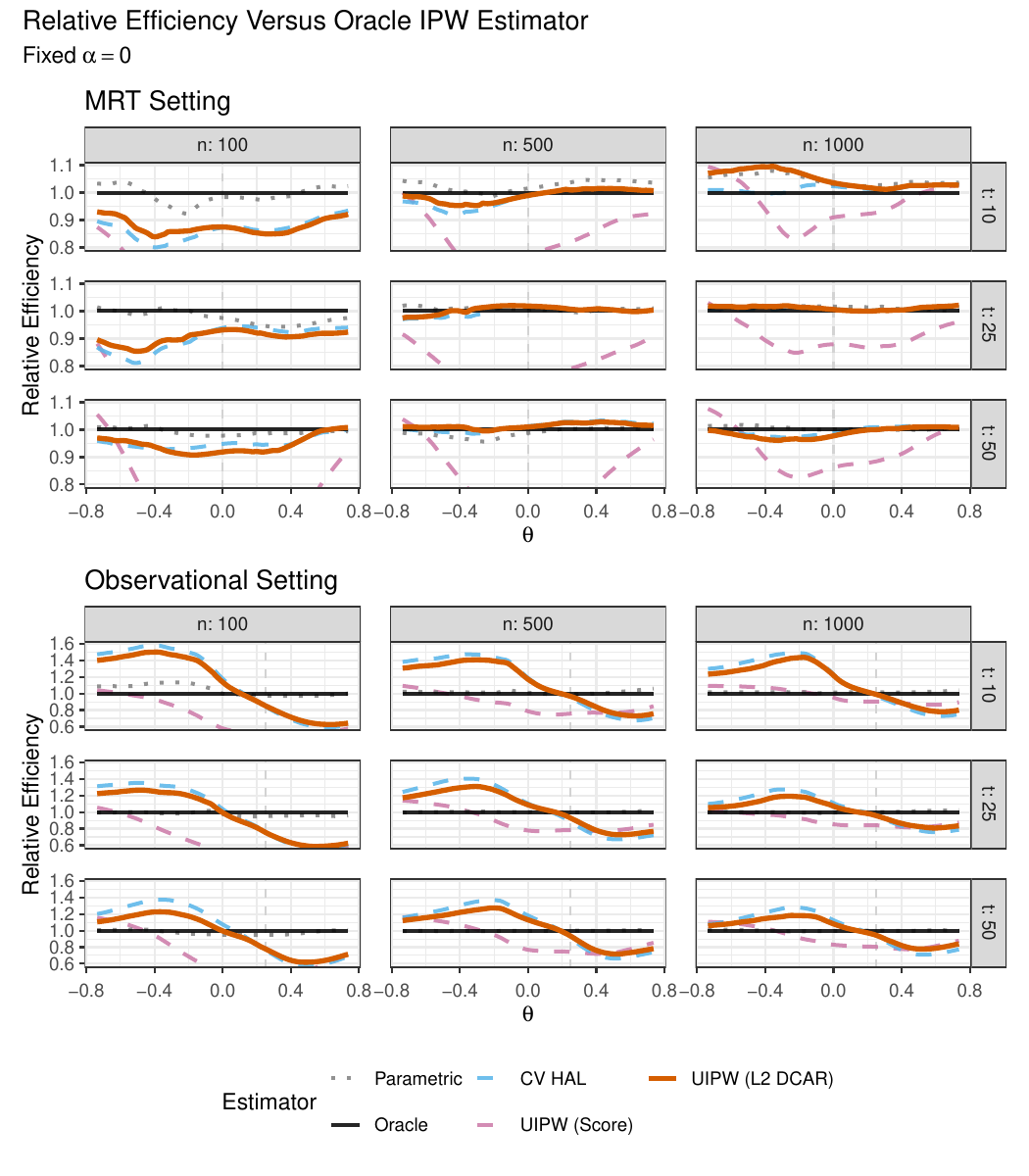}
		\caption{Pointwise relative efficiency for the regimen response curve versus the oracle inverse probability weighted estimator.
			Note: the vertical axis has been truncated below for visual clarity.}
		\label{fig:psi-rel-eff}
	\end{figure}
	
	\cref{fig:psi-rel-mse-tilt} shows that adaptive tilting reduced the MSE for nearly all estimators and settings.
	These gains were not restricted to estimators with flexibly estimated propensity scores---oracle and correctly specified parametric IPW estimators also benefited.
	These gains were often largest in larger-$T$ settings, consistent with the accumulation of weight variability across repeated decision points.
	
	\begin{figure}
		\centering
		\includegraphics[width=\linewidth]{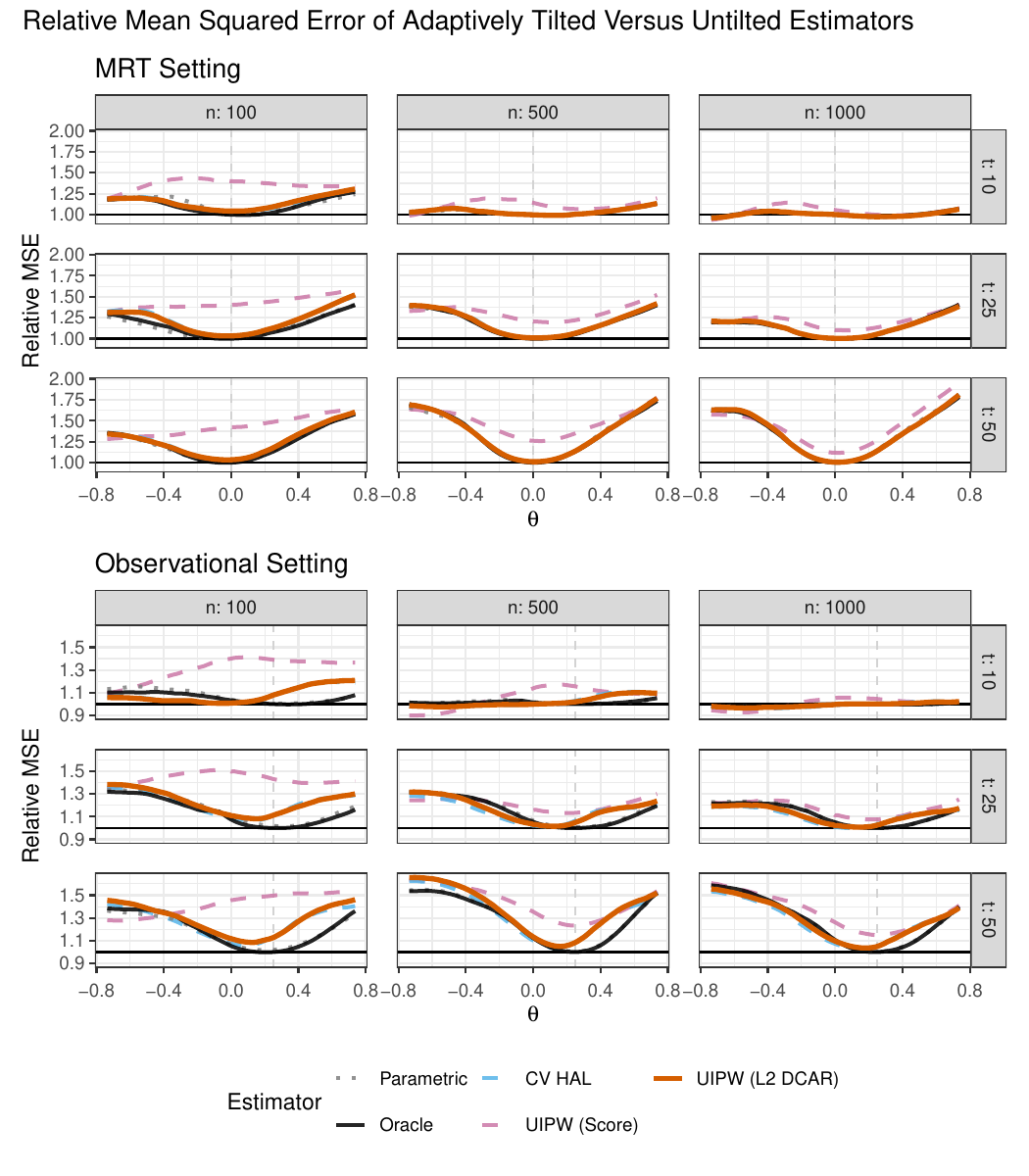}
		\caption{Pointwise relative mean squared error for the regimen response curve for each estimator with data-adaptive tilting versus its untilted analog.}
		\label{fig:psi-rel-mse-tilt}
	\end{figure}
	
	\cref{tbl:policy-rmse} presents the scaled RMSE for each policy selector induced by the corresponding regimen response curve estimator, with and without adaptive tilting.
	Adaptive tilting generally improved policy RMSE relative to the untilted comparators; the largest gaps between tilted and untilted estimators tended to occur in high-$T$ settings.
	Additionally, UIPW using the $L_2$ $\Dcar$ selector generally had similar RMSE and regret to the optimal selector within each setting.
	
	\cref{tbl:policy-coverage} displays the coverage of each policy estimator.
	Most estimators attained or surpassed the nominal coverage rate for moderate $n$.
	In some settings, all estimators demonstrated overcoverage.

	\begin{landscape}
		\begin{table}[htpb]
			\centering
			\caption{Scaled root mean squared error in estimating the optimal stochastic policy.}
			\label{tbl:policy-rmse}
			\setlength{\tabcolsep}{1.5pt}
			\fontsize{12.0pt}{14.0pt}\selectfont
\begin{tabular*}{\linewidth}{@{\extracolsep{\fill}}llrrrrrrrrr}
\toprule
 &  & \multicolumn{9}{c}{\(\sqrt{n}\,\sqrt{\mathrm{E}_B\left[\{\theta_n^{\star}(\alpha_n)-\theta_0^{\star}(0)\}^2\right]}\)} \\ 
\cmidrule(lr){3-11}
 &  & \multicolumn{3}{c}{$T = 10$} & \multicolumn{3}{c}{$T = 25$} & \multicolumn{3}{c}{$T = 50$} \\ 
\cmidrule(lr){3-5} \cmidrule(lr){6-8} \cmidrule(lr){9-11}
Estimator & Tilting & $n = 100$ & $n = 500$ & $n = 1000$ & $n = 100$ & $n = 500$ & $n = 1000$ & $n = 100$ & $n = 500$ & $n = 1000$ \\ 
\midrule\addlinespace[2.5pt]
\multicolumn{11}{l}{Micro-Randomized Trial} \\[2.5pt] 
\midrule\addlinespace[2.5pt]
UIPW (Score) & Adaptive & 3.69 & 1.25 & {\bfseries 1.09} & {\bfseries 4.22} & 2.31 & 1.64 & {\bfseries 4.72} & 3.01 & 2.16 \\ 
UIPW (Score) & Untilted & 4.47 & 2.83 & 1.84 & 5.70 & 6.27 & 5.04 & 6.49 & 8.56 & 8.15 \\ 
UIPW (L2 DCAR) & Adaptive & {\bfseries 3.45} & 1.27 & 1.12 & 4.43 & 2.25 & 1.18 & 4.83 & 3.06 & {\bfseries 1.74} \\ 
UIPW (L2 DCAR) & Untilted & 4.08 & 2.30 & 1.91 & 5.54 & 5.38 & 4.83 & 6.41 & 8.24 & 7.08 \\ 
CV HAL & Adaptive & 3.46 & 1.28 & 1.14 & 4.48 & 2.25 & {\bfseries 1.17} & 4.88 & 3.30 & 1.75 \\ 
CV HAL & Untilted & 4.06 & 2.31 & 1.91 & 5.57 & 5.39 & 4.97 & 6.45 & 8.20 & 7.21 \\ 
Parametric & Adaptive & 3.68 & {\bfseries 1.10} & 1.16 & 4.52 & 2.04 & 1.86 & 4.86 & 3.02 & 1.87 \\ 
Parametric & Untilted & 4.14 & 2.25 & 1.97 & 5.64 & 5.63 & 4.87 & 6.32 & 8.21 & 7.55 \\ 
Oracle & Adaptive & 3.75 & 1.13 & 1.18 & 4.53 & {\bfseries 1.47} & 1.19 & 4.83 & {\bfseries 2.89} & 2.06 \\ 
Oracle & Untilted & 4.16 & 2.16 & 2.02 & 5.55 & 5.30 & 5.03 & 6.32 & 8.32 & 7.49 \\ 
\midrule\addlinespace[2.5pt]
\multicolumn{11}{l}{Observational} \\[2.5pt] 
\midrule\addlinespace[2.5pt]
UIPW (Score) & Adaptive & 4.10 & 1.34 & 1.00 & 4.52 & {\bfseries 2.59} & 1.59 & 5.22 & 4.86 & {\bfseries 2.64} \\ 
UIPW (Score) & Untilted & 4.76 & 2.97 & 1.13 & 6.12 & 7.45 & 5.28 & 6.65 & 10.57 & 9.33 \\ 
UIPW (L2 DCAR) & Adaptive & 3.56 & 0.74 & 0.86 & 4.16 & 2.67 & 0.95 & 4.90 & {\bfseries 4.83} & 2.93 \\ 
UIPW (L2 DCAR) & Untilted & 4.06 & 1.78 & 1.12 & 5.58 & 6.85 & 4.68 & 6.54 & 10.52 & 8.70 \\ 
CV HAL & Adaptive & {\bfseries 3.41} & {\bfseries 0.73} & {\bfseries 0.85} & {\bfseries 4.12} & 2.63 & {\bfseries 0.93} & {\bfseries 4.81} & 4.93 & 2.94 \\ 
CV HAL & Untilted & 4.03 & 1.40 & 1.14 & 5.53 & 6.73 & 4.07 & 6.42 & 10.15 & 8.29 \\ 
Parametric & Adaptive & 4.00 & 1.38 & 1.67 & 4.60 & 3.58 & 1.36 & 5.33 & 5.23 & 3.29 \\ 
Parametric & Untilted & 4.63 & 2.64 & 1.35 & 5.95 & 7.57 & 4.25 & 6.85 & 11.13 & 9.30 \\ 
Oracle & Adaptive & 3.95 & 1.38 & 1.01 & 4.56 & 3.62 & 1.36 & 5.21 & 5.12 & 3.27 \\ 
Oracle & Untilted & 4.54 & 2.71 & 1.37 & 5.95 & 7.73 & 4.43 & 6.76 & 10.90 & 9.08 \\ 
\bottomrule
\end{tabular*}

		\end{table}
	\end{landscape}

	\begin{landscape}
		\begin{table}[htpb]
			\centering
			\caption{Coverage for the optimal stochastic policy.}
			\label{tbl:policy-coverage}
			\setlength{\tabcolsep}{1.5pt}
			\fontsize{12.0pt}{14.0pt}\selectfont
\begin{tabular*}{\linewidth}{@{\extracolsep{\fill}}llrrrrrrrrr}
\toprule
 &  & \multicolumn{9}{c}{Coverage (\%) for \(\theta_0^{\star}(0)\)} \\ 
\cmidrule(lr){3-11}
 &  & \multicolumn{3}{c}{$T = 10$} & \multicolumn{3}{c}{$T = 25$} & \multicolumn{3}{c}{$T = 50$} \\ 
\cmidrule(lr){3-5} \cmidrule(lr){6-8} \cmidrule(lr){9-11}
Estimator & Tilting & $n = 100$ & $n = 500$ & $n = 1000$ & $n = 100$ & $n = 500$ & $n = 1000$ & $n = 100$ & $n = 500$ & $n = 1000$ \\ 
\midrule\addlinespace[2.5pt]
\multicolumn{11}{l}{Micro-Randomized Trial} \\[2.5pt] 
\midrule\addlinespace[2.5pt]
UIPW (Score) & Adaptive & 89 & 99 & 100 & 88 & 99 & 99 & 87 & 98 & 99 \\ 
UIPW (Score) & Untilted & 84 & 96 & 99 & 79 & 93 & 96 & 75 & 88 & 89 \\ 
UIPW (L2 DCAR) & Adaptive & 90 & 99 & 100 & 87 & 99 & 99 & 86 & 98 & 99 \\ 
UIPW (L2 DCAR) & Untilted & 86 & 98 & 98 & 80 & 92 & 96 & 76 & 85 & 87 \\ 
CV HAL & Adaptive & 90 & 99 & 100 & 86 & 99 & 99 & 86 & 98 & 99 \\ 
CV HAL & Untilted & 87 & 98 & 98 & 80 & 92 & 96 & 76 & 85 & 87 \\ 
Parametric & Adaptive & 89 & 100 & 100 & 85 & 99 & 99 & 87 & 99 & 99 \\ 
Parametric & Untilted & 86 & 98 & 99 & 80 & 93 & 96 & 76 & 85 & 87 \\ 
Oracle & Adaptive & 89 & 100 & 100 & 86 & 99 & 99 & 86 & 99 & 99 \\ 
Oracle & Untilted & 86 & 97 & 98 & 80 & 92 & 95 & 76 & 86 & 86 \\ 
\midrule\addlinespace[2.5pt]
\multicolumn{11}{l}{Observational} \\[2.5pt] 
\midrule\addlinespace[2.5pt]
UIPW (Score) & Adaptive & 88 & 100 & 100 & 87 & 98 & 100 & 82 & 96 & 99 \\ 
UIPW (Score) & Untilted & 85 & 99 & 99 & 78 & 89 & 97 & 74 & 82 & 87 \\ 
UIPW (L2 DCAR) & Adaptive & 91 & 100 & 100 & 88 & 97 & 100 & 83 & 95 & 98 \\ 
UIPW (L2 DCAR) & Untilted & 87 & 99 & 100 & 78 & 87 & 94 & 72 & 77 & 83 \\ 
CV HAL & Adaptive & 92 & 100 & 100 & 88 & 97 & 100 & 83 & 94 & 98 \\ 
CV HAL & Untilted & 87 & 99 & 99 & 79 & 87 & 94 & 72 & 77 & 82 \\ 
Parametric & Adaptive & 89 & 100 & 100 & 87 & 98 & 100 & 82 & 96 & 99 \\ 
Parametric & Untilted & 86 & 99 & 100 & 79 & 91 & 97 & 71 & 79 & 86 \\ 
Oracle & Adaptive & 90 & 100 & 100 & 87 & 98 & 100 & 83 & 97 & 99 \\ 
Oracle & Untilted & 86 & 99 & 100 & 78 & 90 & 97 & 72 & 80 & 86 \\ 
\bottomrule
\end{tabular*}

		\end{table}
	\end{landscape}
	
	\subsection{Simulation Conclusions}
	
	In our simulations, we observed that correctly specified parametric and oracle IPW estimators can be more efficient than the UIPW estimators in small-to-moderate samples, despite our asymptotic efficiency theory.
	This suggests that the undersmoothing required to control the leading IPW bias can introduce non-negligible finite-sample variability.
	This pattern is consistent with prior work showing that practical undersmoothing implementations can have inflated variance at smaller sample sizes, diminishing with $n$.
	In the fully randomized setting studied by \cite{ertefaie_nonparametric_ipw_2023}, undersmoothing based on the score criterion yielded larger mean squared errors than both an unadjusted estimator and an IPW estimator using cross-validated HAL propensity scores, while DCAR-based undersmoothing was typically the most stable of the undersmoothing rules considered.
	Similarly, \cite{pham_nonparametric_2025} reported that, at smaller sample sizes, UIPW variants selected by either score- or $\Dcar$-based criteria can exhibit inflated variance compared with the cross-validated HAL IPW and DR estimators, with these gaps diminishing as $n$ increases.
	Our simulations show a similar pattern in longitudinal settings: the score-based selector often appeared to undersmooth aggressively and produced larger MSEs, whereas the $L_2$ $\Dcar$ selector was generally more stable and yielded policy RMSE and regret close to the best-performing selector across many settings.
	Thus, although undersmoothing is essential for first-order remainder control in the nonparametric setting, its finite-sample performance depends strongly on how the undersmoothing level was selected.
	The pragmatic $L_2$ $\Dcar$ selector used here based on a working parametric model provides one way to target the relevant longitudinal remainder while avoiding full nonparametric estimation of the policy- and time- indexed outcome regressions, but further stabilization may be possible.
	
	More broadly, the empirical efficiency gains available from undersmoothing and other asymptotically efficient estimation approaches may be smaller in the stochastic setting relative to the well-established gains in the deterministic setting. 
	Because stochastic interventions remain probabilistic, they typically yield more stable weights, leaving less variability for efficient augmentation or undersmoothing-based corrections to mitigate.
	This interpretation is consistent with the additional simulations in \cref{sec:ipw-aipw-efficiency}, which suggest that potential efficiency gains associated with (nonparametric) efficient estimation increase as stochastic interventions approach near-deterministic regimes.

	\section{Discussion}
	
	In this work, we address estimation and optimization of policies for distal outcomes in settings with many decision points, as arise in MRTs and other JITAI designs.
	In such settings, standard semiparametric-efficient estimation typically requires estimating a large collection of forward outcome regression functions indexed by both time and policy, which becomes impractical and potentially fragile as $T$ and the policy space $\Theta$ grow.
	Moreover, inverse probability weighted estimators involve products of weights across decision points, and can be numerically unstable when $T$ is large.
	We proposed a UIPW approach for estimating the regimen-response curve and optimal policy over a parametric stochastic policy class.
	Our approach is nonparametrically efficient yet avoids explicit modeling of the policy-indexed outcome regressions.
	To improve finite-sample stability, we introduced a data-adaptive tilting device that interpolates between the target policy and the observed treatment mechanism.
	We showed that the estimated regimen-response curve converges to a Gaussian process over $\Theta$, enabling the construction of simultaneous confidence bands.
	We also derived an asymptotic linearization for the policy optimizer, allowing inference for the learned optimal policy. 
	
	Our approach can be viewed as a saturated estimator of the regimen-response curve, in contrast to MSM analyses that summarize this curve through a lower-dimensional working model.
	\cite{ertefaie_nonparametric_estimation_2023} provided a complementary perspective in the single-decision point setting by defining a fully nonparametric MSM in which the covariate-adjusted regimen-response curve is characterized as a risk minimizer over a rich c\`adl\`ag function class with bounded sectional variational norm.
	A natural future direction is to extend this nonparametric MSM framework to our longitudinal stochastic setting by targeting a covariate adjusted curve $m_0(\theta;V)$ for $V\subseteq S_1$ to support $V$-specific optimal policies, all while using undersmoothing to avoid outcome regression estimation.
	Because our stochastic policy parameterization already induces substantial smoothness in $\theta$ through smoothly varying weights, we expect the primary gains of such an MSM formulation to arise from covariate adjustment and structured estimation across $V$, rather than from additional smoothing across $\theta$ alone.
	
	We introduced an $\alpha$-tilting device that interpolates between the target policy and observed treatment mechanism at each decision point, yielding a stabilized policy family that shrinks early weight contributions toward one. 
	Although our default tilting schedule prioritizes stabilization of earlier weights, one could instead use an approach that preserves the target policy more strongly at earlier decision points if early actions are thought to drive distal outcomes. 
	More flexibly, one could allow the direction of tilting to depend on a signed tuning parameter, selected in a data-adaptive way.
	This interpolation has a close conceptual parallel to incremental propensity score interventions, which conceptualize shifting the propensity score away from the observed mechanism, rather than deterministically setting treatment \citep{kennedy_nonparametric_2019}.
	A key contrast is that incremental and related approaches are motivated as positivity-robust estimands whereas our tilted estimand was developed to ease estimation when overlap holds by design but products of weights are numerically volatile; it does not resolve genuine support violations. 
	
	In summary, these developments were motivated by the practical goal of constructing JITAIs using intensive longitudinal data when the scientific objective is a distal end-of-study outcome rather than a discounted sum of proximal rewards.  
	By providing functional inference for the regimen–response curve and principled inference for the induced optimizer in a high-$T$ setting, the proposed framework helps close the methodological gap between MRT data structures and the inferential targets that commonly drive JITAI design.

	\bibliographystyle{apalike}
	\bibliography{References}
	
	\newpage
	
	\section*{Appendix}
	
	\appendix
	\renewcommand{\thetheorem}{S\arabic{theorem}}
	\renewcommand{\theassumption}{S\arabic{assumption}}
	\renewcommand{\thelemma}{S\arabic{lemma}}
	\renewcommand{\thecorollary}{S\arabic{corollary}}
	\renewcommand{\theequation}{S\arabic{equation}}
	\setcounter{theorem}{0}
	\setcounter{assumption}{0}
	\setcounter{lemma}{0}
	\setcounter{corollary}{0}
	\renewcommand{\thefigure}{S\arabic{figure}}
	\renewcommand{\thetable}{S\arabic{table}}
	\setcounter{figure}{0}
	\setcounter{table}{0}
	\renewcommand{\thesection}{S\arabic{section}}
	\setcounter{section}{0}
	
	\noindent This appendix does the following:
	\begin{description}
		\item[\cref{sec:supp-figs}] presents supplemental figures.
		\item[\cref{sec:hal}] describes the highly adaptive lasso (\cref{sec:hal-details}). It includes theoretical undersmoothing conditions (\cref{sec:undersmooth}), a strengthened version of Lemma 1 from \cite{ertefaie_nonparametric_ipw_2023} to offer supremum norm control of remainder terms (\cref{sec:undersmooth-lemma}), and practical undersmoothing implementation details (\cref{sec:undersmooth-practice}). 
		\item[\cref{sec:supp-implementation}] offers practical implementation details including index set elicitation (\cref{sec:index-set}) and accounting for randomization ineligibility common in mobile health trials (\cref{sec:ineligible}).
		\item[\cref{sec:simulation-details}] describes the full longitudinal data generation process used in the numerical simulations.
		\item[\cref{sec:auxiliary}] provides auxiliary statements and proofs used throughout the rest of the appendix.
		\item[\cref{sec:proof-estimation}] supports theoretical arguments regarding regimen-response curve estimation, proving \cref{lemma:identification-expectation,lemma:dcar-expectation,theorem:uipw-asym,corollary:uipw-asymp-alpha}.
		\item[\cref{sec:proof-optimization}] supports theoretical arguments regarding policy optimization, proving \cref{theorem:opt-theta-asym,corollary:opt-theta-asym-alpha}.
		\item[\cref{sec:proof-alpha-tilting}] supports theoretical results regarding data-adaptive tilting, proving \cref{lemma:alpha-clip-op}.
		\item[\cref{sec:ipw-aipw-efficiency}] discusses how the efficiency gains of asymptotically efficient estimators beyond conventional IPW estimators may be reduced for stochastic versus deterministic interventions.
	\end{description}

	\section{Supplemental Figures}
	\label{sec:supp-figs}
	
	\begin{figure}[H]
		\centering
		\includegraphics[width=\linewidth]{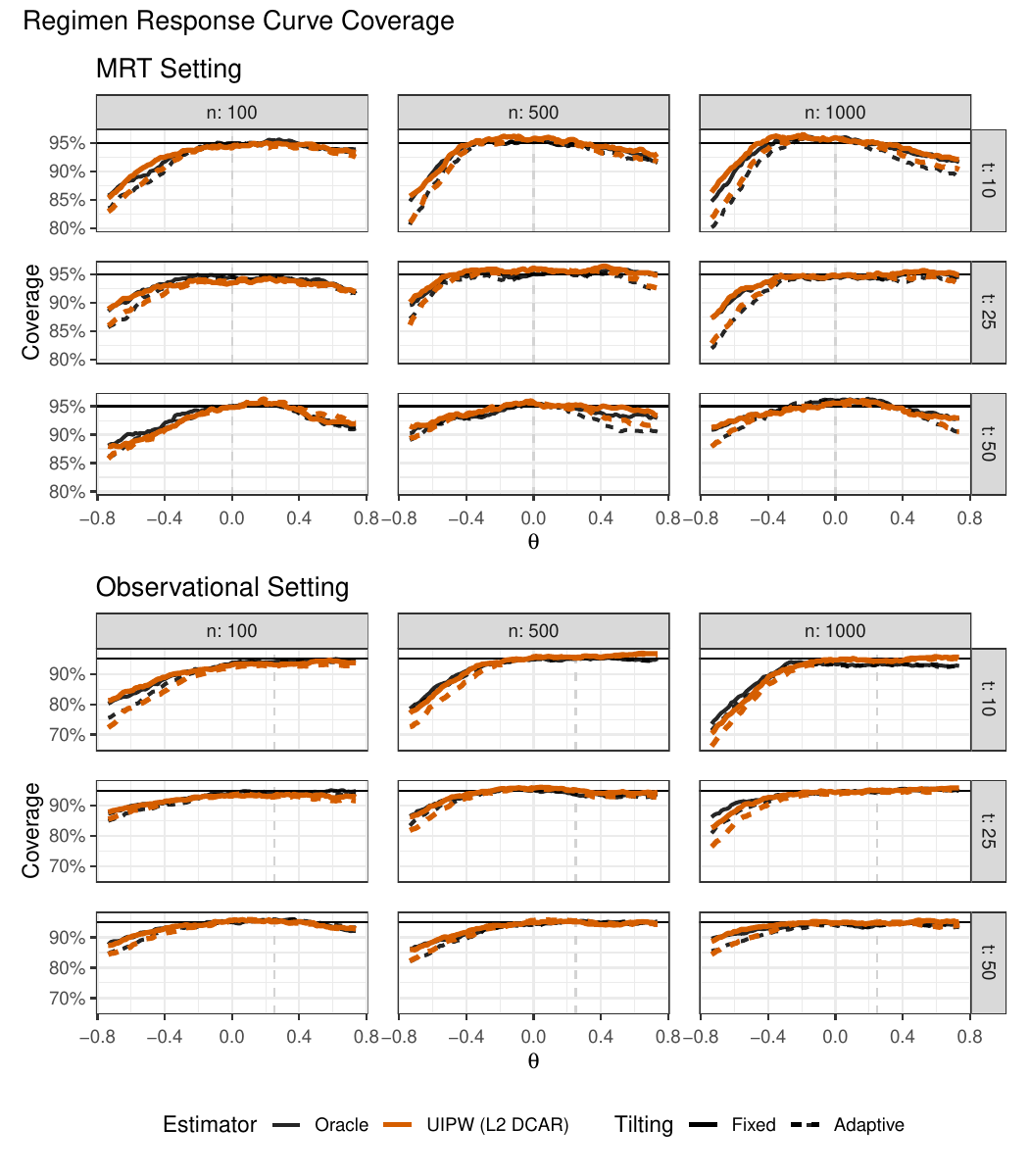}
		\caption{Pointwise regimen response curve coverage for both fixed and data-adaptive $\alpha_n$-tilting.}
		\label{fig:psi-coverage-with-alpha}
	\end{figure}
	
	\begin{figure}[H]
		\centering
		\includegraphics[width=\linewidth]{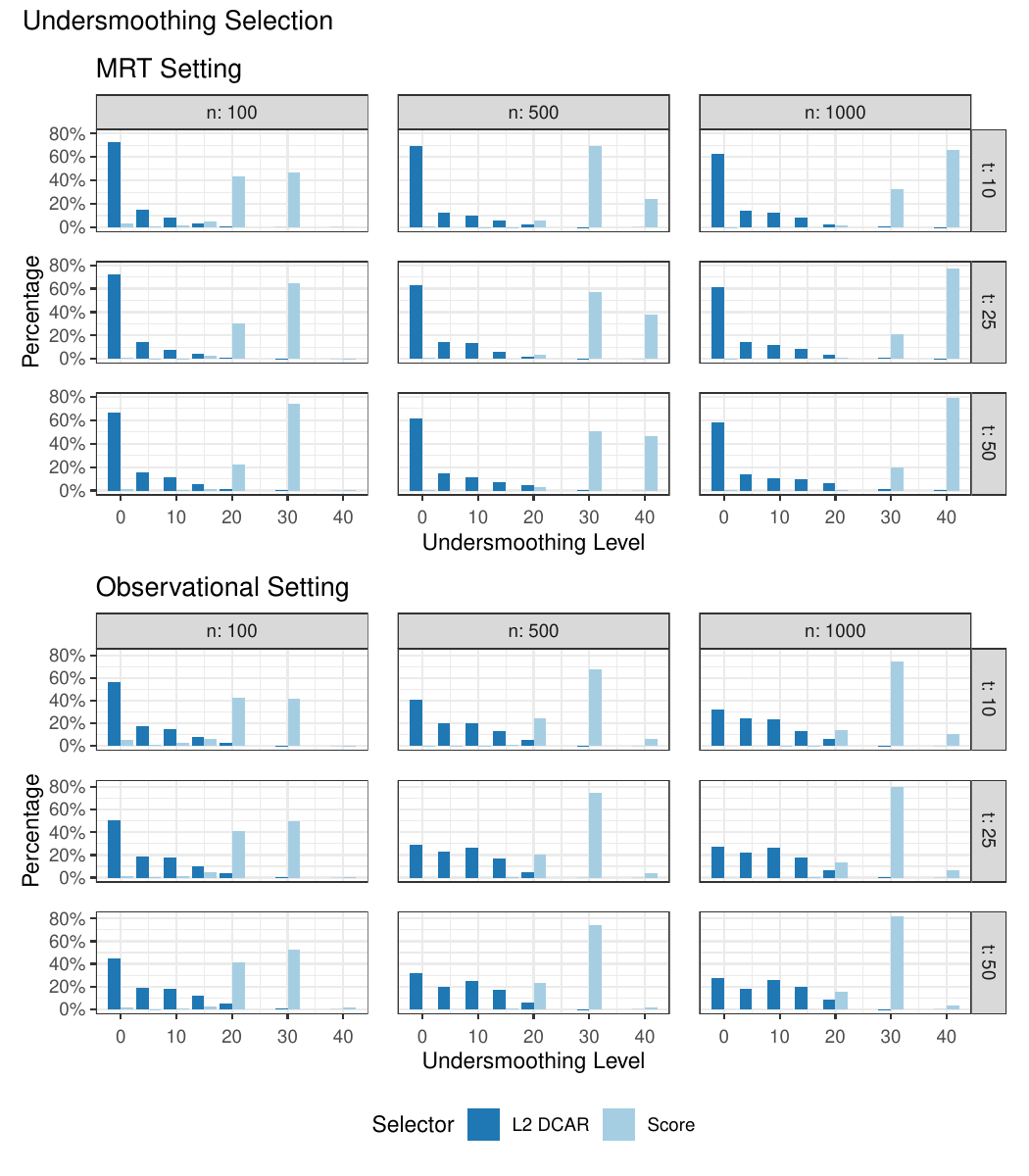}
		\caption{Undersmoothing selection. To align with \texttt{hal9001} parameterization in practice, undersmoothing was performed by manipulating  the cross-validated $L_1$-norm \emph{penalty}, $\lambda_\text{CV}$.
			Undersmoothed models were fit with penalty $\lambda_u=\lambda_\text{cv}\times (0.95)^u$.
		}
		\label{fig:us-selection}
	\end{figure}
	
	\begin{figure}[H]
		\centering
		\includegraphics[width=\linewidth]{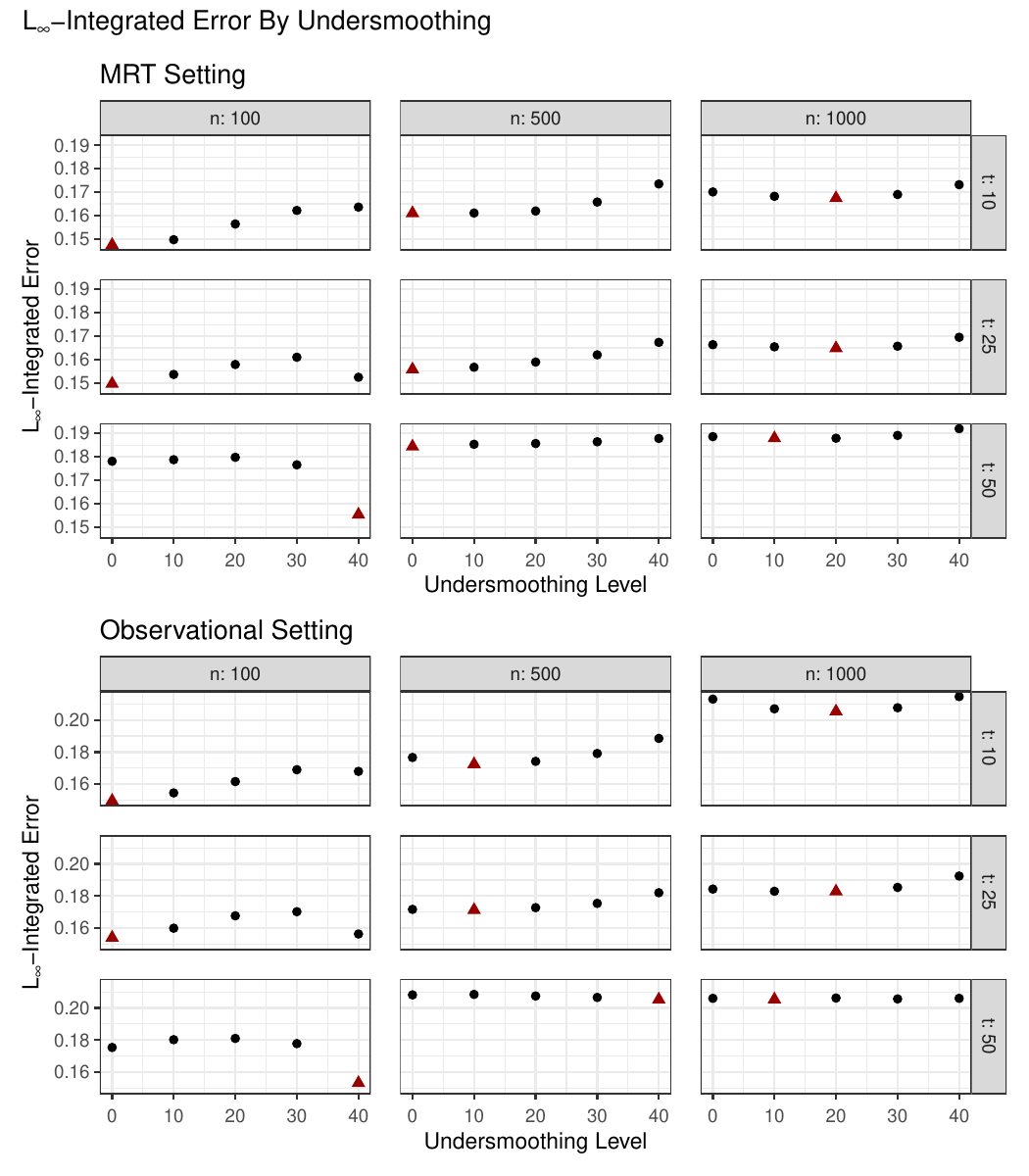}
		\caption{Average $L_\infty$-integrated absolute error ($\sup_{\theta\in\Theta}\lvert \psi_n(\theta)-\psi_0(\theta)\rvert$) as a function of the undersmoothing level. The level with the minimum integrated error is denoted by a triangle.
		}
		\label{fig:us-bias}
	\end{figure}

	\begin{figure}[H]
		\centering
		\includegraphics[width=\linewidth]{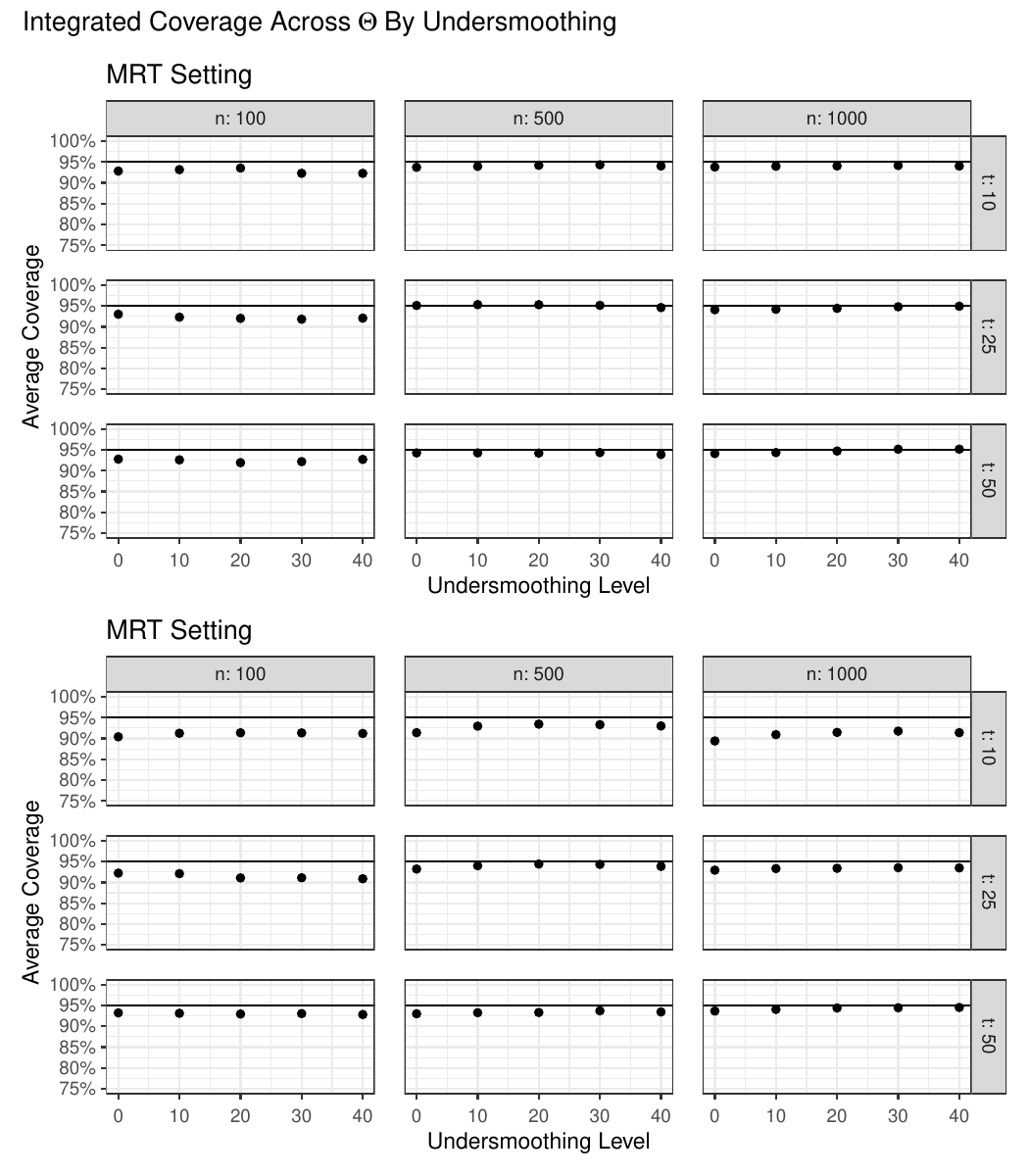}
		\caption{Average integrated pointwise coverage as a function of the undersmoothing level. 
		}
		\label{fig:us-coverage}
	\end{figure}
	
	\begin{figure}[H]
		\centering
		\includegraphics[width=\linewidth]{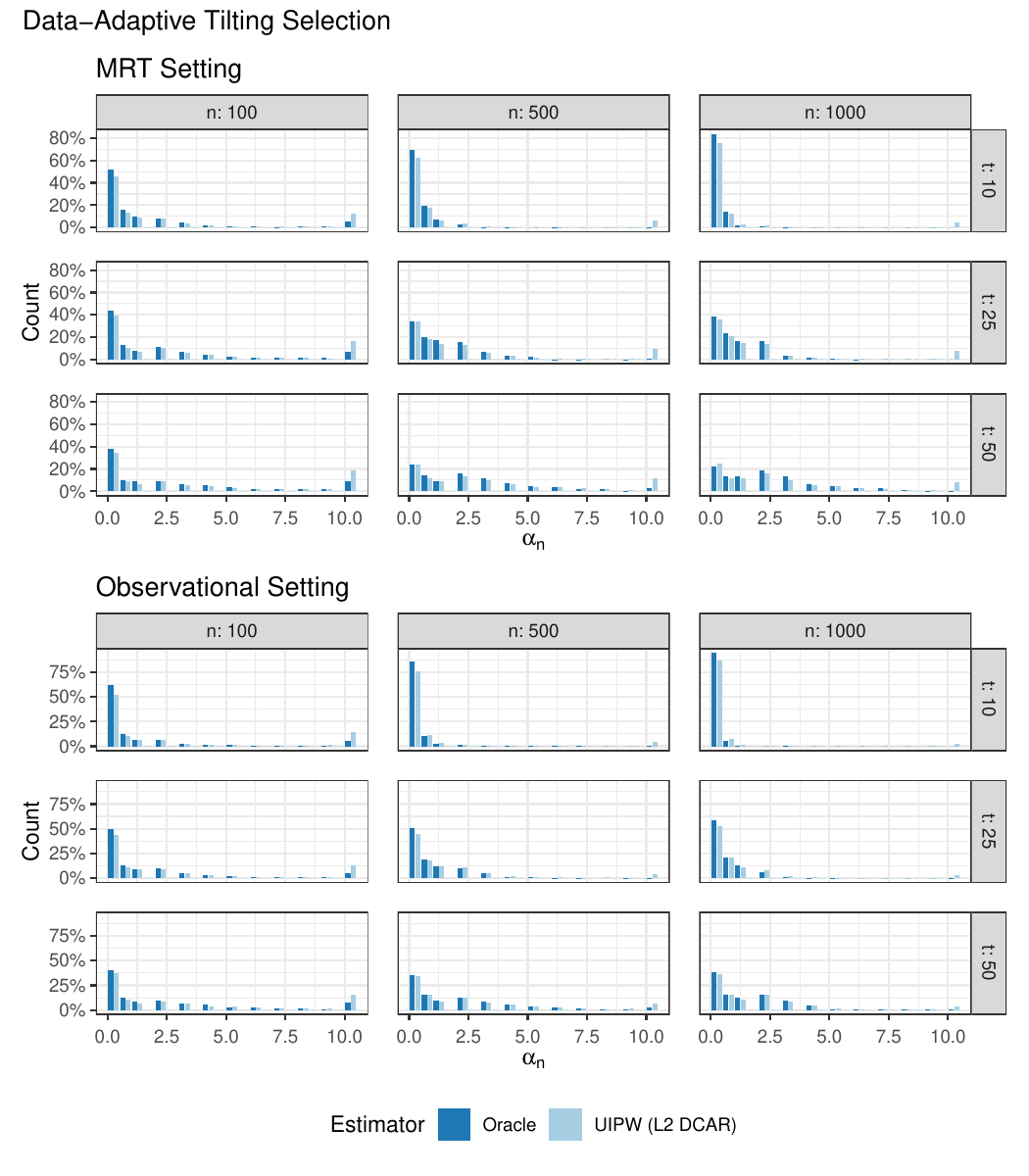}
		\caption{Data-adaptive $\alpha_n$ selection for the oracle IPW and UIPW ($L_2$-DCAR) estimators.
		}
		\label{fig:alpha-selection}
	\end{figure}

	\FloatBarrier
	
	\section{The Highly Adaptive Lasso and Undersmoothing}
	\label{sec:hal}
	
	\subsection{The Highly Adaptive Lasso}
	\label{sec:hal-details}
	Let $\mathbb{D}[0,\tau]$ denote the Banach space of real-valued c\`adl\`ag functions from $[0,\tau]\subset \mathbb{R}^d$ to $\mathbb{R}$.
	Here, $\tau<\infty$ is the upper bound of the support of all functions in this space.
	Consider some subset of the covariates: $s\subset \{1,\ldots,d\}$ and function $f\in\mathbb{D}[0,\tau]$, let the $s$th section of $f$ be given by $f_s(u)=f\{I(1\in s)u_1,\ldots, I(d\in s)u_d\}$,
	that is, varying the function along the variables in $s$ but setting all other variables to $0$.
	The sectional variational norm of $f$ is then given by aggregating all subsets' variational norms:
	$$\norm{f}_v^\star\coloneq \lvert f(0)\rvert+\sum_{s\subset\{1,\ldots,d\}}\int_{0_s}^{\tau_s}\lvert \dd f_s(u) \rvert\mathperiod$$
	
	Consider the propensity score at a given $t \in\{1,\ldots,T\}$, $\pi_{t,0}$.
	Given that $\pi_{t,0}$ is c\`adl\`ag with finite sectional variational norm, \cite{gill_inefficient_1995} shows that 
	\begin{align*}
		\logit \pi_{t,0}(w) &= \logit \pi_{t,0}(0) + \sum_{s\subset\{1,\ldots,d\}} \int_{0_s}^{w_s}\dd\logit \pi_{t,0,s}(u) \\
		&= \logit \pi_{t,0}(0) + \sum_{s\subset\{1,\ldots,d\}} \int_{0_s}^{\tau_s} I(u_s \le w_s)\dd\logit \pi_{t,0,s}(u)\mathperiod\stepcounter{equation}\tag{\theequation}\label{eqn:hal-integral}
	\end{align*}
	The HAL approximates the integral from \cref{eqn:hal-integral} by
	\begin{align*}
		\logit \pi_{t,0}(w;\beta) = \beta_{t,0} + \sum_{s\subset\{1,\ldots,d\}} \sum_{i=1}^n \beta_{t,s,i}\varphi_{t,s,i}(w_s)\mathcomma
	\end{align*}
	where $\varphi_{t,s,i}(w_s)=I(w_s\ge u_{t,s,i})$ and $u_{t,s,i}$ is the $i$th observed value of the $s$th section of the covariates at decision point $t$.
	The HAL estimator $\beta_{t,n,\lambda}$ is obtained by solving the constrained minimization problem: $\beta_{t,n,\lambda} = \arg\min \P_n L(\logit \pi_{t,n,\lambda_{n,t}})$ subject to $\lvert\beta_{t,0}\rvert + \sum_{s\subset\{1,\ldots,d\}} \sum_{i=1}^n \lvert\beta_{t,s,i}\rvert<\lambda$.
	
	\subsection{Undersmoothing Details}
	\label{sec:undersmooth}
	
	Suppose that the set of $L_1$-norms for each HAL fit, $\lambda_n=\{\lambda_{1,n},\ldots,\lambda_{T,n}\}$, is chosen such that, for all $t=1,\ldots,T$,
	\begin{equation} \label{eqn:hal-undersmooth-condition}
		\min_{(s,j)\in\mathcal{J}_{t,n}} \norm{\P_n \frac{\dd}{\dd \logit \pi_{t,n,\lambda_{n,t}}} L(\logit \pi_{t,n,\lambda_{n,t}})(\varphi_{t,s,j})}=o_p(n^{-1/2})\mathcomma
	\end{equation}
	where $L(\cdot)$ is the log-likelihood loss and $\mathcal{J}_{t,n}$ is the set of indices for the basis functions where $\beta_{t,s,j}\neq0$.
	Here, $\dd/\dd \logit \pi_{t,n,\lambda_{n,t}}\{ L(\logit \pi_{t,n,\lambda_{n,t}})(\varphi_{t,s,j})\}$ is $\dd/\dd\epsilon\{L(\logit \pi_{t,n,\lambda_{n,t}} + \epsilon\varphi_{t,sj}\}$, which denotes the directional derivative of the loss along the path $\logit \pi_{t,n,\lambda_{n,t}}^\epsilon = \logit \pi_{t,n,\lambda_{n,t}} + \epsilon \varphi_{s,j}$. Under the log-likelihood loss, the derivative is
	$$\frac{\dd}{\dd\epsilon} L(\logit \pi_{t,n,\lambda_{n,t}})(\varphi_{t,s,j})\bigg\rvert _{\epsilon=0} = \varphi_{t,s,j}\{A_t-\pi_{t,n,\lambda_{n,t}}(1\mid H_t)\}\mathcomma$$
	the corresponding score function.
	Hence, an equivalent statement of the supposition is that
	\begin{equation} \label{eqn:hal-undersmooth-condition2}
		\min_{(s,j)\in\mathcal{J}_{t,n}} \norm{\P_n \varphi_{t,s,j}\{A_t-\pi_{t,n,\lambda_{n,t}}(1\mid H_t)\}}=o_p(n^{-1/2})\mathperiod
	\end{equation}
	Although unrestricted maximum likelihood estimation exactly solves these score equations, the $L_1$ restriction generally prevents exact solutions.
	The assumption that this term is $o_p(n^{-1/2})$ is hence equivalent to assuming that the $L_1$ restriction is weakened until one of the score equations is solved to a precision of $o_p(n^{-1/2})$.
	We say that the HAL is ``sufficiently undersmoothed'' if \cref{eqn:hal-undersmooth-condition} holds for each $t=1,\ldots,T$.
	
	We make the following assumptions to allow the undersmoothed HAL basis to sufficiently approximate the leading terms in the DCAR remainders (thereby satisfying the premise of \cref{lemma:d-bounded-uniform} throughout various applications).
	These are stated as pointwise assumptions, but we will often assume they hold uniformly in $(\theta,\alpha)$ over some subset of $\Theta \times [0,\infty)$.
	
	\begin{assumption}
		\label{assumption:cadlag}
		Let $\mu_{t,0}^{\theta,\alpha}$ and $\pi_{t,0}$ be c\`adl\`ag with finite sectional variational norm for all $t\in\{1,\ldots,T\}$.
	\end{assumption}
	
	\begin{assumption}
		\label{assumption:hal-space}
		Let 
		$$f_{t,0}^{\theta,\alpha}(H_t)\coloneq \left\{\prod_{i<t} W_{i,0}^{\theta,\alpha}(H_i) \right\}(t/T)^\alpha\Bigl\{\frac{\mu_{t,0}^{\theta,\alpha}(1,H_t)q_t^\theta(1,S_t)}{\pi_{t,0}(1\mid H_t)} - \frac{\mu_{t,0}^{\theta,\alpha}(0,H_t)q_t^\theta(0,S_t)}{\pi_{t,0}(0\mid H_t)}\Bigr\}\mathcomma$$
		and let $\tilde{f}_{t,0}^{\theta,\alpha}$ denote the projection of $f_{t,0}^{\theta,\alpha}(H_t)$  onto the linear span of basis functions $\varphi_{t,s,j}$ in $L^2(\P_0)$ for $\varphi_{t,s,j}$ satisfying \cref{eqn:hal-undersmooth-condition}.
		Then, assume that $\|f_{t,0}^{\theta,\alpha}-\tilde{f}_{t,0}^{\theta,\alpha}\|_{2,\P_0}=O_p(n^{-1/4})$ for all $t\in\{1,\ldots,T\}$.
	\end{assumption}
	
	\begin{remark}
		\cref{assumption:hal-space} requires that the leading remainder multiplier $f_{t,0}^{\theta,\alpha}(H_t)$ be well-approximated by the HAL basis used for $\pi_{t,n}(1\mid H_t)$. 
		In particular, this effectively requires that the variables needed to represent $\mu^{\theta,\alpha}_{t,0}(A_t,H_t)$ (and hence $f^{\theta,\alpha}_{t,0}$) are included among the variables used to model $\pi_{t,n}(1\mid H_t)$. 
		Thus, if one imposes a Markov restriction on $\pi_{t,n}$ (e.g., $\pi_{t,n}$ depends only on the current state $S_t$), then \cref{assumption:hal-space} is most plausible when the corresponding $\mu^{\theta,\alpha}_{t,0}(A_t,H_t)$ can also be represented as depending only on $S_t$.
		If $\mu^{\theta,\alpha}_{t,0}(A_t,H_t)$ is believed to depend on additional variables, these should be included in the model for $\pi_{t,n}$.
	\end{remark}
	
	The following assumption is used to ensure that $\nabla_\theta\psi_n(\theta;\alpha)$ is also regular and asymptotically linear, which is needed for the asymptotic properties of $\theta_n^\star(\alpha)$.
	
	\begin{assumption}
		\label{assumption:hal-space-deriv}
		Let 
		\begin{align*}
			&g_{t,0}^{\theta,\alpha}(H_t) \\
			&\coloneq
			\Bigl(\prod_{i<t} W_{i,0}^{\theta,\alpha}\Bigr)
			(t/T)^\alpha
			\Biggl[
			\Bigl(
			\sum_{j<t}
			\nabla_\theta \log W_{j,0}^{\theta,\alpha}
			\Bigr)
			\Biggl\{
			\frac{\mu_{t,0}^{\theta,\alpha}(1,H_t)q_t^\theta(1,S_t)}{\pi_{t,0}(1\mid H_t)}
			-
			\frac{\mu_{t,0}^{\theta,\alpha}(0,H_t)q_t^\theta(0,S_t)}{\pi_{t,0}(0\mid H_t)}
			\Biggr\}
			\\
			&\qquad\qquad
			+
			\frac{
				\{\nabla_\theta\mu_{t,0}^{\theta,\alpha}(1,H_t)\}
				q_t^\theta(1,S_t)
				+
				\mu_{t,0}^{\theta,\alpha}(1,H_t)
				\{\nabla_\theta q_t^\theta(1,S_t)\}
			}{
				\pi_{t,0}(1\mid H_t)
			}
			\\
			&\qquad\qquad
			-
			\frac{
				\{\nabla_\theta\mu_{t,0}^{\theta,\alpha}(0,H_t)\}
				q_t^\theta(0,S_t)
				+
				\mu_{t,0}^{\theta,\alpha}(0,H_t)
				\{\nabla_\theta q_t^\theta(0,S_t)\}
			}{
				\pi_{t,0}(0\mid H_t)
			}
			\Biggr]\mathcomma
		\end{align*}
		and let $\tilde g_{t,0}^{\theta,\alpha}$ denote its coordinatewise projection  onto the linear span of basis functions $\varphi_{t,s,j}$ in $L^2(\P_0)$ for $\varphi_{t,s,j}$ satisfying \cref{eqn:hal-undersmooth-condition}.
		Assume that $\lVert g_{t,0}^{\theta,\alpha}-\tilde{g}_{t,0}^{\theta,\alpha}\rVert_{2,\P_0}=O_p(n^{-1/4})$ for all $t\in\{1,\ldots,T\}$.
	\end{assumption}
	
	\subsection{Formal Lemma and Proof}
	\label{sec:undersmooth-lemma}
	
	To control remainders uniformly across $\Theta$, we require a strengthened version of \cite{ertefaie_nonparametric_ipw_2023} Lemma 1.
	The original lemma provides pointwise control of terms of the form $\P_n \{ f \times (A_t-\pi_n)\}$ when $\pi_{t,n}$ is a sufficiently undersmoothed HAL estimate of the propensity score.
	We extend the result to offer uniform control over a class of leading functions $f_\theta:\theta\in\Theta$.
	This proof largely follows the argument provided to bound $\P_n D(\tilde{f},\pi_n)$ pointwise in \cite{ertefaie_nonparametric_ipw_2023}; the key differences are that the $\kappa_\theta$ is now bounded uniformly, and that obtaining a bound on $D(f,\pi_n)$ from the bound on $D(\tilde{f},\pi_n)$ requires a stronger empirical process argument that leverages uniform convergence of $\tilde f_\theta$ to $f_\theta$.
	
	\begin{lemma}
		\label{lemma:d-bounded-uniform}
		Fix $t\in\{1,\ldots,T\}$ and let $\pi_n\coloneq\pi_{t,n}$ be a HAL estimator for $\pi_{0}\coloneq\pi_{t,0}$ with $L_1$-norm bound $\lambda_n\coloneq\lambda_{t,n}$ chosen such that \cref{eqn:hal-undersmooth-condition} holds.
		Let $D(f_\theta,\pi_n)=f_\theta\cdot(A-\pi_n)$
		where the function class $\{f_\theta:\theta\in\Theta\}$ is c\`adl\`ag with uniformly bounded sectional variational norm.
		Let $\tilde{f}_\theta$ be the projection of $f_\theta$ onto the linear span of basis functions $\varphi_{s,j}\coloneq\varphi_{t,s,j}$ satisfying \cref{eqn:hal-undersmooth-condition}.
		Then, $\sup_{\theta\in\Theta }\lvert\P_n D(\tilde{f}^\theta, \pi_n)\rvert=o_p(n^{-1/2})$.
		Moreover, if $\sup_{\theta\in\Theta}\norm*{f_\theta-\tilde{f}_\theta}_{2,\P_0}=O_p(n^{-1/4})$, then 
		$\sup_{\theta\in\Theta}\lvert\P_n D(f^\theta, \pi_n)\rvert=o_p(n^{-1/2})$.
	\end{lemma}
	
	\begin{proof}[Proof of \cref{lemma:d-bounded-uniform}]
		Let $\pi^\dag\coloneq\logit \pi$.
		Consider a path on the logit scale indexed by $\epsilon$ for some bounded vector $h$, $\pi_\epsilon^\dag = \pi_n^\dag(1+\epsilon \sum_{(s,j)}h(s,j)\beta_{n,s,j}\varphi_{s,j})$; the corresponding score is
		\begin{equation*}
			S_h(\pi_n) = (A-\pi_{n,\lambda_n})\sum_{(s,j)}h(s,j)\beta_{n,s,j}\varphi_{s,j}\mathperiod
		\end{equation*}
		Then, for sufficiently small $\epsilon$,
		\begin{align*}
			\sum_{(s,j)} \abs{1+\epsilon h(s,j)\}\beta_{n,s,j}} &= \sum_{(s,j)} \{1 + \epsilon h(s,j)\} \abs{\beta_{n,s,j}} \\
			&= \sum_{(s,j)} \abs{\beta_{n,s,j}} + \epsilon r(h,\pi_n)\mathcomma
		\end{align*}
		where $r(h,\pi_n)=\sum_{(s,j)}h(s,j)\abs{\beta_{n,s,j}}$.
		It follows that if $h$ satisfies $r(h,\pi_n)=0$, then $\P_n S_h(\pi_n)=0$.
		
		Let $D(f_\theta,\pi_n)=f_\theta \cdot(A-\pi_n)$ and let $\tilde{f}_\theta$ be the projection of $f_\theta$ onto the basis functions satisfying \cref{eqn:hal-undersmooth-condition}.
		Then for any $\theta\in\Theta$, $D(\tilde{f}_\theta,\pi_n)\in\{S_h(\pi_n):\norm{h}_\infty<\infty\}$ and there exists some $h_\theta^\star$ with $\norm{h_\theta^\star}_\infty < \infty$ such that $D(\tilde{f}_\theta,\pi_n)=S_{h^\star_\theta}(\pi_n)$.
		We note that $r(h_\theta^\star,\pi_n)$ may be nonzero; to correct this, construct $h_\theta$ such that $h_\theta(s,j)=h_\theta^\star(s,j)$ for $(s,j)\neq(s^\star,j^\star)$ and
		\begin{equation*}
			h_\theta(s^\star,j^\star)\coloneq  -\abs{\beta_{n,s^\star,j^\star}}^{-1}\sum_{(s,j)\neq(s^\star,j^\star)}h^\star_\theta(s,j)\abs{\beta_{n,s,j}}\mathcomma
		\end{equation*}
		where $(s^\star,j^\star)$ are the arguments of the minima of the left hand side of \cref{eqn:hal-undersmooth-condition}.
		It follows that $r(h_\theta,\pi_n)=0$ and $\P_n S_{h_\theta}(\pi_n)=0$ for all $\theta\in\Theta$.
		
		Now, consider $\P_n S_{h_\theta}(\pi_n)-\P_n D\{\tilde{f}_\theta(O),\pi_n\}$. 
		We have
		\begin{align*}
			\P_n S_{h_\theta}(\pi_n)-\P_n D\{\tilde{f}_\theta(O),\pi_n\} &= \kappa_\theta(s^\star,j^\star) \P_n\left\{ \frac{\dd}{\dd \pi_n^\dag} L(\pi_n^\dag )(\varphi_{t,s^\star,j^\star})\right\}\mathcomma
		\end{align*}
		where
		\begin{align*}
			\kappa_\theta(s^\star,j^\star) &= \{h_\theta(s^\star,j^\star)-h_\theta^\star(s^\star,j^\star)\}\beta_{n,s^\star,j^\star} \\
			&=  -\abs{\beta_{n,s^\star,j^\star}}^{-1}\beta_{n,s^\star,j^\star}\sum_{(s,j)\neq(s^\star,j^\star)}h^\star_\theta(s,j)\abs{\beta_{n,s,j}} - h_\theta^\star(s^\star,j^\star)\}\beta_{n,s^\star,j^\star}\mathperiod 
		\end{align*}
		From the second equality, $\abs{\kappa_\theta(s^\star,j^\star)}\le \sum_{(s,j)} \abs{h_\theta^\star(s,j)\beta_{n,s,j}}$.
		Given that $f_\theta$ has uniformly bounded sectional variational norm, the $L_1$-norm of coefficients approximating $f$ will be finite which implies that  $\sup_{\theta\in\Theta}\sum_{(s,j)}\abs{h^\star_\theta(s,j)\beta_{n,s,j}}=O_p(1)$ and $\sup_{\theta\in\Theta}\abs{\kappa_\theta(s^\star,j^\star)}=O_p(1)$.
		Thus
		\begin{align*}
			\sup_{\theta\in\Theta}\abs{\P_n S_{h_\theta}(\pi_n)-\P_n D\{\tilde{f}_\theta(O),\pi_n\}} &= O_p\left[ \P_n\left\{ \frac{\dd}{\dd \pi_n^\dag} L(\pi_n^\dag )(\varphi_{t,s^\star,j^\star})\right\}\right] \\
			&= o_p(n^{-1/2})\mathcomma
		\end{align*}
		with the final equality following from the undersmoothing criteria (\cref{eqn:hal-undersmooth-condition}).
		Given the fact that $\P_n S_{h_\theta}(\pi_n)=0$ for all $\theta$, it follows that $\lVert \P_n D\{\tilde{f}_\theta(O),\pi_n\}\rVert_\Theta=o_p(n^{-1/2})$ as well.
		
		We conclude by uniformly bounding the quantity of interest, $\lVert \P_n D\{{f}_\theta(O),\pi_n\}\rVert_\Theta=o_p(n^{-1/2})$.
		We have
		\begin{align*}
			&\P_n D\{\tilde{f}_\theta(O),\pi_n\}-\P_n D\{{f}_\theta(O),\pi_n\} \\
			&= \P_n\left[ (\tilde{f}_\theta-f_\theta) \left\{ \frac{\dd}{\dd \pi_n^\dag} L(\pi_n^\dag )\right\}\right] \\
			&= \P_n\left[(\tilde{f}_\theta-f_\theta)\left\{ \frac{\dd}{\dd \pi_0^\dag} L(\pi_0^\dag )\right\}\right] + \P_n\left\{(\tilde{f}_\theta-f_\theta)(\pi_0-\pi_n)\right\} \\
			&= (\P_n-\P_0)\left[(\tilde{f}_\theta-f_\theta)\left\{ \frac{\dd}{\dd \pi_0^\dag} L(\pi_0^\dag )\right\}\right] + (\P_n-\P_0)\left\{(\tilde{f}_\theta-f_\theta)(\pi_0-\pi_n)\right\} \\
			&\qquad {} + \P_0\left\{(\tilde{f}_\theta-f_\theta)(\pi_0-\pi_n)\right\}\mathperiod
		\end{align*}
		Each of these terms is bounded uniformly.
		
		For the first term, let $c(O)\coloneq \frac{\dd}{\dd \pi_0^\dag} L(\pi_0^\dag)=A_t-\pi_0$, $g_\theta(O)\coloneq \{\tilde{f}_\theta(O)-f_\theta(O)\}c(O)$, and $\mathcal{G}\coloneq \{g_\theta:\theta\in\Theta\}$.
		We obtain a bound for $\norm{\G_n}_\mathcal{G}$ using an empirical-process argument.
		Let $G\coloneq M \abs{c_\theta(O)}$ where $M\coloneq \sup_{\theta\in\Theta}\abs*{\tilde{f}_\theta(O)-f_\theta(O)}$ and $\abs{c(O)}\le1$.
		Moreover, as $f_\theta$ and $\tilde{f}_\theta$ are c\`adl\`ag with bounded sectional variational norm, $M=O_p(1)$.
		Hence, $\abs{g_\theta(O)}\le M G(O)=O_p(1)$ and the $\mathcal{G}$ has an $O_p(1)$ envelope.
		Additionally, $\norm{G}_{2,\P_0}\le \sup_{\theta\in\Theta}\norm*{\tilde{f}_\theta-f_\theta}_{2,\P_0}\le O_p(n^{-1/4})$ by assumption.
		Finally, $\mathcal{G}$ is Donsker so 
		$\sup_{\theta\in\Theta}\abs{(\P_n-\P_0)g_\theta}=o_p(n^{-1/2})$ by \cref{corollary:maximal-inequality-shortcut}.
		
		Then,
		\begin{align*}
			\sup_{\theta\in\Theta}\abs{(\P_n-\P_0)\left\{(\tilde{f}_\theta-f_\theta)(\pi_0-\pi_n)\right\}} &\le \sup_{\theta\in\Theta}\abs*{(\P_n-\P_0)(\tilde{f}_\theta-f_\theta)} \abs{(\P_n-\P_0)(\pi_0-\pi_n)} \\
			&= O_p(n^{-1/4})o_p(n^{-1/4}) \\
			&= o_p(n^{-1/2})\mathcomma
		\end{align*}
		under the assumption that $\sup_{\theta\in\Theta}\norm*{f_\theta-\tilde{f}_\theta}_{2,\P_0}=O_p(n^{-1/4})$ and the HAL convergence rate for $\pi_n$.
		
		Finally,
		\begin{align*}
			\sup_{\theta\in\Theta}\abs{\P_n\left\{(\tilde{f}_\theta-f_\theta)(\pi_0-\pi_n)\right\}} &\le 
			\sup_{\theta\in\Theta}\norm*{\tilde{f}_\theta-f_\theta}_{2,\P_0} \norm{\pi_0-\pi_n}_{2,\P_0} \\
			&= O_p(n^{-1/4}) o_p(n^{-1/4}) \\
			&= o_p(n^{-1/2})\mathcomma
		\end{align*}
		again under the assumption that $\sup_{\theta\in\Theta}\norm*{f_\theta-\tilde{f}_\theta}_{2,\P_0}=O_p(n^{-1/4})$ and the convergence rate for $\pi_n$.
		
		Therefore we have 
		$$\sup_{\theta\in\Theta}\abs{\P_n D\{\tilde{f}_\theta(O),\pi_n\}-\P_n D\{{f}_\theta(O),\pi_n\}} = o_p(n^{-1/2})\mathperiod$$ 
		Hence, $\sup_{\theta\in\Theta}\abs{\P_n D\{{f}_\theta(O),\pi_n\}} = o_p(n^{-1/2})$, which completes the proof.
	\end{proof}
	
	\subsection{Undersmoothing in Practice}
	\label{sec:undersmooth-practice}
	
	The original undersmoothing criteria (\cref{eqn:hal-undersmooth-condition}) ensures that the penalty parameters for each propensity score model, $\lambda_{1,n},\ldots,\lambda_{T,n}$, are chosen such that the bias term in the asymptotic representation of $\sqrt{n}\{\psi_n(\theta;\alpha)-\psi_0(\theta;\alpha)\}$ is $o_p(n^{-1/2})$.
	However, this does not directly imply a procedure to select $\lambda_n$.
	
	Given that the goal of the undersmoothing is to ensure that the bias term satisfies $\P_n\{ \Dcart(\mu_0^{\theta,\alpha},\pi_n;\theta,\alpha)\}=o_p(n^{-1/2})$ for all $t=1,\ldots,T$; one approach is to directly minimize an approximation of this quantity.
	\cite{ertefaie_nonparametric_estimation_2023,ertefaie_nonparametric_ipw_2023,pham_nonparametric_2025} all offer selectors that directly approximate the DCAR term by nonparametrically estimating all nuisance functions and minimize a cross-validated estimate.
	However, this is computationally prohibitive in our longitudinal policy-index setting because the augmentation term involves the forward regressions $\{\mu_{t,0}^{\theta,\alpha}(A_t, H_t):t<T,\theta\in\Theta\}$, and hence would require estimating a large family of nuisance functions jointly indexed by time and policy.
	Instead, we use a pragmatic DCAR-approximated selector that targets the same remainder but replaces $\mu_{t,0}^{\theta,\alpha}$ with a low-dimensional working approximation fit only for tuning.
	We target performance for the untilted estimator and therefore fix $\alpha=0$ henceforth.
	For each fold $v\in1\ldots,V$ and time $t$, we construct pseudo-outcomes on the training sample of the form $Z_{i,t}^{\theta,\alpha}=\{\prod_{j>t}W_{i,j,\text{ref}}^{\theta,\alpha}\}Y_i$, where $W_{i,j,\text{ref}}^{\theta,\alpha}$ uses a reference propensity score $\pi_{t,j,\text{ref}}$ that is equal to either the cross-validated HAL estimate (i.e., no undersmoothing) or the known propensity score in MRT settings.
	We then fit working regressions $\mu_{t,n}^{-(v)}(a,h;\theta)$ by pooling pseudo-outcomes across a modest training subset $\Theta_{\text{train}}\subset\Theta$ and modeling $Z_{t}^{\theta,\alpha}$ as a function of $(H_t,A_t)$ together with a smooth basis $g(\theta)$ (e.g., a spline basis) and optional interactions.
	Then, given $\mu_{t,n}^{-(v)}$, for each undersmoothing level $\lambda$ we evaluate
	\begin{multline*}
		R^{(v)}(\lambda;\theta) \coloneq \P_{n,v}^1 \Bigl\{\prod_{i<t} W_{i,n,\lambda}^{\theta,\alpha}\Bigr\}
		(t/T)^\alpha\ \\
		\times \sum_{a\in\{0,1\}}\mu_{t,n}^{-(v)}(a, H_t;\theta)\, q_t^\theta(a,S_t)
		\frac{I(A_t=a)-\pi_{t,n,\lambda}^{(-v)}(a\mid H_t)}{\pi_{t,n,\lambda}^{(-v)}(a\mid H_t)}\mathcomma
	\end{multline*}
	where $\P_{n,v}^1$ denotes the empirical measure on the testing data from the $v$th fold.
	We then select $\lambda_{t,n}$ by minimizing either an $L_2$-integrated criterion 
	\begin{equation*}
		\lambda_{t,n,L_2} \coloneq \arg \min_{\lambda} V^{-1}\sum_{v=1}^V \lvert \Theta_{\text{test}}\rvert^{-1}\sum_{\theta\in\Theta_{\text{test}}} \{R^{(v)}(\lambda;\theta)\}^2\mathcomma    
	\end{equation*}
	or a uniform criterion 
	\begin{equation*}
		\lambda_{t,n,L_\infty} \coloneq \arg \min_{\lambda} V^{-1}\sum_{v=1}^V \max_{\theta\in\Theta_{\text{test}}} \lvert R^{(v)}(\lambda;\theta) \rvert\mathcomma
	\end{equation*}
	where $\Theta_{\text{test}}\coloneq\Theta\setminus \Theta_\text{train}$.
	
	Alternatively, we propose the alternate selection criteria based solely on the score function for the propensity score model:
	\begin{equation}
		\label{eqn:undersmooth-score}
		\lambda_{t,n,\text{score}}=\arg\min_{\lambda} V^{-1}\sum_{v=1}^V 
		\frac{1}{\norm{\beta_{n,\lambda,v}}_{1}}\sum_{(s,j)\in\mathcal{J}_n^{(v)}(\lambda)}  \lvert \P_{n,v}^1 \tilde{S}_{t,s,j}(\phi, \pi_{t,n,\lambda,v}) \rvert\mathcomma  
	\end{equation}
	where
	$\mathcal{J}_n^{(v)}(\lambda)$ is a set of indices for the active basis functions such that $\beta_{t,s,j,n,\lambda,v} \neq 0$, and 
	$$\tilde{S}_{t,s,j}(\varphi, \pi_{t,n,\lambda,v})=\varphi_{t,s,j}( H_t)\{A_t-\pi_{t,n,\lambda,v}(1\mid  H_t)\}\{\pi_{t,n,\lambda,v}(1\mid  H_t)\}^{-1}$$ 
	is the component of the score equation corresponding to $\beta_{t,s,j}$ divided by the estimated propensity score.
	This criterion only requires estimating the $T$ propensity scores used in the general estimation procedure.
	
	Finally, to ensure that the undersmoothed HAL fit still converges at the $o_p(n^{-1/4})$ rate, as needed for our asymptotic theory, we impose a cap such that the basis does not grow too quickly, as proposed by \cite{pham_nonparametric_2025}.
	Let $\tilde\lambda_{t,n}\coloneq \max \{\lambda :  V^{-1}\sum_{v=1}^V \lvert \mathcal{J}_n^{(v)}(\lambda)\rvert < n^{1/2}\}$ be the largest $L_1$ norm that yields an active basis with fewer than $n^{1/2}$ functions.
	We then select $\min\{\lambda_{t,n,\cdot}, \tilde\lambda_{t,n}\}$ for $\cdot\in\{L_2,L_\infty,\text{score}\}$.
	
	\section{Practical Implementation Considerations}
	\label{sec:supp-implementation}
	
	\subsection{Index Set Elicitation}
	\label{sec:index-set}
	
	Many policy-learning approaches enforce regularity implicitly through generic parameter-norm constraints or penalties (e.g., $L_1$ or $L_2$ regularization), which can be difficult to interpret in terms of scientifically meaningful limits on intervention frequency. 
	By contrast, we impose regularity by directly bounding treatment probabilities, which yields an immediately interpretable constraint on intervention frequency.
	
	To ensure that the explored policies are genuinely stochastic at all decision points and covariate values, we fix user-specified constants $0<\munderbar q < \bar q < 1$ and restrict attention to policies satisfying
	\begin{equation}
		\label{eqn:q-bound}
		q_t^\theta(1, S_t) \in [\munderbar q, \bar q]\quad \text{for all}\quad S_t \in \mathcal{S}_t\mathperiod  
	\end{equation}
	This constraint directly encodes investigator-specified minimum and maximum tolerable treatment probabilities (e.g., ``the policy never treats with probability below $\munderbar q$ or above $\bar q$'').
	
	Suppose that the intervention is characterized by some monotone link function $g$ and linear predictor $b(S_t)$: $q_t^\theta(1,S_t)=g\{b(S_t)\trans\theta\}$.
	Then, \cref{eqn:q-bound} is equivalent to a uniform bound on the linear predictor. 
	Let $\ell \coloneq g^{-1}(\munderbar q)$ and $u \coloneq g^{-1}(\bar q)$.
	Then \cref{eqn:q-bound} holds if and only if
	\begin{equation*}
		\label{eqn:q-bound-lin-predictor}
		\ell \le b(s)\trans\theta \le u \quad \text{for all}\quad s \in \mathcal{S}\mathperiod
	\end{equation*}
	Under this framework, given probability bounds $\munderbar q$ and $\bar q$, define
	\begin{equation*}
		\label{eqn:def-Theta}
		\Theta \coloneq \left\{\theta \in \mathbb R^p : \inf_{s\in\mathcal S} b(s)\trans\theta \ge \ell \quad \text{and} \quad \sup_{s\in\mathcal S} \le u \right\}\mathperiod
	\end{equation*}
	This set is closed and convex; under mild conditions on the basis, it is also compact. 
	
	In many applications, $q_t^\theta$ will follow a form based on a generalized linear model with a linear basis.
	In this setting, $\Theta$ can be constructed analytically.
	Suppose for exposition, without loss of generality, that $b(s) = (1,s_1,\ldots,s_d)\trans$, that is, including all covariate main effects but no interactions or higher-order terms.
	Let $s_j\in[a_j,b_j]$ denote the bounds for each component.
	Moreover, suppose that $\mathcal{S} = \prod_{j=1}^d [a_j,b_j]$ is a hyper-rectangle. 
	Violations of this supposition may result in overly conservative $\Theta$; however, the desired bound on $q^\theta$ is still attained.
	Expressing $\theta=(\theta_0,\theta_1,\ldots,\theta_d)\trans$, $M_0\coloneq 1$ and $M_j \coloneq \max\{\lvert a_j\rvert ,\lvert b_j\rvert\}$ for $j=1,\ldots,d$, we have, for any $\theta$,
	\begin{align*}
		\sup_{s \in \mathcal S} b(s)\trans\theta &= \sum_{0=1}^d \lvert\theta_j\rvert M_j \mathcomma\quad \text{and}\\
		\inf_{s \in \mathcal S} b(s)\trans\theta &= \sum_{0=1}^d \lvert\theta_j\rvert M_j \mathperiod
	\end{align*}
	Then, the index set can be written as the weighted $\ell_1$-ball:
	\begin{equation*}
		\Theta = \Bigl\{\theta \in \mathbb R^{d+1} : \sum_{j=0}^d \lvert\theta_j\rvert M_j \le \min(-\ell, u)  \Bigr\}\mathperiod
	\end{equation*}
	
	\subsection{Randomization Ineligibility}
	\label{sec:ineligible}
	
	In many mobile health studies, the protocol may prohibit intervention delivery at certain decision points due to practical or ethical constraints (e.g., if the participant is driving). 
	In such cases, the decision rule is partially constrained by design: when randomization is disallowed, the study assigns a default ``do nothing'' action.
	This feature is easily accommodated in our framework by a minor refinement of the policy class, with no other changes downstream.
	
	Let $E_t\in\{0,1\}$ denote an eligibility indicator determined prior to assignment at each decision point $t$ (i.e., such that the temporal ordering is $S_t\rightarrow E_t \rightarrow A_t$).
	We assume without loss of generality that when $E_t=0$, the protocol deterministically assigns $A_t=0$.
	
	To enforce the design constraint, define the corresponding policy $q_{t,\text{Elig}}^{\theta}$ by
	\begin{equation*}
		q_{t,\text{Elig}}^{\theta}(a, S_t,E_t) \coloneq E_t q_t^\theta(a, S_t)+(1-E_t)I(a=0)\mathperiod
	\end{equation*}
	That is, when $E_t=1$, the policy collapses to $q_t^\theta$; otherwise it assigns $A_t=0$ deterministically.
	This modified policy both respects the support of the observed data, and prevents the consideration of policies that would intervene on ineligible participants. 
	Similarly, the behavioral assignment mechanism is
	\begin{equation*}
		\pi_{t,\cdot,\text{Elig}}(a\mid H_t,E_t) \coloneq E_t \pi_{t,\cdot}(a \mid H_t,E_t=1) + (1-E_t)I(a=0)\mathcomma
	\end{equation*}
	for $\cdot\in\{0,n\}$.
	Importantly, the inner propensity score is now conditioned on eligibility.
	
	Under this construction, all expressions remain valid after replacing $q_t^\theta$ and $\pi_{t,\cdot}$ with their new counterparts.
	In particular, the time-specific weights $q_{t,\text{Elig}}^{\theta}(A_t, S_t,E_t)/\pi_{t,\cdot,\text{Elig}}(A_t\mid H_t,E_t)$ equal the standard weights $q_{t}^{\theta}(A_t, S_t)/\pi_{t,\cdot}(A_t\mid H_t,E_t=1)$ when $E_t=1$, whereas if $E_t=0$ then $A_t=0$ and the weight is equal to one.
	Consequently, ineligible decision points contribute multiplicative factors of one and can equivalently be omitted from the product.

	\section{Simulation Details}
	\label{sec:simulation-details}
	
	The first state component (representing negative affect) was initialized as $S_{1,1}\sim N(0,1)$ and then evolved according to a persistent autoregressive process with prompt effects.
	Specifically for $t\ge2$, 
	$$S_{t,1} = \operatorname{clip}\bigl\{ \rho S_{t-1,1} + A_{t-1}\tau(S_{t-1}) + \sigma\epsilon_t, -4, 4 \bigr\}\mathcomma$$
	where $\rho=0.96$ and $\sigma = (1-\rho^2)^{1/2}$.
	Clipping restricts states to the closed interval $[-4,4]$, preventing extreme simulated trajectories and aligning with \cref{assumption:bounded} while preserving temporal dependence.
	Thus, in the absence of treatment, negative affect follows a highly-persistent mean-zero process.
	Treatment can shift the subsequent affect through prompt effect $\tau$, which depends on the current state and is detailed in the following paragraph.
	The second component (representing cigarette availability) was generated as an independent Bernoulli variable: $S_{t,2}\sim \operatorname{Bernoulli}(p_2)$ with $p_2=0.25$.
	
	The prompt model encodes two behavioral features.
	First, randomization to a prompt is helpful when an individual has high negative affect and has a high potential for benefit.
	Conversely, the prompt carries a burden cost when an individual is doing relatively well.
	To implement this, we define the smooth vulnerability score $V_t\coloneq \expit\{k(S_{t,1}-c)\}$ for $k=3$ and $c=0.2$. 
	The prompt effect is then given by
	$$\tau(S_{t-1}) = -\gamma_\text{help} V_{t-1}(1+\gamma_\text{avail}S_{t-1,2})+\gamma_\text{harm}(1-V_{t-1})\{1+\gamma_\text{avail}(1-S_{t-1,2})\}\mathcomma$$
	with $\gamma_\text{help}=0.35$, $\gamma_\text{harm}=0.18$ and $\gamma_\text{avail}=0.6$.
	
	Treatment assignments were generated through a logistic model with potential dependence on the current state:
	$\pi_{t,0}(A_t=1\mid H_t) = \expit(a S_{t,1})$.
	We considered $a\in\{0,0.25\}$.
	When $a=0$, treatment is assigned with probability $0.5$ at every decision point, corresponding to a micro-randomized trial.
	When $a=0.25$, participants with higher negative affect are more likely to receive treatment; this represents a generally good observational policy that moves treatment in the same direction as the optimal stochastic policy in our policy class.
	
	The distal outcome was generated as a nonlinear function of the final state, thereby generating nonlinear outcome regression models and challenging bias terms to control. 
	Let $B\coloneq \log\{1+\exp(S_{T,1}-1)\}$ and $H=(\max\{0,S_{T,1}-0.8\})^2$.
	The outcome mean was then
	$$m = \expit\{\eta_\text{soft}B -\eta_\text{hinge}H + \eta_{S_2}S_{T,2}-\eta_{S_1S_2}S_{T,1}S_{T,2}\}\mathcomma$$
	with $\eta_\text{soft}=0.9$, $\eta_\text{hinge}=0.9$, $\eta_{S_2}=0.3$, and $\eta_{S_1S_2}=-0.2$.
	Finally, we let $Y = \operatorname{clip}(m + \epsilon_Y,0,1)$, where $\epsilon_Y\sim N(0,0.05^2)$.
	The term $B$ acts as a softplus function that introduces a smooth penalty for elevated negative affect whereas $H$ is a squared hinge term that imposes an additional penalty once negative affect exceeds a higher threshold.
	Thus, moderate increases in negative affect at time $T$ reduce the outcome gradually, but more extreme affects are penalized more sharply.
	The terms involving $S_{T,2}$ allow cigarette availability to directly influence the final outcome and modify the association between negative affect and the outcome.
	
	\section{Auxiliary Results}
	\label{sec:auxiliary}
	
	We present the following result,
	which is used to control the supremum of several terms throughout the theoretical results.

	\begin{lemma}
		\label{corollary:maximal-inequality-shortcut}
		Let $\mathcal F$ be $\P_0$-Donsker and define $\rho(f,g)\coloneq\lVert f-g\rVert_{2,\P_0}$.
		Let $\mathcal A_n \subseteq \mathcal F \times \mathcal F$ be a random set such that $\sup_{(f,g)\in\mathcal A_n} \rho(f,g)=o_p(1)$.
		Then,
		$$\sup_{(f,g)\in\mathcal A_n} \lvert \G_n(f-g)\rvert =o_p(1)\mathperiod$$
		Equivalently,
		$$\sup_{(f,g)\in\mathcal A_n} \lvert (\P_n-\P_0)(f-g)\lvert = o_p(n^{-1/2})\mathperiod$$
	\end{lemma}
	
	\begin{proof}[Proof of \cref{corollary:maximal-inequality-shortcut}]
		By assumption, there exists a deterministic $\delta_n\downarrow0$ such that 
		$$\P_0 \Bigl(\sup_{(f,g)\in\mathcal A_n} \rho(f,g) > \delta_n\Bigr) \to 0\mathperiod$$
		Hence, for any $\epsilon>0$,
		\begin{align*}
			\P_0 \Bigl( \sup_{(f,g)\in\mathcal A_n} \lvert \G_n(f-g)\rvert > \epsilon \Bigr) &\le \P_0 \Bigl( \sup_{(f,g)\in\mathcal A_n} \rho(f,g) > \delta_n \Bigr) \\
			&\qquad {} + \P_0 \Bigl( \sup_{ \substack{(f,g)\in\mathcal A_n \\ \rho(f,g) \le \delta_n}} \lvert \G_n(f-g)\rvert > \epsilon \Bigr)\mathperiod
		\end{align*}
		The first term tends to zero by construction. 
		The second tends to zero by the asymptotic uniform equicontinuity with respect to the semimetric $\rho$ of the empirical process on the $\P_0$-Donsker class $\mathcal F$.
		Therefore, $\sup_{(f,g)\in\mathcal A_n}\lvert \G_n(f-g)\rvert=o_p(1)$.
	\end{proof}

	\section{Proofs Regarding Regimen-Response Curve Estimation}
	\label{sec:proof-estimation}
	
	\subsection{Identification}

	\begin{proof}[Proof of \cref{lemma:identification-expectation}]
		We begin with a more rigorous definition of the potential outcome, $Y^{\theta,\alpha}$.
		Let $\bar{A}^{\theta,\alpha}$ be generated sequentially as follows. 
		For $t =1,\ldots,T$,
		\begin{itemize}
			\item Given history $H_t^{\theta,\alpha}=(\bar{S}_t^{\theta,\alpha},\bar{A}_{t-1}^{\theta,\alpha})$, draw $A_t^{\theta,\alpha}\sim q_t^{\theta,\alpha}(\cdot \mid H_t^{\theta,\alpha})$.
			\item Draw counterfactual state variables $S_{t+1}^{\theta,\alpha}$ given $(H_t^{\theta,\alpha}, A_t^{\theta,\alpha})$ according to the same conditional law as under $\P_0$.
		\end{itemize}
		Then, define $Y^{\theta,\alpha}\coloneq Y^{\bar{A}^{\theta,\alpha}}$, as the potential outcome that would be observed if treatment were assigned according to this stochastic policy.
		Under \cref{assumption:identification}, the joint distribution of $(\bar{S}^{\theta,\alpha},\bar{A}^{\theta,\alpha},Y^{\theta,\alpha})$ is fully characterized by replacing the observed treatment mechanism $\pi_{t,0}(a_t\mid H_t)$ with $q_t^{\theta,\alpha}(a_t \mid H_t)$ and leaving the conditional laws of the covariates and outcome unchanged.
		Let $\P_0^{\theta,\alpha}$ denote the induced distribution under this replacement.
		Then, for any integrable $f(Y)$,
		$$\P_0\{f(Y^{\theta,\alpha})\} = \int f(y)\dd \P_0^{\theta,\alpha}(o)\mathsemicolon$$
		specifically,
		$$\P_0(Y^{\theta,\alpha}) = \int y\dd \P_0^{\theta,\alpha}(o)\mathperiod$$
		
		Under $\P_0$, the observed data law factorizes as
		$$\dd \P_0(o) = p(s_1) \Biggl\{\prod_{t=1}^T\pi_{t,0}(a_t\mid h_t)\Biggr\} \Biggl\{\prod_{t=1}^{T-1} p(s_{t+1}\mid h_t, a_t)\Biggr\}p(y\mid\bar{s},\bar{a})\mathcomma$$
		whereas
		$$\dd \P_0^{\theta,\alpha}(o) = p(s_1) \Biggl\{\prod_{t=1}^Tq_t^{\theta,\alpha}(a_t\mid h_t)\Biggr\} \Biggl\{\prod_{t=1}^{T-1}p(s_{t+1}\mid h_t, a_t)\Biggr\}p(y\mid\bar{s},\bar{a})\mathperiod$$
		Then, under \cref{assumption:identification:positivity}, we have $q_t^{\theta,\alpha}(\cdot\mid H_t) \ll \pi_{t,0}(\cdot \mid H_t)$ almost surely and hence, $\P_0^{\theta,\alpha} \ll \P_0$.
		The Radon-Nikod\`ym derivative of $\P_0^{\theta,\alpha}$ with respect to $\P_0$ is therefore
		$$\frac{\dd\P_0^{\theta,\alpha}}{\dd\P_0}(o)=\prod_{t=1}^T \frac{q_t^{\theta,\alpha}(a_t \mid h_t)}{\pi_{t,0}(a_t\mid h_t)}=W_0^{\theta,\alpha}\mathperiod$$
		Hence, for any integrable $f$,
		$$\int f(o) \dd\P_0^{\theta,\alpha}(o) = \int f(o)\prod_{t=1}^T \frac{q_t^{\theta,\alpha}(a_t \mid h_t)}{\pi_{t,0}(a_t\mid h_t)} \dd\P_0(o) =\P_0\{f(O)W_0^{\theta,\alpha}\}\mathperiod$$
		Taking $f(O)=Y$ we have $\P_0(Y^{\theta,\alpha})=\int y \dd\P_0^{\theta,\alpha}(o) = \P_0(W_0^{\theta,\alpha}Y)$, which completes the proof.
	\end{proof}
	
	\subsection{Canonical Gradient}

	\begin{proof}[Proof of \cref{lemma:dcar-expectation}]
		Let $\{\P_\epsilon:\epsilon\in\mathbb{R}\}$ be a regular parametric submodel through $\P_0$ with density $p_\epsilon$ and $\P_{\epsilon=0}=\P_0$. We denote by $h(O)$ the mean-zero score of the submodel at $\epsilon=0$, in the sense that for any integrable $f$, $[\partial_\epsilon \P_\epsilon f]_{\epsilon=0}=[\partial_\epsilon \int f(o)\,p_\epsilon(o)\dd{o}]_{\epsilon=0}= \P_0\{f(O)h(O)\}$.
		For each $\P$, let $\pi_{t,\P}(a\mid  H_t)$ denote the propensity at time $t$, and define $W^{\theta,\alpha}_{t}(\P)(A_t, H_t)=1+(t/T)^{\alpha}\{\frac{q_t^\theta(A_t, H_t)}{\pi_{t,\P}(A_t\mid  H_t)}-1\}$ with product weight $W^{\theta,\alpha}(\P)(O)=\prod_{t=1}^T W^{\theta,\alpha}_{t}(\P)(A_t, H_t)$. 
		The target parameter introduced in \cref{sec:ipw-mapping} is $\Psi(\P,\theta,\alpha)=\P\{W^{\theta,\alpha}(\P)(O)Y\}$. 
		By the defining property of the canonical gradient at $\P_0$, there exists a mean-zero measurable function $\phi_\psi(O;\eta_0,\theta,\alpha)$ such that $\{\partial_\epsilon \psi(\P_\epsilon,\theta,\alpha)\}_{\epsilon=0}= \P_0\{\phi_\psi(O;\eta_0,\theta,\alpha)h(O)\}$. 
		We continue by computing this pathwise derivative and verifying the stated closed-form for $\phi_\psi$. 
		
		We consider the general case of perturbation weights $\delta=(\delta_1,\ldots,\delta_T)\in [0,1]^T$ in place of the weights $(t/T)^\alpha$.
		Under this notation, 
		$$W_{t}^{\theta,\delta_t}(\P_\epsilon)\coloneq (1-\delta_t) + \delta_t \frac{q_t^\theta(A_t, S_t)}{\pi_{t,\epsilon}(A_t\mid H_t)},\quad W^{\theta,\delta}(\P_\epsilon)\coloneq \prod_{t=1}^T  W_{t}^{\theta,\delta_t}(\P_\epsilon)\mathcomma$$
		and $\psi_0(\theta,\delta;\P_\epsilon)=\P_\epsilon\{W^{\theta,\delta}(\P_\epsilon)Y\}$.
		We have
		\begin{align*}
			&\partial_\epsilon \psi_0(\theta,\delta;\P_\epsilon)|_{\epsilon=0}\\
			&= \left.\partial_\epsilon\P_0\left[\left\{\prod_{t=1}^T \frac{\delta_t q_t^\theta+(1-\delta_t)\pi_{t,\epsilon}}{\pi_{t,\epsilon}}\right\}Y\{1+\epsilon h(O)\}\right]\right|_{\epsilon=0} \\
			&= \P_0\left[\left\{\prod_{t=1}^T \frac{\delta_t q_t^\theta+(1-\delta_t)\pi_{t,0}}{\pi_{t,0}}\right\}Yh(O)\right] + \P_0\left[\left.\partial_\epsilon\left\{\prod_{t=1}^T \frac{ q_t^\theta+(1-\delta_t)\pi_{t,\epsilon}}{\pi_{t,\epsilon}}\right\}Y \right|_{\epsilon=0}\right]\mathperiod
		\end{align*}
		Then,
		\begin{align*}
			&\P_0\left[\left.\partial_\epsilon\left\{\prod_{t=1}^T \frac{\delta_t q_t^\theta+(1-\delta_t)\pi_{t,\epsilon}}{\pi_{t,\epsilon}}\right\}Y \right|_{\epsilon=0}\right] \\
			&= \P_0 \Bigg[\sum_{t=1}^T\left.\partial_\epsilon\left\{ \frac{\delta_t q_t^\theta+(1-\delta_t)\pi_{t,\epsilon}}{\pi_{t,\epsilon}}\right\} \right|_{\epsilon=0} \Biggl\{\prod_{j\neq t}\frac{\delta_t q_t^\theta+(1-\delta_t)\pi_{t,0}}{\pi_{t,0}} \Biggr\}Y\Bigg] \\
			&= \P_0\Bigg[\sum_{t=1}^T \underbrace{\Biggl\{ \frac{-\delta_tq_t^\theta}{\pi_{t,0}}h(A_t\mid H_t)\Biggr\}\Biggl\{\prod_{j\neq t}\frac{\delta_t q_t^\theta+(1-\delta_t)\pi_{t,0}}{\pi_{t,0}} \Biggr\}Y}_{\eqqcolon X_t}\Bigg] \\
			&= \sum_{t=1}^T \P_0(X_t)\mathcomma
		\end{align*}
		where
		\begin{align*}
			&\P_0(X_t) \\
			&= \P_0 \{\P_0(X_t \mid H_t)\} \\
			&= \P_0 \Biggl[-\delta_t \Bigl(\prod_{j<t}W_{j,0}^{\theta,\delta} \Bigr) \frac{q_t^\theta}{\pi_{t,0}} 
			h(A_t\mid H_t) \underbrace{\P_0\biggl\{ \Bigl(\prod_{s>t}W_{j,0}^{\theta,\delta} \Bigr)Y \mid H_t \biggr\}}_{=\mu_t^{\theta,\delta}(H_t, A_t)}\Biggr] \\
			&= \P_0 \Biggl[-\delta_t \Bigl(\prod_{j<t}W_{j,0}^{\theta,\delta} \Bigr)\sum_{a\in\{0,1\}}\Biggl\{ \frac{I(A_t=a) q_t^\theta(a,S_t)}{\pi_{t,0}(A_t\mid H_t)}\Biggr\} \mu_t^{\theta,\delta}(H_t,a)h(A_t\mid H_t)\Biggr] \\
			&=  \P_0 \Biggl[-\delta_t \Bigl(\prod_{j<t}W_{j,0}^{\theta,\delta} \Bigr)\sum_{a\in\{0,1\}}\Biggl\{ \frac{I(A_t=a) q_t^\theta(a,S_t)}{\pi_{t,0}(A_t\mid H_t)}\Biggr\} \mu_t^{\theta,\delta}(H_t,a)h(A_t\mid H_t)\Biggr] \\
			&\qquad - \P_0\Biggl[-\delta_t \Bigl(\prod_{j<t}W_{j,0}^{\theta,\delta} \Bigr)  \sum_{a\in\{0,1\}} 
			q^\theta_t(a,S_t)\mu_t^{\theta,\delta}(H_t,a)h(A_t\mid H_t)\Biggr] \\
			&= \P_0\Biggl[\underbrace{-\delta \Bigl(\prod_{j<t}W_{j,0}^{\theta,\delta} \Bigr)\sum_{a\in\{0,1\}} \frac{I(A_t=a)-\pi_{t,0}(A_t\mid H_t)}{\pi_{t,0}(A_t\mid H_t)} q_t^\theta(a,H_t) \mu_t^{\theta,\delta}(H_t,a)}_{\eqqcolon f(A_t,H_t)}h(A_t\mid H_t)\Biggr] \\
			&= \P_0\{f(A_t, H_t)h(O)\}\mathcomma
		\end{align*}
		where the final equality holds because $f(A_t, H_t)$ is mean zero conditional on $ H_t$.
		It follows that the canonical gradient is
		\begin{align*}
			\phi_\psi(O) &=W_0^{\theta,\delta}Y \\ &\quad- \sum_{t=1}^T\delta_t \Bigl(\prod_{j<t} W_{j,0}^{\theta,\delta} \Bigr) \sum_{a\in\{0,1\}} \frac{I(A_t=a)-\pi_{t,0}(A_t\mid H_t)}{\pi_{t,0}(A_t\mid H_t)} q_t^\theta(A_t=a,H_t) \mu_t^{\theta,\delta}(H_t,A_t=a) \\
			&\quad -\psi_{0}(\theta,\delta)\mathperiod
		\end{align*}
		Finally, the result of \cref{lemma:dcar-expectation} is given by letting $\delta_t=(t/T)^\alpha$.
	\end{proof}

	\subsection{Theorem 1}

	\begin{proof}[Proof of \cref{theorem:uipw-asym}]
		We begin with an algebraic decomposition of $\psi_n(\theta;\alpha)-\psi_0(\theta;\alpha)$ for any given $(\theta,\alpha)$.
		We have
		\begin{align*}
			\psi_n(\theta;\alpha)-\psi_0(\theta;\alpha) &= (\P_n-\P_0)(W_0^{\theta,\alpha}Y) + \P_n(W_n^{\theta,\alpha}Y-W_0^{\theta,\alpha}Y) \\
			&= (\P_n-\P_0)(W_0^{\theta,\alpha}Y) + \P_0(W_n^{\theta,\alpha}Y-W_0^{\theta,\alpha}Y) + R_{1,n}(\theta,\alpha) \mathcomma
		\end{align*}
		where $R_{1,n}(\theta,\alpha)=(\P_n-\P_0)(W_n^{\theta,\alpha}Y-W_0^{\theta,\alpha}Y)$. 
		
		We continue with the second term, $\P_0(W_n^{\theta,\alpha}Y-W_0^{\theta,\alpha}Y)=\P_0\{(\prod_{t=1}^T W_{t,n}^{\theta,\alpha}-\prod_{t=1}^T W_{t,0}^{\theta,\alpha})Y\}$.
		We factor the difference of products using the following identity: 
		$\prod_{t=1}^Ta_t - \prod_{t=1}^Tb_t = \sum_{t=1}^T(\prod_{i<t}a_i)(a_t-b_t)(\prod_{j>t}b_j)$. 
		We have that
		\begin{align*}
			W_{t,n}^{\theta,\alpha}-W_{t,0}^{\theta,\alpha}&= (t/T)^\alpha q^\theta_t(A_t,\bar{S}_t)\left\{\frac{1}{\pi_{t,n}(A_t\mid H_t)}-\frac{1}{\pi_{t,0}(A_t\mid H_t)}\right\} \\
			&= (t/T)^\alpha q^\theta_t(A_t,\bar{S}_t)\frac{\pi_{t,0}(A_t\mid H_t)-\pi_{t,n}(A_t\mid H_t)}{\pi_{t,0}(A_t\mid H_t)\pi_{t,n}(A_t\mid H_t)}\mathcomma    
		\end{align*}
		Hence,
		\begin{align*}
			&\P_0\big(W_n^{\theta,\alpha}Y - W_0^{\theta,\alpha}Y\big) \\
			&= \P_0 \left[\sum_{t=1}^T \Bigl(\prod_{i<t} W_{i,n}^{\theta,\alpha}\Bigr\}
			(t/T)^\alpha\, q^\theta_t(A_t,S_t)\,
			\frac{\pi_{t,0}(A_t\mid H_t)-\pi_{t,n}(A_t\mid H_t)}{\pi_{t,0}(A_t\mid H_t)\pi_{t,n}(A_t\mid H_t)}\Bigl(\prod_{j>t} W_{j,0}^{\theta,\alpha}\Bigr)Y\right] \\
			&= \P_0 \left[\sum_{t=1}^T \Bigl(\prod_{i<t} W_{i,n}^{\theta,\alpha}\Bigr)
			(t/T)^\alpha\, q_t^\theta(A_t,S_t)\,
			\frac{\pi_{t,0}(A_t\mid H_t)-\pi_{t,n}(A_t\mid H_t)}{\pi_{t,0}(A_t\mid H_t)\pi_{t,n}(A_t\mid H_t)}\mu_{t,0}^{\theta,\alpha}(H_t, A_t)\right] \\
			&= \P_0 \left[\sum_{t=1}^T \Bigl(\prod_{i<t} W_{i,n}^{\theta,\alpha}\Bigl)
			(t/T)^\alpha \sum_{a\in\{0,1\}} I (A_t=a)\mu_{t,0}^{\theta,\alpha}(H_t, a)\, q_t^\theta(a,S_t)\,
			\frac{\pi_{t,0}(a\mid H_t)-\pi_{t,n}(a\mid H_t)}{\pi_{t,0}(a\mid H_t)\pi_{t,n}(a\mid H_t)}\right] \\
			&= \P_0 \left[\sum_{t=1}^T \Bigl(\prod_{i<t} W_{i,n}^{\theta,\alpha}\Bigr)
			(t/T)^\alpha \sum_{a\in\{0,1\}} \mu_{t,0}^{\theta,\alpha}(H_t,a)\, q_t^\theta(a,S_t)\,
			\frac{I(A_t=a)-\pi_{t,n}(a\mid H_t)}{\pi_{t,n}(a\mid H_t)}\right]\mathperiod\stepcounter{equation}\tag{\theequation}\label{eqn:dcar-remainder}
		\end{align*}
		
		Let $\Dcar(\mu_0^{\theta,\alpha},\pi_n;\theta,\alpha) \coloneq \sum_{t=1}^T \Dcart(\mu_0^{\theta,\alpha},\pi_n;\theta,\alpha)$, where
		\begin{align*}
			&\Dcart(\mu_0^{\theta,\alpha},\pi_n;\theta,\alpha) \\
			&\coloneq \Bigl(\prod_{i<t} W_{i,n}^{\theta,\alpha}\Bigr)
			(t/T)^\alpha\, \sum_{a\in\{0,1\}}\mu_{t,0}^{\theta,\alpha}(H_t,a)\, q_t^\theta(a,S_t)
			\frac{I(A_t=a)-\pi_{t,n}(a\mid H_t)}{\pi_{t,n}(a\mid H_t)} \\
			&= \Bigl(\prod_{i<t} W_{i,n}^{\theta,\alpha}\Bigr)
			(t/T)^\alpha\Bigl\{\frac{\mu_{t,0}^{\theta,\alpha}(1,H_t)q_t^\theta(1,S_t)}{\pi_{t,n}(1\mid H_t)} - \frac{\mu_{t,0}^{\theta,\alpha}(0,H_t)q_t^\theta(0,S_t)}{\pi_{t,n}(0\mid H_t)}\Bigr\}\{A_t-\pi_{t,n}(1\mid H_t)\}\mathperiod
		\end{align*}
		That is, $\Dcart(\mu_0^{\theta,\alpha},\pi_n;\theta,\alpha)$ is the argument of the expectation in \cref{eqn:dcar-remainder}.
		Then,
		\begin{align*}
			&\P_0\{\Dcar(\mu_0^{\theta,\alpha},\pi_n;\theta,\alpha)\} \\
			&= (\P_0-\P_n)\{\Dcar(\mu_0^{\theta,\alpha},\pi_n;\theta,\alpha)\}  + \P_n\{\Dcar(\mu_0^{\theta,\alpha},\pi_n;\theta,\alpha)\}  \\
			&= (\P_0-\P_n)\{\Dcar(\mu_0^{\theta,\alpha},\pi_0;\theta,\alpha)\}  + \P_n\{\Dcar(\mu_0^{\theta,\alpha},\pi_n;\theta,\alpha)\}  + R_{2,n}(\theta,\alpha)\mathcomma
		\end{align*}
		where
		$\Dcar(\mu_0^{\theta,\alpha},\pi_0;\theta,\alpha)$ is the augmentation term from the canonical gradient $\phi_\psi$ as given in \cref{lemma:dcar-expectation}, and 
		$R_{2,n}(\theta,\alpha)=(\P_0-\P_n)\{\Dcar(\mu_0^{\theta,\alpha},\pi_n;\theta,\alpha)-\Dcar(\mu_0^{\theta,\alpha},\pi_0;\theta,\alpha)\}$.
		Combining everything, we have
		\begin{align*}
			&\psi_n(\theta;\alpha)-\psi_0(\theta;\alpha) \\
			&= (\P_n-\P_0)(W_0^{\theta,\alpha}Y) + (\P_0-\P_n)\{\Dcar(\mu_0^{\theta,\alpha},\pi_0;\theta,\alpha)\}+\P_n\{\Dcar(\mu_0^{\theta,\alpha},\pi_n;\theta,\alpha)\} \\
			&\qquad\, {} + (R_{1,n}+R_{2,n})(\theta,\alpha) \\
			&= (\P_n-\P_0) \{\phi_\psi(O; \eta_0, \theta,\alpha)\}+ \P_n\{\Dcar(\mu_0^{\theta,\alpha},\pi_n;\theta,\alpha)\} + (R_{1,n}+R_{2,n})(\theta,\alpha) \mathperiod \stepcounter{equation}\tag{\theequation}\label{eqn:remainder-full}
		\end{align*}    
		We prove that $\P_n\{\Dcar(\mu_0^{\theta,\alpha},\pi_n;\theta,\alpha)\}$ and the two remainder terms and are bounded uniformly over $\Theta$ in \cref{corollary:dcar-bounded-uniform,lemma:r1-bounded-uniform,lemma:r2-bounded-uniform}, respectively.
		Together, these bounds give 
		$$\sup_{\theta\in\Theta}\lvert \P_n\{\Dcar(\mu_0^{\theta,\alpha},\pi_n;\theta,\alpha)\}+(R_{1,n}+R_{2,n})(\theta,\alpha) \rvert=o_p(n^{-1/2})\mathperiod$$    
		This immediately gives the desired pointwise result for any fixed $\theta$:
		$$\psi_n(\theta; \alpha)-\psi_0(\theta; \alpha) = (\P_n-\P_0)\phi_\psi(O;\eta_0,\theta,\alpha) + o_p(n^{-1/2})\mathcomma$$
		and, hence,
		$$\sqrt{n}\{\psi_n(\theta,\alpha)-\psi_0(\theta,\alpha)\} \dto N\big[0, \P_0\{\phi_\psi^2(O;\eta_0,\theta,\alpha)\}\big]\mathperiod $$
		
		Finally, we argue that the leading term in \cref{eqn:remainder-full} is a tight empirical process indexed by $\Theta$.
		Let $\mathcal{F}=\{\phi(O,\eta_0,\theta,\alpha) :\theta\in\Theta\}$ denote the function class. 
		First, note that $\mathcal{F}$ has finite sectional variational norm under \cref{assumption:cadlag} and therefore $\mathcal{F}$ is Donsker and is bounded by some envelope function $F$ with $\P_0  F^2<\infty$.
		Moreover, $\Theta$ is compact and $\phi$ is almost surely continuous in $\theta$. 
		It follows that $\sqrt{n}(\P_n-\P_0) \{\phi(O,\eta_0,\theta,\alpha)\}\dto \G(\theta;\alpha)$ in $\ell^\infty(\Theta)$, where $\G(\cdot;\alpha)$ is a mean-zero Gaussian process with covariance
		\begin{equation}
			\label{eqn:g-covar}
			\Cov\{\G(\theta_1;\alpha),\G(\theta_2;\alpha)\}=\Cov\{ \phi(O,\eta_0,\theta_1,\alpha),\phi(O,\eta_0,\theta_2,\alpha)\}\mathperiod  
		\end{equation}
		Combining this with the uniform bounds on the remainder terms, we have 
		$$\sqrt{n}\{\psi_n(\theta,\alpha)-\psi_0(\theta,\alpha)\}\dto \G(\theta;\alpha)\mathcomma$$
		as desired.
	\end{proof}
	
	We conclude this subsection by providing bounds for the non-leading terms in \cref{eqn:dcar-remainder}.
	The bias term is controlled by \cref{lemma:d-bounded-uniform}, which extends the general UIPW result from \cite[Lemma 1]{ertefaie_nonparametric_ipw_2023}.
	We proceed by applying this result (\cref{lemma:d-bounded-uniform}) to the specific bias terms that arise in the longitudinal setting.
	Finally, we establish bounds for the remainder terms using standard empirical process theory arguments.
	
	\begin{lemma}
		\label{corollary:dcar-bounded-uniform}
		Suppose that \cref{assumption:identification:positivity,,assumption:q-smooth} hold, that
		\cref{assumption:cadlag,,assumption:hal-space} hold uniformly in $\Theta$ at $\alpha$, and that $\pi_{t,n}$ is sufficiently undersmoothed. 
		Then $$\sup_{\theta\in\Theta} \lvert \P_n\{\Dcar(\mu_0^{\theta,\alpha},\pi_n;\theta,\alpha)\} \rvert=o_p(n^{-1/2})\mathperiod$$
	\end{lemma}
	
	\begin{proof}[Proof of \cref{corollary:dcar-bounded-uniform}]
		We show that for any given $t$, $\Dcart(\mu_0^{\theta,\alpha},\pi_n;\theta,\alpha)\}=D\{f_\theta(O),\pi_n\}$ where $f_\theta$ satisfies the conditions of \cref{lemma:d-bounded-uniform}.
		
		Let
		\begin{equation*}
			f_{t,n}^{\theta,\alpha}(O) \coloneq \Bigl\{\prod_{i<t} W_{i,n}^{\theta,\alpha}(H_i) \Bigr\}(t/T)^\alpha \Bigl\{\frac{\mu_{t,0}^{\theta,\alpha}(1,H_t)q_t^\theta(1,S_t)}{\pi_{t,n}(1\mid H_t)} - \frac{\mu_{t,0}^{\theta,\alpha}(0,H_t)q_t^\theta(0,S_t)}{\pi_{t,n}(0\mid H_t)}\Bigr\}
		\end{equation*}
		so
		\begin{equation*}
			\Dcar(O,\mu_0^{\theta,\alpha},\pi_n;\theta,\alpha) = f_{t,n}^{\theta,\alpha}(O)\{A_t-\pi_{t,n}(1\mid H_t)\}\eqqcolon  D(f_{t,n}^{\theta,\alpha},\pi_{t,n})\mathperiod
		\end{equation*}
		Likewise let
		\begin{equation*}
			f_{t,0}^{\theta,\alpha}(O) \coloneq \Bigl\{\prod_{i<t} W_{i,0}^{\theta,\alpha}(H_i) \Bigr\}(t/T)^\alpha \Bigl\{\frac{\mu_{t,0}^{\theta,\alpha}(1,H_t)q_t^\theta(1,S_t)}{\pi_{t,0}(1\mid H_t)} - \frac{\mu_{t,0}^{\theta,\alpha}(0,H_t)q_t^\theta(0,S_t)}{\pi_{t,0}(0\mid H_t)}\Bigr\}
		\end{equation*}
		be a version that uses the true treatment model.
		
		We first show that $f_{t,n}^{\theta,\alpha}$ is c\`adl\`ag with uniformly bounded total SVN. 
		First, under \cref{assumption:cadlag}, $\mu_{t,0}^{\theta,\alpha}$ is c\`adl\`ag with uniformly bounded total SVN.
		Likewise we have that $\pi_{t,0}$ is c\`adl\`ag with finite total SVN; under \cref{assumption:identification:positivity}, the same holds for $\pi_{t,0}^{-1}$.
		Moreover, $\pi_{i,n}:i<t$ are c\`adl\`ag with finite total SVN by construction and, by the universal consistency of the HAL, so are $\pi_{i,n}^{-1}$.
		Finally, $q_i^\theta:i\le t$ is c\`adl\`ag with finite total SVN under \cref{assumption:q-smooth}.
		It follows that $f_{t,n}^{\theta,\alpha}$ is a finite product of terms which are all c\`adl\`ag with uniformly bounded total SVN and therefore $f_{t,n}^{\theta,\alpha}$ is also c\`adl\`ag with uniformly bounded total SVN.
		
		We conclude by showing that the projection error is controlled uniformly to invoke \cref{lemma:d-bounded-uniform}. 
		Under \cref{assumption:hal-space}, we have $\sup_{\theta\in\Theta}\|\tilde{f}_{t,0}^{\theta,\alpha}-f_{t,0}^{\theta,\alpha}\|_{2,\P_0}=O_p(n^{-1/4})$.
		Furthermore, $\sup_{\theta\in\Theta}\|f_{t,n}^{\theta,\alpha}-f_{t,0}^{\theta,\alpha}\|_{2,\P_0}=o_p(n^{-1/4})$ by the HAL convergence of $\pi_{i,n}$ for $i<t$. 
		Then, by the triangle inequality,
		\begin{align*}
			\sup_{\theta\in\Theta} \| f_{t,n}^{\theta,\alpha}-\tilde{f}_{t,n}^{\theta,\alpha} \| &\le \sup_{\theta\in\Theta}\| f_{n,t}^{\theta,\alpha}-f_{t}^{\theta,\alpha}\| + \sup_{\theta\in\Theta}\|\tilde{f}_{n,t}^{\theta,\alpha}-\tilde{f}_t^{\theta,\alpha}\| + \sup_{\theta\in\Theta}\|f_t^{\theta,\alpha}-\tilde{f}_t^{\theta,\alpha}\| \\
			&\le 2\sup_{\theta\in\Theta}\| f_{n,t}^{\theta,\alpha}-f_{t}^{\theta,\alpha}\|  + \sup_{\theta\in\Theta}\|f_t^{\theta,\alpha}-\tilde{f}_t^{\theta,\alpha}\| \\ 
			&= o_p(n^{-1/4}) + O_p(n^{-1/4})\mathperiod
		\end{align*}
		Hence, $\sup_{\theta\in\Theta}\| f_{t,n}^{\theta,\alpha}-\tilde{f}_{t,n}^{\theta,\alpha} \|=O_p(n^{-1/4})$.
		The conditions are met to apply \cref{lemma:d-bounded-uniform} for $f_{t,n}^{\theta,\alpha}$, hence
		$\sup_{\theta\in\Theta} \lvert \P_n\{\Dcart(\mu_0^{\theta,\alpha},\pi_n;\theta,\alpha) \rvert =o_p(n^{-1/2})$ for each $t\in\{1,\ldots,T\}$.
		It follows that $\sup_{\theta\in\Theta}\lvert \P_n\{\Dcar(\mu_0^{\theta,\alpha},\pi_n;\theta,\alpha)\rvert =o_p(n^{-1/2})$.
	\end{proof}
	
	\begin{lemma}
		\label{lemma:r1-bounded-uniform}
		The remainder $R_{1,n}(\theta,\alpha)$ is uniformly bounded in $\theta$:
		$$\sup_{\theta\in\Theta}\abs{(\P_n-\P_0)\{W_n(O;\theta,\alpha)Y-W_0(O;\theta,\alpha)Y\}}=o_p(n^{-1/2})\mathperiod$$
	\end{lemma}
	
	\begin{proof}[Proof of \cref{lemma:r1-bounded-uniform}]
		Consider the following function class indexed by $\theta$:
		\begin{equation*}
			\mathcal{F}_{n,1}\coloneq\big[f_{n,\theta,\alpha}(O)=\left\{W_n(O;\theta,\alpha)-W_0(O;\theta,\alpha)\right\}Y : \theta \in\Theta\big]\mathperiod
		\end{equation*}
		
		First, we establish a uniformly bounded envelope function.
		Note that under strong positivity (\cref{assumption:identification}),  $\abs*{W_{t,0}^{\theta,\alpha}(A_t, H_t)}$ is bounded uniformly over $(\theta,\alpha)$; let $C_W$ be the maximum of these bounds over all $t$.
		Additionally, given that $\pi_{t,n}$ is consistent for $\pi_{t,0}$, we also have that $\abs*{W_{t,n}^{\theta,\alpha}(A_t, H_t)}<C_W$ with high probability. 
		Moreover, we note that the difference in time-specific weights is bounded.
		Recall that $W_{t,0}^{\theta,\alpha}-W_{t,n}^{\theta,\alpha}=(t/T)^\alpha q_t^\theta (\pi_{t,0}-\pi_{t,n})/(\pi_{t,n}\pi_{t,0})$.
		Hence,
		\begin{equation}
			\label{eqn:w-diff-bound}
			\abs*{W_{t,0}^{\theta,\alpha}-W_{t,n}^{\theta,\alpha}}\le (t/T)^\alpha \norm*{q_t^\theta}_\infty \norm*{\pi_{t,n}^{-1}}_\infty \norm*{\pi_{t,0}^{-1}}_\infty\abs*{\pi_{t,0}-\pi_{t,n}}\mathperiod
		\end{equation}
		This bound is finite under strong positivity of both $\pi_{t,n}$ and $\pi_{t,0}$.
		Using a telescoping expansion, we have that 
		\begin{equation*}
			\abs{\prod_{t=1}^T W^{\theta,\alpha}_{t,0}-\prod_{t=1}^T W^{\theta,\alpha}_{t,n}}\le \sum_{t=1}^T \abs{W^{\theta,\alpha}_{t,0}-W^{\theta,\alpha}_{t,n}}C_W^{T-1}<C'_W\mathperiod   
		\end{equation*}
		Thus, we obtain the envelope function $F_1(O)\coloneq C_YC'_W$ with $\norm{F_1(O)}_{\P_0,p}<\infty$ for all $p\ge1$.
		
		Next, we control the $L_2(\P_0)$-seminorm of the function class.
		\cref{eqn:w-diff-bound} and strong positivity give
		\begin{equation*}
			\norm*{W_{t,0}^{\theta,\alpha}-W_{t,n}^{\theta,\alpha}}_{\P_0,2}\lesssim \norm*{\pi_{t,0}-\pi_{t,n}}_{\P_0,2}\mathperiod
		\end{equation*}
		Under consistency of $\pi_{t,n}$ for $\pi_{t,0}$, $\norm*{\pi_{t,0}-\pi_{t,n}}_{\P_0,2}\to0$, and therefore
		\begin{equation*}
			\sup_{\theta\in\Theta} \norm*{W_{t,0}^{\theta,\alpha}-W_{t,n}^{\theta,\alpha}}_{\P_0,2}\to 0\mathperiod
		\end{equation*}
		Therefore, we have that
		\begin{equation*}
			\sup_{\theta\in\Theta} \norm*{f_{n,\theta,\alpha}}_{\P_0,2}\le C_Y C_W^{T-1}\sum_{t=1}^T \sup_{\theta\in\Theta} \norm*{W_{t,0}^{\theta,\alpha}-W_{t,n}^{\theta,\alpha}}_{\P_0,2}\to 0\mathperiod
		\end{equation*}
		
		Additionally, the entropy of the index class is controlled.
		Under \cref{assumption:cadlag}, $\pi_{t,0}$ has finite sectional variational norm.
		Likewise, $\pi_{t,n}$ has finite sectional variational norm by construction.
		Therefore the classes $\{\pi_{t,\cdot}(A_t\mid H_t)\}$ for $\cdot\in\{0,n\}$ and $\{q_t^\theta(A_t, H_t):\theta\in\Theta$ are Donsker.
		Under strong positivity, $\{\pi^{-1}_{t,\cdot}(A_t\mid H_t)\}$ is also Donsker.
		Given that finite products and sums of Donsker functions are Donsker we have that $\mathcal{F}_{n,1}$ is also Donkser.
		
		Finally, applying \cref{corollary:maximal-inequality-shortcut}, we have  $\sup_{\theta\in\Theta}\abs{R_{1,n}}=o_p(n^{-1/2})$.
	\end{proof}
	
	\begin{lemma}
		\label{lemma:r2-bounded-uniform}
		The remainder term $R_{2,n}(\theta,\alpha)$ is bounded uniformly in $\theta$:
		$$\sup_{\theta\in\Theta}\abs{(\P_0-\P_n)\{\Dcar(\mu_0^{\theta,\alpha},\pi_n;\theta,\alpha)-\Dcar(\mu_0^{\theta,\alpha},\pi_0;\theta,\alpha)\}}=o_p(n^{-1/2})\mathperiod$$
	\end{lemma}
	
	\begin{proof}[Proof of \cref{lemma:r2-bounded-uniform}]
		Fix $t\in\{1,\ldots,T\}$ and consider the function class indexed by $\theta$:
		\begin{equation*}
			\mathcal{F}_{n,2}\coloneq\big[f_{n,\theta,\alpha}(O)=\{g_{n,\theta,\alpha}(O)-g_{0,\theta,\alpha}(O)\} : \theta \in\Theta\big]\mathperiod
		\end{equation*}
		where
		$g_{n,\theta,\alpha}\coloneq\Dcart(\mu_0^{\theta,\alpha},\pi_n;\theta,\alpha)$ and $g_{0,\theta,\alpha}\coloneq\Dcart(\mu_0^{\theta,\alpha},\pi_0;\theta,\alpha)$.
		Recall that, by a telescoping decomposition,
		\begin{align*}
			\abs{\prod_{i<t} W_{i,n}^{\theta,\alpha}-\prod_{i<t}W_{i,0}^{\theta,\alpha}} \le C \sum_{i<t} \abs{\pi_{i,n}-\pi_{i,0}}\mathcomma 
		\end{align*}
		and
		\begin{equation*}
			\norm{\prod_{i<t} W_{i,n}^{\theta,\alpha}-\prod_{i<t}W_{i,0}^{\theta,\alpha}}_{\P_0,2} \le C \sum_{i<t} \norm{\pi_{i,n}-\pi_{i,0}}_{\P_0,2} \mathcomma
		\end{equation*}
		for $C<\infty$ that does not depend on $(\theta,\alpha)$.
		Moreover, for $R^{\theta,\alpha}_{t,\cdot}=\sum_{a\in\{0,1\}} \mu_{t,0}^{\theta,\alpha}q_t^\theta\{I(A_t=a)-\pi_{t,\cdot}\}/\pi_{t,0}$, we have
		\begin{equation*}
			\lVert R_{t,n}^{\theta,\alpha}-R_{t,0}^{\theta,\alpha} \rVert_{\P_0,2} \le C \lVert \pi_{t,n}-\pi_{t,0} \rVert_{\P_0,2}\mathperiod
		\end{equation*}
		Therefore, we have
		\begin{equation*}
			g_{n,\theta,\alpha}-g_{0,\theta,\alpha} = \Bigl(\prod_{i<t} W_{i,n}^{\theta,\alpha}-\prod_{i<t}W_{i,0}^{\theta,\alpha}\Bigr)R_{t,n}^{\theta,\alpha} + \Bigl(\prod_{i<t}W_{i,0}^{\theta,\alpha}\Bigr)(R_{t,n}^{\theta,\alpha}-R_{t,0}^{\theta,\alpha})\mathcomma
		\end{equation*}
		and
		\begin{equation*}
			\norm{g_{n,\theta,\alpha}-g_{0,\theta,\alpha}}_{\P_0,2} \le C' \sum_{i\le t}\norm{\pi_{i,n}-\pi_{i,0}}_{\P_0,2} \mathcomma
		\end{equation*}
		for some $C'$ which does not depend on $(\theta,\alpha)$.
		Therefore
		\begin{equation*}
			\sup_{\theta\in\Theta}\norm{f_{n,\theta,\alpha}}_{\P_0,2} \le C' \sum_{i\le t}\norm{\pi_{i,n}-\pi_{i,0}}_{\P_0,2}=o_p(n^{-1/4})=o_p(1)\mathperiod
		\end{equation*}
		Finally, given that $\mathcal{F}_{n,2}$ is Donsker, we have that $\norm{\G_n}_{\mathcal{F}_{n,2}}=o_p(1)$ by \cref{corollary:maximal-inequality-shortcut}, and 
		$$\sup_{\theta\in\Theta}\abs{(\P_n-\P_0)\{\Dcart(\mu_0^{\theta,\alpha},\pi_n;\theta,\alpha)-\Dcart(\mu_0^{\theta,\alpha},\pi_0;\theta,\alpha)\}}=o_p(n^{-1/2})\mathcomma$$
		for all $t=1,\ldots,T$.
		Taking the sum over $1,\ldots,T$, we have that 
		$$\sup_{\theta\in\Theta}\abs{(\P_n-\P_0)\{\Dcar(\mu_0^{\theta,\alpha},\pi_n;\theta,\alpha)-\Dcar(\mu_0^{\theta,\alpha},\pi_0;\theta,\alpha)\}}=o_p(n^{-1/2})\mathcomma$$
		as desired.
	\end{proof}
	
	\subsection{Corollary 1}
	
	Before proving results as $\alpha_n=o_p(n^{-1/2})$ we introduce some auxiliary results.
	
	\begin{lemma}
		\label{lemma:uipw-asymp-uniform-alpha}
		Fix $\alpha_\epsilon>0$.
		Suppose \cref{assumption:identification:positivity,assumption:bounded} hold and that $\pi_{t,n}$ is sufficiently undersmoothed such that \cref{eqn:hal-undersmooth-condition} holds.
		Assume moreover that \cref{assumption:cadlag,,assumption:hal-space} hold uniformly over $\theta\in\Theta$ and $\alpha\in[0,\alpha_\epsilon]$.
		Then,
		\begin{enumerate}[label=(\alph*),ref=\thelemma(\alph*)]
			\item \label{lemma:uipw-asymp-uniform-alpha:uniform-remainder} \textbf{Uniform remainder control.} Writing the expansion from \cref{theorem:uipw-asym},
			$$\psi_n(\theta;\alpha)-\psi_0(\theta;\alpha)=(\P_n-\P_0)\phi_\psi(O;\eta_0,\theta,\alpha) + R_n(\theta,\alpha)$$
			we have
			$$\sup_{\theta\in\Theta,\alpha\in[0,\alpha_\epsilon]} \lvert R_n(\theta;\alpha) \rvert= o_p(n^{-1/2})\mathperiod$$
			\item \label{lemma:uipw-asymp-uniform-alpha:equicontinuity}\textbf{Stochastic equicontinuity at $\alpha=0$.} For any $\alpha_n=o_p(1)$,
			$$\sup_{\theta\in\Theta} \lvert (\P_n-\P_0)\left\{\phi_\psi(O;\eta_0,\theta,\alpha_n) - \phi_\psi(O;\eta_0,\theta,0) \right\} \rvert=o_p(n^{-1/2})$$
			\item \label{lemma:uipw-asymp-uniform-alpha:drift-bound}\textbf{Uniform drift bound.} For any $\alpha\in[0,\alpha_\epsilon]$,
			$\sup_{\theta\in\Theta} \lvert \psi_0(\theta;\alpha)-\psi_0(\theta;0)\rvert \le C\alpha$,
			for some finite constant $C$.
		\end{enumerate}
	\end{lemma}
	
	\begin{proof}[Proof of \cref{lemma:uipw-asymp-uniform-alpha}]
		The first statement follows by replicating the bounds obtained in \cref{lemma:r1-bounded-uniform,,lemma:r2-bounded-uniform} uniformly over $\alpha\in[0,\alpha_\epsilon]$.
		These hold under the assumption that \cref{assumption:cadlag,,assumption:hal-space} hold uniformly in $\Theta\times[0,\alpha_\epsilon]$ because $\alpha$ only enters through the multipliers $(t/T)^\alpha\in [(1/T)^{\alpha_\epsilon},1]$.
		Hence, all arguments about bounded weight and Donsker envelope arguments go through with constraints replaced by suprema over $\alpha\in[0,\alpha_\epsilon]$. 
		
		We proceed with stochastic equicontinuity.
		Note that $d_t(\alpha)\coloneq (t/T)^\alpha$ is differentiable and Lipschitz on $[0,\alpha_\epsilon]$ with
		$$\lvert d_t(\alpha)-d_t(0) \rvert \le \alpha \sup_{\tilde\alpha\in[0,\alpha_\epsilon]}\lvert\partial_\alpha d_t(\tilde\alpha) \rvert = \alpha \sup_{\tilde\alpha\in[0,\alpha_\epsilon]} d_t(\tilde\alpha)\lvert \log(t/T) \rvert \le \alpha \lvert \log(t/T)\rvert\mathperiod$$
		Thus $\max_{t\le T}\lvert d_t(\alpha)-1\rvert \lesssim \alpha$.
		Moreover, the form of $\phi_\psi$ as bounded products of weights and bounded $Y$, there exists some $C<\infty$ such that
		$$\sup_{\theta\in\Theta} \lVert \phi_\psi(\cdot;\eta_0,\theta,\alpha)-\phi_\psi(\cdot;\eta_0,\theta,0) \rVert_{2,\P_0} \le C\alpha, \quad \alpha\in[0,\alpha_\epsilon]\mathperiod$$
		Then, define the class
		$$\mathcal F_\delta \coloneq \big\{f_{\theta,\alpha}(\cdot)=\phi_\psi(\cdot;\eta_0,\theta,\alpha)-\phi_\psi(\cdot;\eta_0,\theta,0):\theta\in\Theta,\alpha\in[0,\delta] \big\}\mathperiod$$
		For each fixed $\delta\le\alpha_\epsilon$, $\mathcal F_\delta$ is Donsker by the same arguments used in the proof of \cref{theorem:uipw-asym}.
		Moreover, the previous argument gives $\sup_{f\in\mathcal F_\delta}\lVert f \rVert_{2,\P_0} \lesssim\delta$.
		\cref{corollary:maximal-inequality-shortcut} then implies
		$$\sup_{\theta\in\Theta} \lvert (\P_n-\P_0)\left\{\phi_\psi(O;\eta_0,\theta,\alpha_n) - \phi_\psi(O;\eta_0,\theta,0) \right\} \rvert=o_p(n^{-1/2})\mathcomma$$
		as desired.    
		
		Finally we prove the uniform drift bound.
		By Taylor's Theorem, for each $\theta$,
		$$\psi_0(\theta;\alpha)=\psi_0(\theta;0)+\alpha\partial_\alpha \psi_0(\theta;0)+\frac12 \alpha^2 \partial^2_\alpha\psi_0(\theta;\tilde\alpha), \quad \tilde\alpha\in(0,\alpha)\mathperiod$$
		Standard bounding arguments leveraging \cref{assumption:identification:positivity,,assumption:bounded,,assumption:q-smooth} give us the following remainder bound:
		$\sup_{\theta\in\Theta}\lvert \tfrac12\alpha^2\partial^2_\alpha\psi_0(\theta;\tilde\alpha) \rvert\lesssim \alpha^2$.
		Hence, $\sup_{\theta\in\Theta}\lvert \psi_0(\theta;\alpha)-\psi_0(\theta;0)\rvert \le C\alpha$ for $\alpha\in[0,\alpha_\epsilon]$.
	\end{proof}
	
	\begin{proof}[Proof of \cref{corollary:uipw-asymp-alpha}]
		We consider the decomposition:
		$$\psi_n(\theta;\alpha_n) - \psi_0(\theta,0) = \underbrace{\big\{\psi_n(\theta,\alpha_n) - \psi_0(\theta; \alpha_n)\big\}}_{\text{estimation error}} - \underbrace{\big\{\psi_0(\theta;\alpha_n) - \psi_0(\theta; 0)\big\}}_{\text{estimand drift}}\mathperiod$$
		
		The first term represents the estimation error from \cref{theorem:uipw-asym}.
		Under \cref{assumption:identification:positivity,,assumption:q-smooth,,assumption:bounded}, sufficient undersmoothing \cref{eqn:hal-undersmooth-condition}, and uniform \cref{assumption:cadlag,,assumption:hal-space} over $\Theta\times[0,\alpha_\epsilon]$ we have, by \cref{theorem:uipw-asym} and the uniform control over $[0,\alpha_\epsilon]$ from \cref{lemma:uipw-asymp-uniform-alpha}, 
		$$\sup_{\theta\in\Theta} \lvert \{\psi_n(\theta;\alpha_n)-\psi_0(\theta;\alpha_n) - (\P_n-\P_0)\phi_\psi(O;\eta_0,\theta,\alpha_n\} \rvert=o_p(n^{-1/2})\mathperiod$$
		That is, that \cref{theorem:uipw-asym}'s expansion holds at a random $\alpha_n$ that lies within $[0,\alpha_\epsilon]$ with high probability.
		Next, by the stochastic continuity result of \cref{lemma:uipw-asymp-uniform-alpha},
		$$\sup_{\theta\in\Theta} \lvert (\P_n-\P_0)\left\{\phi_\psi(O;\eta_0,\theta,\alpha_n)-\phi_\psi(O;\eta_0,\theta,0) \right\} \rvert = o_p(n^{-1/2})\mathperiod$$
		Combining the two displayed equations gives the uniform linearization
		$$\sup_{\theta\in\Theta}\lvert \{\psi_n(\theta;\alpha_n)-\psi_0(\theta;0)\} - (\P_n-\P_0)\phi_\psi(O;\eta_0,\theta,0) \rvert\mathperiod$$
		Hence, the estimation error is uniformly linear with influence function equal to the influence function from \cref{theorem:uipw-asym} with fixed $\alpha=0$.
		
		Next we control the estimand drift.
		By \cref{lemma:uipw-asymp-uniform-alpha}, for $\alpha\in[0,\alpha_\epsilon]$, $\sup_{\theta\in\Theta}\lvert\psi_0(\theta;\alpha)-\psi_0(\theta;0)\rvert \le C \alpha$.
		Therefore, as $\alpha_n=o_p(n^{-1/2})$,
		$$\sup_{\theta\in\Theta}\lvert \psi_0(\theta;\alpha_n)-\psi_0(\theta;0)\rvert\le C \alpha_n = o_p(n^{-1/2})\mathperiod$$
		
		Combining these two results into the original decomposition gives
		$$\psi_n(\theta;\alpha_n)-\psi_0(\theta;0)=(\P_n-\P_0)\phi_\psi(O;\eta_0,\theta,0)+o_p(n^{-1/2})$$
		uniformly in $\Theta$.
		That is,
		\begin{equation}
			\label{eqn:psi-n-alpha-n-unif}
			\sup_{\theta\in\Theta}\lvert\psi_n(\theta;\alpha_n)-\psi_0(\theta;0)-(\P_n-\P_0)\phi_\psi(O;\eta_0,\theta,0) \rvert=o_p(n^{-1/2})\mathperiod    
		\end{equation}
		
		Finally, we prove weak convergence in $\ell^\infty(\Theta)$.
		\cref{theorem:uipw-asym} yields
		$$\sqrt{n}(\P_n-\P_0)\phi_\psi(O;\eta_0,\theta,0)\dto \G(\theta;0)$$
		in $\ell^\infty(\Theta)$.
		Multiplying the previous linearization (\cref{eqn:psi-n-alpha-n-unif}) by $\sqrt{n}$ gives
		$$\sup_{\theta\in\Theta}\lvert \sqrt{n}\left\{\psi_n(\theta;\alpha_n)-\psi_0(\theta;0) \right\} - \sqrt{n}(\P_n-\P_0)\phi_\psi(O;\eta_0,\theta,0)\rvert = o_p(1)\mathperiod$$
		Hence, by Slutsky's theorem,
		$$\sqrt{n}\left\{\psi_n(\theta;\alpha_n)-\psi_0(\theta;0)\right\} \dto \G(\theta;0)$$
		in $\ell^\infty(\Theta)$.
		Finally, the pointwise result at a fixed $\theta_0$ is a special case of this result. 
	\end{proof}
	
	\section{Proofs Regarding Policy Optimization}
	\label{sec:proof-optimization}
	
	\subsection{Supporting Results}
	
	Before proving \cref{theorem:opt-theta-asym} and \cref{corollary:opt-theta-asym-alpha}, we establish some supporting lemmata regarding the derivatives of the regimen-response curve.
	
	\begin{lemma}
		\label{lemma:deriv-gradient}
		Suppose \cref{assumption:identification:positivity,,assumption:bounded,,assumption:q-smooth} hold. Then,
		\begin{enumerate}
			\item The map $\theta\mapsto \psi_0(\theta;\alpha)$ is continuously differentiable on $\Theta$, with gradient
			\begin{equation*}
				\nabla_\theta \psi_0(\theta;\alpha) = \P_0\left\{  W_0^{\theta,\alpha}Y \sum_{t=1}^T \frac{(t/T)^\alpha}{W_{t,0}^{\theta,\alpha}\pi_{t,0}}\nabla_\theta q_t^\theta\right\}\mathperiod
			\end{equation*}
			\item The parameter $\nabla_\theta \psi_0(\theta;\alpha)$ is pathwise differentiable with canonical gradient
			\begin{equation*}
				\phi_{\nabla\psi}(O;\eta_0,\theta,\alpha) = \nabla_\theta \phi_\psi(O;\eta_0,\theta,\alpha) - \P_0\{\nabla_\theta \phi_\psi(O;\eta_0,\theta,\alpha) \}\mathcomma
			\end{equation*}
			where $\phi_\psi$ is the canonical gradient given in \cref{lemma:dcar-expectation}.
		\end{enumerate}
	\end{lemma}
	
	\begin{proof}[Proof of \cref{lemma:deriv-gradient}]
		We have $\nabla_\theta W_{t,0}^{\theta,\alpha}=(t/T)^\alpha {\nabla_\theta q_t^\theta(A_t,S_t)}{\pi_{t,0}^{-1}(A_t\mid H_t)}$.
		Then,
		\begin{align*}
			\nabla W_0^{\theta,\alpha}&= W_0^{\theta,\alpha} \sum_{t=1}^T \frac{\nabla_\theta W_{t,0}^{\theta,\alpha}}{W_{t,0}^{\theta,\alpha}} = W_0^{\theta,\alpha}\sum_{t=1}^T \frac{(t/T)^\alpha}{W_{t,0}^{\theta,\alpha}\pi_{t,0}}\nabla_\theta q_t^\theta\mathperiod
		\end{align*}
		Under \cref{assumption:q-smooth}, $W_0^{\theta,\alpha}Y$ is almost surely differentiable in $\theta$ with derivative dominated an integrable envelope uniformly in $\theta$.
		Hence, by the dominated convergence theorem
		$$\nabla_\theta \P_0(W_0^{\theta,\alpha}Y)=\P_0(\nabla_\theta W_0^{\theta,\alpha}Y)\nabla_\theta \P_0(W_0^{\theta,\alpha}Y) = \P_0\left\{   W_0^{\theta,\alpha}Y\sum_{t=1}^T \frac{(t/T)^\alpha}{W_{t,0}^{\theta,\alpha}\pi_{t,0}}\nabla_\theta q_t^\theta \right\}\mathperiod$$
		
		The parameter $\nabla_\theta \psi_0(\theta;\alpha)$ is pathwise differentiable.
		Under \cref{assumption:q-smooth}, $\theta\mapsto\psi_0(\theta;\alpha)$ is continuously differentiable and its canonical gradient is differentiable in $\theta$.
		Because differentiation with respect to $\theta$ is a linear and continuous operator, the canonical gradient of $\nabla_\theta \psi_0(\theta;\alpha)$ is given by the derivative of the canonical gradient of $\psi_0(\theta;\alpha)$, centered to have mean zero.
		That is,
		\begin{equation*}
			\phi_{\nabla \psi}(O;\eta_0,\theta,\alpha) = \nabla_\theta \phi_{\psi}(O;\eta_0,\theta,\alpha) - \P_0\left\{\nabla_\theta \phi_{\psi}(O;\eta_0,\theta,\alpha)\right\}\mathperiod 
		\end{equation*}
		Recall that the canonical gradient of the regimen-response curve is given by $\phi_{\psi}(O;\eta_0,\theta,\alpha)=W_0^{\theta,\alpha}(O)Y - \Dcar(O;\eta_0,\theta,\alpha)-\psi_0(\theta;\alpha)$.
		All that remains is to derive $\nabla_\theta \phi_{\psi}(O;\eta_0,\theta,\alpha)$.
		The augmentation is given by $\Dcar(O;\eta_0,\theta,\alpha)=\sum_{t=1}^T\Dcart(O;\eta_0,\theta,\alpha)$.
		We have
		\begin{align}
			\begin{split}
				\label{eqn:dcar-deriv}
				\nabla_\theta \Dcart(O;\eta_0,\theta,\alpha)
				&= \Big(\sum_{j<t} \Xi_j^{\theta,\alpha}\Big) \Dcart(O;\eta_0,\theta,\alpha) \\
				&\quad + (t/T)^\alpha \Big(\prod_{j<t} W_{j,0}^{\theta,\alpha}\Big) \sum_{a\in\{0,1\}} \frac{I(A_t=a)-\pi_{t,0}(a\mid H_t)}{\pi_{t,0}(a\mid H_t)} \\
				&\quad\quad \times \left[\left\{\nabla_\theta q_t^\theta(a, H_t)\right\}\mu_t^{\theta,\alpha}(a\mid H_t) + q_t^\theta(a, H_t)\left\{\nabla_\theta \mu_t^{\theta,\alpha}(a\mid H_t)\right\}\right]\mathcomma
			\end{split}
		\end{align}
		where
		\begin{equation}
			\label{eqn:def-xi}
			\Xi_t^{\theta,\alpha} \coloneq \nabla_\theta \log W_{j,0}^{\theta,\alpha} = \frac{(t/T)^\alpha}{W_{t,0}^{\theta,\alpha}\pi_{t,0}}\nabla_\theta q_t^\theta\mathperiod
		\end{equation}
		This term is seen to have mean zero by conditioning on $H_t$.
		Therefore, the canonical gradient is given by
		\begin{align*}
			\phi_{\nabla \psi}(O;\eta_0,\theta,\alpha) = YW^{\theta,\alpha}(O) \sum_{t=1}^T \frac{(t/T)^\alpha}{W_{t,0}^{\theta,\alpha}\pi_{t,0}}\nabla_\theta q_t^\theta - \sum_{t=1}^T\nabla_\theta \Dcart(O;\theta,\alpha) + \nabla_\theta\psi_0(\theta;\alpha)\mathperiod
		\end{align*}
	\end{proof}
	
	\begin{lemma}
		\label{lemma:mu-deriv-cadlag-svn}
		Suppose that \cref{assumption:q-smooth,,assumption:cadlag} hold, then for each $t\in\{1,\ldots,T\}$, $(H_t,A_t)\mapsto
		\nabla_\theta \mu_{t,0}^{\theta,\alpha}(A_t, H_t)$ is c\`al\`ag with finite total SVN. 
	\end{lemma}
	
	\begin{proof}[Proof of \cref{lemma:mu-deriv-cadlag-svn}]
		We prove by backward recursion.
		Let $t=T$.
		Then
		$$\mu_{T,0}^{\theta,\alpha}(A_T,H_T)=\P_0(Y\mid H_T,A_T)\mathcomma$$
		which does not depend on $\theta$ and hence $\nabla_\theta \mu_{T,0}^{\theta,\alpha}=0$, which is trivially c\`adl\`ag with finite total SVN.
		Assume that for some $t < T$, $\nabla_\theta \mu_{t+1,0}^{\theta,\alpha}$ is c\`adl\`ag with finite total SVN. Then,
		\begin{align*}
			\nabla_\theta \mu_{t,0}^{\theta,\alpha}(A_t, H_t) &= 
			\P_0\Big[ \sum_{a_{t+1}\in\{0,1\}} \left\{\nabla_\theta q_{t+1}^{\theta,\alpha}(a_{t+1},H_{t+1}) \right\}\mu_{t+1,0}^{\theta,\alpha}(a_{t+1},H_{t+1}) \\
			&\qquad\qquad\qquad + q_{t+1}^{\theta,\alpha}(a_{t+1},H_{t+1})\left\{\nabla_\theta \mu_{t+1,0}^{\theta,\alpha}(a_{t+1},H_{t+1}) \right\} \mid A_t,H_t\Big]\mathperiod
		\end{align*}
		Each term inside the conditional expectation is a product of two c\`adl\`ag functions of $H_{t+1}$ with finite total SVN by either the induction hypothesis ($\nabla_\theta\mu_{t+1}^{\theta,\alpha}$), \cref{assumption:q-smooth} ($q_t^{\theta,\alpha}$ and $\nabla_\theta q_t^{\theta,\alpha}$), or \cref{assumption:cadlag} ($\mu_{t+1}^{\theta,\alpha}$). 
		It follows that $\nabla_\theta \mu_{t,0}^{\theta,\alpha}$ is c\`adl\`ag with finite total SVN.
		By backward induction, the result holds for all $t$.
	\end{proof}
	
	\begin{lemma}
		\label{lemma:psi-deriv-asymp}
		Suppose that \cref{assumption:identification:positivity,assumption:bounded,assumption:q-smooth} hold.
		Moreover, suppose that $\pi_{t,n}$ is sufficiently undersmoothed, satisfying \cref{eqn:hal-undersmooth-condition}, that \cref{assumption:cadlag} holds, and now that \cref{assumption:hal-space-deriv} holds.
		Then, for any $\theta$ on the interior of $\Theta$,
		\begin{equation*}
			\nabla_\theta \psi_n(\theta;\alpha)-\nabla_\theta \psi_0(\theta;\alpha) = (\P_n-\P_0)\{\phi_{\nabla \psi}(O;\eta_0,\theta,\alpha)\}+ o_p(n^{-1/2})\mathperiod%
		\end{equation*}
	\end{lemma}
	
	\begin{proof}
		Recall that 
		\begin{align*}
			\psi_n(\theta;\alpha)-\psi_0(\theta;\alpha) &= (\P_n-\P_0) \{\phi(O;\eta_0,\theta,\alpha)\}+ \P_n\{\Dcar(\mu_0^{\theta,\alpha},\pi_n;\theta,\alpha)\} \\
			&\qquad\, {}+ (R_{1,n}+R_{2,n})(\theta,\alpha) \mathperiod 
		\end{align*}
		Hence, 
		\begin{align*}
			\nabla_\theta \psi_n(\theta;\alpha)-\nabla_\theta \psi_0(\theta;\alpha) &= (\P_n-\P_0) \{\nabla_\theta \phi(O; \eta_0,\theta,\alpha)\} \\
			&\qquad\, {} + \P_n\{\nabla_\theta \Dcar(\mu_0^{\theta,\alpha},\pi_n;\theta,\alpha)\} \\
			&\qquad\, {} + (\nabla_\theta R_{1,n}+ \nabla_\theta R_{2,n})(\theta,\alpha) \mathperiod
		\end{align*}
		The leading term is equal to $(\P_n-\P_0)\{\phi_{\nabla \psi}(O; \eta_0,\theta,\alpha)\}$.
		We show all of the remaining terms are $o_p(n^{-1/2})$:
		\setlist[description]{font=\normalfont}
		\begin{description}
			\item[$\nabla_\theta R_{1,n}$] is the empirical process term, $(\P_n-\P_0)\{\nabla_\theta W_n(O;\theta,\alpha)Y-\nabla_\theta W_0(O;\theta,\alpha)Y\}$.
			Under \cref{assumption:identification:positivity,assumption:bounded,assumption:q-smooth,assumption:cadlag}, $O\mapsto\nabla_\theta W_n(O;\theta,\alpha)Y-\nabla_\theta W_0(O;\theta,\alpha)Y$ lies in a fixed Donsker class.
			Because $\norm{\pi_{t,n}-\pi_{t,0}}_{2,\P_0}=o_p(1)$ for all $t$, $\|\nabla_\theta W_n(O;\theta,\alpha)Y-\nabla_\theta W_0(O;\theta,\alpha)Y\|_{2,\P_0}=o_p(1)$, and $\nabla_\theta R_{1,n}(\theta,\alpha)=o_p(n^{-1/2})$ by \citet[Lemma~19.24]{van_der_vaart_asymptotic_1998}.
			
			\item[$\nabla_\theta R_{2,n}$] is equal to $(\P_0-\P_n)\{\nabla_\theta \Dcar(\mu_0^{\theta,\alpha},\pi_n;\theta,\alpha)\allowbreak-\allowbreak \nabla_\theta\Dcar(\mu_0^{\theta,\alpha},\pi_0;\theta,\alpha)\}$.
			The derivative for the $\pi_{t,0}$ case is given in \cref{eqn:dcar-deriv}.
			Using similar bounding arguments to what was used for $\nabla_\theta R_{n,2}(\theta,\alpha)$, we are able to show that $f_{2,n,t}(O)\coloneq\nabla_\theta \Dcart(\mu_0^{\theta,\alpha},\pi_n;\theta,\alpha)\allowbreak-\allowbreak \nabla_\theta\Dcart(\mu_0^{\theta,\alpha},\pi_0;\theta,\alpha)$ is bounded and belongs to a Donsker class.
			Furthermore, using a telescoping decomposition similar to the proof of \cref{lemma:r2-bounded-uniform} we can show that $\lVert f_{2,n,t}(O;\theta,\alpha)\rVert_{2,\P_0}\lesssim \sum_{i\le t} \lVert\pi_{i,n}-\pi_{i,0}\rVert_{2,\P_0}=o_p(1)$.
			It follows that $\nabla_\theta R_{n,3}(\theta,\alpha)=o_p(n^{-1/2})$ under empirical process theory \citep[Lemma 19.24]{van_der_vaart_asymptotic_1998}.
			
			\item[$\P_n\{\nabla_\theta \Dcar(\mu_0^{\theta,\alpha},\pi_n;\theta,\alpha)\}$] is the bias term, which is once again controlled by undersmoothing.
			Then, as in the proof of \cref{corollary:dcar-bounded-uniform}, the empirical $D_\mathrm{CAR}$ remainder at time \(t\) can be written as
			\begin{equation*}
				\Dcart(\mu_0^{\theta,\alpha},\pi_n;\theta,\alpha)=
				f_{t,n}^{\theta,\alpha}(H_t)
				\{A_t-\pi_{t,n}(1\mid H_t)\}\mathcomma    
			\end{equation*}
			where
			\begin{equation*}
				f_{t,n}^{\theta,\alpha}(H_t) =\Bigl(\prod_{i<t} W_{i,n}^{\theta,\alpha}\Bigr)
				(t/T)^\alpha \Biggl\{ \frac{\mu_{t,0}^{\theta,\alpha}(1,H_t)q_t^\theta(1,S_t)}{\pi_{t,n}(1\mid H_t)} -
				\frac{\mu_{t,0}^{\theta,\alpha}(0,H_t)q_t^\theta(0,S_t)}{\pi_{t,n}(0\mid H_t)
				}\Biggr\}\mathperiod    
			\end{equation*}
			Since \(\pi_{t,n}\) is estimated independently of \(\theta\), differentiation
			with respect to \(\theta\) only acts on the weights, the target policy, and the
			forward regressions.  
			Hence, letting $g_{t,n}^{\theta,\alpha}(H_t) \coloneq \nabla_\theta f_{t,n}^{\theta,\alpha}(H_t)$, we have
			\begin{equation*}
				\nabla_\theta \Dcart(\mu_0^{\theta,\alpha},\pi_n;\theta,\alpha)
				=
				g_{t,n}^{\theta,\alpha}(H_t)
				\{A_t-\pi_{t,n}(1\mid H_t)\}\mathperiod        
			\end{equation*}
			Deriving, we obtain
			\begin{align*}
				&g_{t,n}^{\theta,\alpha}(H_t) \\
				&=
				\Bigl(\prod_{i<t} W_{i,n}^{\theta,\alpha}\Bigr)
				(t/T)^\alpha
				\Biggl[
				\Bigl(
				\sum_{j<t}
				\nabla_\theta \log W_{j,n}^{\theta,\alpha}
				\Bigr)
				\Biggl\{
				\frac{\mu_{t,0}^{\theta,\alpha}(1,H_t)q_t^\theta(1,S_t)}{\pi_{t,n}(1\mid H_t)}
				-
				\frac{\mu_{t,0}^{\theta,\alpha}(0,H_t)q_t^\theta(0,S_t)}{\pi_{t,n}(0\mid H_t)}
				\Biggr\}
				\\
				&\qquad\qquad
				+
				\frac{
					\{\nabla_\theta\mu_{t,0}^{\theta,\alpha}(1,H_t)\}
					q_t^\theta(1,S_t)
					+
					\mu_{t,0}^{\theta,\alpha}(1,H_t)
					\{\nabla_\theta q_t^\theta(1,S_t)\}
				}{
					\pi_{t,n}(1\mid H_t)
				}
				\\
				&\qquad\qquad
				-
				\frac{
					\{\nabla_\theta\mu_{t,0}^{\theta,\alpha}(0,H_t)\}
					q_t^\theta(0,S_t)
					+
					\mu_{t,0}^{\theta,\alpha}(0,H_t)
					\{\nabla_\theta q_t^\theta(0,S_t)\}
				}{
					\pi_{t,n}(0\mid H_t)
				}
				\Biggr]\mathcomma
			\end{align*}
			and
			\begin{equation*}
				\nabla_\theta\log W_{j,n}^{\theta,\alpha}
				= \frac{(j/T)^\alpha \nabla_\theta q_j^\theta(A_j,S_j)}{(j/T)^\alpha q_j^\theta(A_j,S_j)+\{1-(j/T)^\alpha\}\pi_{j,n}(A_j\mid H_j)}\mathperiod    
			\end{equation*}
			
			Let \(g_{t,0}^{\theta,\alpha}\) denote the population version obtained by replacing \(\pi_{\cdot,n}\) and \(W_{\cdot,n}^{\theta,\alpha}\) by \(\pi_{\cdot,0}\) and \(W_{\cdot,0}^{\theta,\alpha}\), respectively. 
			By \cref{assumption:hal-space-deriv}, \(g_{t,0}^{\theta,\alpha}\) is approximated by the active HAL basis to order \(O_p(n^{-1/4})\) in \(L^2(\P_0)\): $\lVert g_{t,0}^{\theta,\alpha} - \tilde g_{t,0}^{\theta,\alpha}\lVert_{2,\P_0}=O_p(n^{-1/4})$. Moreover, by consistency of the HAL propensity score estimator, $\lVert g_{t,n}^{\theta,\alpha} - g_{t,0}^{\theta,\alpha} \rVert_{2,\P_0}=o_p(n^{-1/4})$.  
			By the triangle inequality, we obtain $\lVert g_{t,n}^{\theta,\alpha} - \tilde g_{t,n}^{\theta,\alpha}] \rVert_{2,\P_0} = O_p(n^{-1/4})$.
			Moreover, $g_{t,n}^{\theta,\alpha}$ term is finite product of c\`adl\`ag functions with finite total SVN (see \cref{assumption:q-smooth,assumption:cadlag}, \cref{lemma:mu-deriv-cadlag-svn}, and the proof of \cref{corollary:dcar-bounded-uniform}) and therefore is c\`adl\`ag with finite total SVN.
			Hence, conditions are met to apply \cref{lemma:d-bounded-uniform} at $(\theta,\alpha)$.
			
			Applying \cref{lemma:d-bounded-uniform} yields
			\begin{equation*}
				\P_n \nabla_\theta\Dcart(\mu_0^{\theta,\alpha},\pi_n;\theta,\alpha)=
				\P_n g_{t,n}^{\theta,\alpha}(H_t) \{A_t-\pi_{t,n}(1\mid H_t)\}
				= o_p(n^{-1/2})    
			\end{equation*}
			for each fixed $t$.  Since $T$ is fixed,
			\begin{equation*}
				\P_n\nabla_\theta\Dcar(\mu_0^{\theta,\alpha},\pi_n;\theta,\alpha)
				=\sum_{t=1}^T\P_n\nabla_\theta\Dcart(\mu_0^{\theta,\alpha},\pi_n;\theta,\alpha)
				=o_p(n^{-1/2})\mathperiod        
			\end{equation*}
		\end{description}
		Combining, we have $\nabla_\theta \psi_n(\theta;\alpha)-\nabla_\theta \psi_0(\theta;\alpha)= (\P_n-\P_0) \{\phi_{\nabla \psi}(O; \eta_0, \theta,\alpha)\}+o_p(n^{-1/2})$.
		This completes the proof.
	\end{proof}
	
	\begin{lemma}
		\label{lemma:hessian-est-properties}
		Suppose that \cref{assumption:identification:positivity,,assumption:bounded,,assumption:q-smooth} hold.
		Then,
		\begin{enumerate}[label=(\alph*),ref=\thelemma(\alph*)]
			\item The Hessian $H_0(\theta_0;\alpha_0)\coloneq\nabla_\theta^2\psi_0(\theta_0;\alpha_0)$ exists and is finite for any fixed $\theta_0,\alpha_0$.
			\item The Hessian is Lipschitz continuous in $\Theta$:
			$$\|\nabla_\theta^2\psi_\cdot(\theta_1;\alpha)-\nabla_\theta^2\psi_\cdot(\theta_2;\alpha)\|\le L\|\theta_1-\theta_2\|\mathcomma$$ for some $L<\infty$.
			\item The Hessian of $\psi_n$ converges uniformly in $\Theta$:
			$$\sup_{\theta\in\Theta}\|\nabla_\theta^2 \psi_n(\theta;\alpha)-\nabla_\theta^2\psi_0(\theta;\alpha)\|=o_p(1)\mathperiod$$
		\end{enumerate}
	\end{lemma}
	
	\begin{proof}[Proof of \cref{lemma:hessian-est-properties}]
		First, recall from \cref{lemma:deriv-gradient} that $\nabla_\theta\psi_0(\theta;\alpha)=\P_0\{W_0^{\theta,\alpha}Yf^\theta(O)\}$ where $f^\theta(O)$ is a smooth function of $\{q_t^\theta,\nabla_\theta q_t^\theta\}_{t=1}^T$.
		Differentiating once more in $\theta$ introduces terms involving $\nabla_\theta^2 q_t^\theta$ and $\nabla_\theta q_t^\theta (\nabla_\theta q_j^\theta)\trans$
		All such terms are c\`adl\`ag with finite total SVN under \cref{assumption:q-smooth}, and the Dominated Convergence Theorem justifies $\nabla_\theta^2\psi_0(\theta;\alpha)=\P_0\{\nabla_\theta^2 W_0^{\theta,\alpha}Y\}$.
		
		We proceed with establishing a Lipschitz condition by bounding the third derivative and applying the Mean Value Theorem to the Hessian.
		Under \cref{assumption:q-smooth} we have
		$$\nabla_\theta^k\psi_\cdot(\theta;\alpha)=\P_\cdot(Y \nabla_\theta^k W_\cdot^{\theta,\alpha})\mathcomma$$
		for any $k\in\{1,2,3\}$ and $\cdot\in\{0,n\}$ by the Dominated Convergence Theorem.
		Recall that $W_\cdot^{\theta,\alpha}=\prod_{t=1}^T W_{t,\cdot}^{\theta,\alpha}$ for $W_{t,\cdot}^{\theta,\alpha}=1+(t/T)^\alpha\{q_t^\theta(A_t,S_t)\pi_{t,\cdot}^{-1}(A_t\mid H_t)-1\}$.
		Under \cref{assumption:q-smooth,assumption:identification:positivity}, each weight and its inverse $(W_{t,\cdot}^{\theta,\alpha})^{-1}$ are uniformly bounded over $\Theta$.
		Then, as $\pi_{t,\cdot}$ does not depend on $\theta$, we have
		\begin{align*}
			\nabla_\theta^k W_{t,\cdot}^{\theta,\alpha} &= (t/T)^\alpha \frac{\nabla_\theta^k q_t^\theta(A_t,S_t)}{\pi_{t,\cdot}(A_t\mid H_t)}\mathperiod
		\end{align*}
		Moreover, \cref{assumption:q-smooth} gives $\sup_{\theta\in\Theta}\|\nabla_\theta^k q_t^\theta(A_t,S_t)\| < \infty$ almost surely for $k\in\{1,2,3\}$.
		Hence, each derivative $\nabla_\theta^k W_{t,\cdot}^{\theta,\alpha}$ is uniformly bounded over $\Theta$ for $k\in\{1,2,3\}$.
		Repeated application of the product rule for derivatives shows that $\nabla_\theta W_\cdot^{\theta,\alpha}$ is a finite sum of finite products of
		$$\frac{\nabla_\theta W_{r,\cdot}^{\theta,\alpha}}{W_{r,\cdot}^{\theta,\alpha}},\quad\frac{\nabla^2_\theta W_{s,\cdot}^{\theta,\alpha}}{W_{s,\cdot}^{\theta,\alpha}},\quad\text{and}\quad\frac{\nabla^3_\theta W_{t,\cdot}^{\theta,\alpha}}{W_{t,\cdot}^{\theta,\alpha}}\mathcomma$$
		multiplied by $W_\cdot^{\theta,\alpha}$.
		Each of these components is uniformly bounded over $\Theta$; hence
		$$\sup_{\theta\in\Theta} \|\nabla_\theta^3 W_\cdot^{\theta,\alpha} \|<\infty$$
		almost surely.
		Then, under \cref{assumption:bounded}, $\|Y\|_\infty<\infty$, and hence $$\sup_{\theta\in\Theta}\|\nabla_\theta^3\psi_\cdot(\theta;\alpha)\| \le \|Y\|_\infty\sup_{\theta\in\Theta} \|\nabla_\theta^3 W_\cdot^{\theta,\alpha} \|<\infty$$
		almost surely.
		Finally, by the Mean Value Theorem. for any $\theta_1,\theta_2\in\Theta$,
		$$\|\nabla_\theta^2\psi_\cdot(\theta_1;\alpha)-\nabla_\theta^2\psi_\cdot(\theta_2;\alpha)\|\le \sup_{\theta\in\Theta}\|\nabla_\theta^3 \psi_\cdot(\theta;\alpha)\|\, \|\theta_1-\theta_2\|\mathperiod$$
		Since the third derivative is uniformly bounded, it follows that the Hessian is globally Lipschitz continuous.
		
		Finally, we prove uniform convergence.
		We denote $\nabla_\theta^2\psi_n(\theta;\alpha)=\P_n\{H(O;\theta,\alpha,\pi_n)\}$ and $\nabla_\theta^2\psi_0(\theta;\alpha)=\P_0\{H(O;\theta,\alpha,\pi_0)\}$.
		For brevity, let $H(\theta;\pi)\coloneq H(O;\theta,\alpha,\pi)$.
		Consider the decomposition:
		\begin{multline*}
			\nabla_\theta^2\psi_n(\theta;\alpha)-\nabla_\theta^2\psi_0(\theta;\alpha)=(\P_n-\P_0)H(\theta,\pi_0)  +(\P_n-\P_0)\big\{H(\theta,\pi_n)-H(\theta,\pi_0)\big\}\\ + \P_0\big\{H(\theta,\pi_n)-H(\theta,\pi_0) \big\}\mathperiod
		\end{multline*}
		We proceed by controlling the three terms separately.
		\begin{itemize}
			\item 
			The first term is an empirical process term.
			Define the deterministic class
			$$\mathcal H_0 \coloneq \big\{H(\cdot;\theta,\alpha,\pi_0):\theta\in\Theta \big\}\mathperiod$$
			We claim that $\mathcal H_0$ is Glivenko-Cantelli with an integrable envelope.
			The integrable envelope exists following similar logic used to bound 
			$\nabla_\theta^3\psi_\cdot(\theta,\alpha)$ previously in this proof.
			Moreover, under \cref{assumption:q-smooth}, $\mathcal H_0$ inherits bounded sectional variational norm and therefore is Glivenko-Cantelli under \cref{assumption:bounded}.
			Hence,
			$$\sup_{\theta\in\Theta}\|(\P_n-\P_0)H(\cdot;\theta,\alpha,\pi_0)\|=o_p(1)\mathperiod$$
			\item 
			The second term also an empirical process term. We have that $H(O;\theta,\pi_0)$ is Donsker and $\lVert H(O;\theta,\pi_n)-H(O;\theta,\pi)\rVert_{\P_0,2} = o_p(1)$.
			It follows that $(\P_n-\P_0)\{H(\theta,\pi_n)-H(\theta,\pi)\}=o_p(n^{-1/2})=o_p(1)$ by \citet[Lemma~19.24]{van_der_vaart_asymptotic_1998}.
			\item 
			The final term represents the plug-in bias. 
			First, note that for any $x,y\in[\epsilon_\pi/2,1]$,
			\begin{equation*}
				\lvert x^{-1}-y^{-1}\rvert = \frac{\lvert x-y\rvert}{xy} \le \frac4{\epsilon_\pi^2}\lvert x-y\rvert\mathperiod
			\end{equation*}
			Because $\pi$ enters $H(O;\theta,\alpha,\pi)$ only through a finite number of factors $\pi(\cdot\mid H_t)^{-1}$ inside the weights and their $\theta$-derivatives, repeated application of the previous inequality and the product rule yields the existence of an integrable envelope $G$ such that
			$$\sup_{\theta\in\Theta}\lVert H(O;\theta,\alpha,\pi_n)-H(O;\theta,\alpha,\pi_0) \rVert\le G(O) \max_{t\le T} \lvert \pi_{t,n}(1\mid H_t)-\pi_{t,0}(1\mid H_t) \rvert\mathperiod$$
			That is, that $H$ is Lipschitz in $\pi$.
			Then, by Cauchy-Schwawrz, we have
			$$\sup_{\theta\in\Theta}  \lVert \P_0\left\{H(\theta;\pi_n)-H(\theta;\pi_0) \right\} \rVert \le \lVert G \rVert_{2,\P_0} \max_{t\le T}\lVert \pi_{t,n}-\pi_{t,0}\rVert_{2,\P_0}\mathperiod$$
			HAL consistency gives $\lVert\pi_{t,n}-\pi_{t,0} \rVert_{2,\P_0}\to0$ for each $t$; hence, the term is $o_p(1)$.
		\end{itemize}
		Combining the three terms, we have $\sup_{\theta\in\Theta}\|\nabla_\theta^2\psi_n(\theta;\alpha)-\nabla_\theta^2\psi_0(\theta;\alpha)\| = o_p(1)\mathperiod$
	\end{proof}

	\subsection{Theorem 2}
	
	\begin{proof}[Proof of \cref{theorem:opt-theta-asym}]
		Recall that $V_\cdot(\theta;\alpha)=\psi_\cdot(\theta;\alpha)-\mathcal P(\theta)$ where $\cdot\in\{0,n\}$, and that $\theta_0^\star(\alpha)$ is the unique interior maximizer of $V_0(\cdot;\alpha)$, which exists by \cref{assumption:v0-regularity}.
		
		By definition $\theta_n^\star(\alpha)$ satisfies the first order condition $\nabla_\theta V_n\{\theta_n^\star(\alpha);\alpha\}=0$ and likewise $\nabla_\theta V_0\{\theta_0^\star(\alpha);\alpha\}=0$. 
		Moreover, by \cref{lemma:deriv-gradient}, $\theta\mapsto \nabla_\theta V_0(\theta;\alpha)$ is continuously differentiable in a neighborhood of $\theta_0^\star(\alpha)$ with Hessian $H_0(\alpha)=\nabla_\theta^2 V_0\{\theta_0^\star(\alpha);\alpha)$ that is negative definite under \cref{assumption:v0-regularity}. 
		
		By a first order Taylor expansion to $\nabla_\theta V_n(\cdot;\alpha)$ about $\theta_0^\star(\alpha)$, we obtain
		\begin{equation*}
			0 = \nabla_\theta V_n\{\theta_n^\star(\alpha);\alpha\} = \nabla_\theta V_n\{\theta_0^\star(\alpha);\alpha\} + H_0(\alpha)\{\theta_n^\star(\alpha)-\theta_0^\star(\alpha)\} + R_n\mathcomma
		\end{equation*}
		where
		\begin{equation*}
			R_n = \big[\nabla_\theta^2 V_n(\tilde\theta_n;\alpha)- \nabla_\theta^2 V_0\{\theta_0^\star(\alpha);\alpha\} \big]\{\theta_n^\star(\alpha)-\theta_0^\star(\alpha) \}\mathcomma
		\end{equation*}
		and $\tilde\theta_n$ lies on the line segment between $\theta_n^\star(\alpha)$ and $\theta_0^\star(\alpha)$.
		Note that $\nabla_\theta V_n(\theta;\alpha)-\nabla_\theta V_0(\theta;\alpha)=\nabla_\theta\psi_n(\theta;\alpha)-\nabla_\theta\psi_0(\theta;\alpha)$.
		Evaluating the previous equality at $\theta_0^\star(\alpha)$ we obtain
		\begin{equation*}
			\nabla_\theta V_n\{\theta_0^\star(\alpha);\alpha\}= \nabla_\theta \psi_n\{\theta_0^\star(\alpha);\alpha\} - \nabla_\theta \psi_0\{\theta_0^\star(\alpha);\alpha\}\mathperiod
		\end{equation*}
		Substituting into the Taylor expansion and rearranging we obtain
		\begin{equation}
			\label{eqn:theta-n-taylor}
			\theta_n^\star(\alpha)-\theta_0^\star(\alpha) = -H_0^{-1}(\alpha)\big[\nabla_\theta \psi_n\{\theta_0^\star(\alpha);\alpha\} - \nabla_\theta \psi_0\{\theta_0^\star(\alpha);\alpha\} \big] - H_0(\alpha)^{-1}R_n\mathperiod
		\end{equation}
		Under \cref{assumption:identification:positivity,assumption:bounded,assumption:q-smooth,assumption:cadlag,assumption:hal-space,assumption:hal-space-deriv} and sufficient undersmoothing (\cref{eqn:hal-undersmooth-condition}), the leading term satisfies
		$$\nabla_\theta \psi_n\{\theta_0^\star(\alpha);\alpha\} - \nabla_\theta \psi_0\{\theta_0^\star(\alpha);\alpha\} = (\P_n-\P_0)\big[\phi_{\nabla \psi}\{O;\eta_0,\theta_0^\star(\alpha),\alpha\}\big]+ o_p(n^{-1/2})\mathperiod $$
		by \cref{lemma:psi-deriv-asymp}.
		Substituting into \cref{eqn:theta-n-taylor} we obtain
		\begin{equation}
			\label{eqn:theta-n-taylor-2}
			\theta_n^\star(\alpha)-\theta_0^\star(\alpha) 
			= (\P_n-\P_0)\big[{-H_0}(\alpha)^{-1}\phi_{\nabla \psi}\{O;\eta_0,\theta_0^\star(\alpha),\alpha\}\big] - H_0(\alpha)^{-1}R_n + o_p(n^{-1/2})\mathperiod
		\end{equation}
		
		It remains to show that $R_n = o_p(n^{-1/2})$.
		We decompose
		\begin{align*}
			\nabla_\theta^2 V_n(\tilde\theta_n;\alpha)- \nabla_\theta^2 V_0\{\theta_0^\star(\alpha);\alpha\} &= \big[\nabla_\theta^2 V_n(\tilde\theta_n;\alpha)- \nabla_\theta^2 V_0(\tilde\theta_n;\alpha)\big]  \\
			&\quad\, + \big[\nabla_\theta^2 V_0(\tilde\theta_n;\alpha)- \nabla_\theta^2 V_0\{\theta_0^\star(\alpha);\alpha\}\big] 
		\end{align*}
		The first term concerns the convergence of the Hessian of the value function
		We note that $\nabla_\theta^2 V_n(\theta;\alpha)-\nabla_\theta^2V_0(\theta;\alpha)=\nabla_\theta^2 \psi_n(\theta;\alpha)-\nabla_\theta^2\psi_0(\theta;\alpha)$ and that the latter term is controlled by \cref{lemma:hessian-est-properties}.
		Hence,
		\begin{equation*}
			\sup_{\theta\in \Theta} \|\nabla_\theta^2 V_n(\theta;\alpha)-\nabla_\theta^2V_0(\theta;\alpha) \|=\sup_{\theta\in \Theta} \|\nabla_\theta^2 \psi_n(\theta;\alpha)-\nabla_\theta^2\psi_0(\theta;\alpha) \|= o_p(1)\mathcomma
		\end{equation*}
		and
		\begin{equation*}
			\|\nabla_\theta^2 V_n(\tilde\theta_n;\alpha)- \nabla_\theta^2 V_0(\tilde\theta_n;\alpha)\|=o_p(1)\mathperiod
		\end{equation*}
		The second term is controlled by the Lipschitz condition of \cref{lemma:hessian-est-properties} under \cref{assumption:identification:positivity,assumption:q-smooth,eqn:q-bound}:
		\begin{equation*}
			\|\nabla_\theta^2 V_0(\tilde\theta_n;\alpha)- \nabla_\theta^2 V_0\{\theta_0^\star(\alpha);\alpha\}\| \lesssim \|\tilde\theta_n-\theta_0^\star(\alpha)\|=o_p(1)\mathperiod%
		\end{equation*}
		Thus,
		$$\lVert R_n \rVert \le o_p(1) \lVert \theta_n^\star(\alpha)-\theta_0^\star(\alpha) \rVert\mathperiod$$
		Returning to \cref{eqn:theta-n-taylor-2} and taking norms, we have
		\begin{align*}
			\lVert\theta_n^\star(\alpha)-\theta_0^\star(\alpha) \rVert &\le \lVert H_0(\alpha)^{-1} \rVert \lVert (\P_n-\P_0) \phi_{\nabla \psi}\{O;\eta_0,\theta_0^\star(\alpha),\alpha\}\rVert + \lVert H_0(\alpha)^{-1} \rVert \lVert R_n\rVert \\
			&= O_p(n^{-1/2})+o_p(1) \lVert \theta_n^\star(\alpha)-\theta_0^\star(\alpha) \rVert\mathperiod
		\end{align*}
		Hence, 
		$$\lVert\theta_n^\star(\alpha)-\theta_0^\star(\alpha) \rVert=O_p(n^{-1/2})\mathperiod$$
		Combining $\lVert R_n\rVert \le o_p(1)\lVert\theta_n^\star(\alpha)-\theta_0^\star(\alpha) \rVert$ with the previous bound gives $\lVert R_n\rVert = o_p(n^{-1/2})$.
		Finally, substituting into \cref{eqn:theta-n-taylor-2} we obtain
		\begin{equation*}
			\theta_n^\star(\alpha)-\theta_0^\star(\alpha) 
			= (\P_n-\P_0)\big[{-H_0}(\alpha)^{-1}\phi_{\nabla \psi}\{O;\eta_0,\theta_0^\star(\alpha),\alpha\}\big]+ o_p(n^{-1/2})\mathcomma
		\end{equation*}
		establishing asymptotic linearity.
		Asymptotic normality follows immediately by the Central Limit Theorem with covariance
		$$H_0^{-1}(\alpha)\P_0\big[\phi_{\nabla \psi}^{\otimes2}\{O;\eta_0,\theta_0^\star(\alpha),\alpha\}\big]H_0^{-1}(\alpha)\mathperiod$$
	\end{proof}
	
	\subsection{Corollary 2}
	
	Before proving \cref{corollary:opt-theta-asym-alpha} we provide some auxiliary results analogous to \cref{lemma:uipw-asymp-uniform-alpha}.
	First, we require a stronger uniform version of \cref{assumption:v0-regularity}.
	\begin{assumption}\label{assumption:v0-regularity-uniform}
		Let $\alpha_\epsilon>0$. Suppose that
		\begin{enumerate}[label=(\alph*),ref=\thelemma(\alph*)]
			\item For all $\alpha\in[0,\alpha_\epsilon]$, $\theta_0^\star(\alpha)$ exists uniquely in the interior.
			\item There exists a $m>0$ such that for all $\alpha\in[0,\alpha_\epsilon]$ and all $\theta$ on the line segment $\mathcal L(\alpha)\coloneq \{(1-t)\theta_0^\star(0) + t\theta_0^\star(\alpha):t\in[0,1]\}$,
			$$-\nabla_\theta^2 V_0(\theta;\alpha) \succeq mI\mathperiod$$
			Equivalently, $\sup_{\alpha\in[0,\alpha_\epsilon]}\sup_{\theta\in \mathcal L(\alpha)} \lVert \{\nabla_\theta^2 V_0(\theta;\alpha)\}^{-1}\rVert\le m^{-1}$.
		\end{enumerate}
	\end{assumption}
	
	\begin{lemma}
		\label{lemma:opt-theta-uniform-alpha}
		Fix $\alpha_\epsilon>0$.
		Suppose that
		\cref{assumption:identification:positivity,,assumption:bounded,,assumption:q-smooth,,assumption:v0-regularity-uniform} hold,
		that $\pi_{t,n}$ is sufficiently undersmoothed such that \cref{eqn:hal-undersmooth-condition} holds, and that
		\cref{assumption:cadlag,,assumption:hal-space,,assumption:hal-space-deriv} hold in a neighborhood of $\theta_0^\star(\alpha)$ uniformly over $\alpha\in[0,\alpha_\epsilon]$.
		\begin{enumerate}[label=(\alph*),ref=\thelemma(\alph*)]
			\item  \textbf{Uniform remainder control.} Writing the expansion from \cref{theorem:opt-theta-asym},
			\begin{equation*}
				\theta_n^\star(\alpha)-\theta_0^\star(\alpha) = (\P_n-\P_0)\left[ -H_0^{-1}(\alpha)\phi_{\nabla \psi}\{O;\eta_0,\eta_0',\theta_0^\star(\alpha),\alpha\}\right] + R_n(\alpha)\mathcomma
			\end{equation*}
			we have
			$$\sup_{\alpha\in[0,\alpha_\epsilon]} \lvert R_n(\alpha) \rvert= o_p(n^{-1/2})\mathperiod$$
			\item \textbf{Stochastic equicontinuity at $\alpha=0$.} 
			For any $\alpha_n=o_p(1)$, 
			\begin{equation*}
				(\P_n-\P_0) \big[ H_0(\alpha_n)^{-1}\phi_{\nabla\psi}\{\cdot;\theta_0^\star(\alpha_n),\alpha_n \} - H_0(0)^{-1}\phi_{\nabla\psi}\{\cdot;\theta_0^\star(0),0 \} \big]=o_p(n^{-1/2})\mathperiod
			\end{equation*}
			\item \textbf{Uniform drift bound.} For any $\alpha\in[0,\alpha_\epsilon]$, $\lVert\theta_0^\star(\alpha)-\theta_0^\star(0)\rVert\le C\alpha$ for some finite constant $C$.
		\end{enumerate}
	\end{lemma}
	
	\begin{proof}[Proof of \cref{lemma:opt-theta-uniform-alpha}]
		We begin with claim (a) and note that this result largely follows the logic of \cref{theorem:opt-theta-asym}.
		Let $\Delta_n(\alpha)\coloneq\theta_n^\star(\alpha)-\theta_0^\star(\alpha)$.
		From \cref{eqn:theta-n-taylor} we have that
		$$\Delta_n(\alpha)=-H_0^{-1}(\alpha)\big[\nabla_\theta \psi_n\{\theta_0^\star(\alpha);\alpha\} - \nabla_\theta \psi_0\{\theta_0^\star(\alpha);\alpha\} \big] - H_0(\alpha)^{-1}R_n(\alpha)\mathcomma$$
		where $R_n(\alpha)=[\nabla_\theta^2 V_n\{\tilde\theta_n(\alpha);\alpha\}- \nabla_\theta^2 V_0\{\theta_0^\star(\alpha);\alpha\}]\{\theta_n^\star(\alpha)-\theta_0^\star(\alpha) \}$
		and $\tilde\theta_n(\alpha)$ lies on the line segment between $\theta_n^\star(\alpha)$ and $\theta_0^\star(\alpha)$.
		\cref{lemma:psi-deriv-asymp} gives a pointwise result for the leading term:
		$$ \nabla_\theta \psi_n\{\theta_0^\star(\alpha);\alpha\} - \nabla_\theta \psi_0\{\theta_0^\star(\alpha);\alpha\} = (\P_n-\P_0)\{\phi_{\nabla \psi}(O;\eta_0,\theta,\alpha)\} + o_p(n^{-1/2})\mathperiod$$
		Using similar logic to \cref{lemma:uipw-asymp-uniform-alpha} and the strengthened versions of \cref{assumption:cadlag,,assumption:hal-space,,assumption:hal-space-deriv}, which hold uniformly over $\alpha\in[0,\alpha_\epsilon]$ and in a neighborhood of $\theta_0^\star(\alpha)$, we claim that result can be taken uniformly over $\alpha\in[0,\alpha_\epsilon]$.
		Next, we decompose the Hessian difference in $R_n(\alpha)$ as  
		\begin{align*}
			\nabla_\theta^2 V_n\{\tilde\theta_n(\alpha);\alpha\}- \nabla_\theta^2 V_0\{\theta_0^\star(\alpha);\alpha\} &= 
			\big[\nabla_\theta^2 V_n\{\tilde\theta_n(\alpha)\;\alpha\}- \nabla_\theta^2 V_0\{\tilde\theta_n(\alpha);\alpha\}\big]  \\
			&\quad\, + \big[\nabla_\theta^2 V_0\{\tilde\theta_n(\alpha);\alpha\}- \nabla_\theta^2 V_0\{\theta_0^\star(\alpha);\alpha\}\big]\mathperiod
		\end{align*}
		\cref{lemma:hessian-est-properties}(c) applied uniformly over $\alpha\in[0,\alpha_\epsilon]$ implies that the first bracketed term is uniformly $o_p(1)$ in $\alpha$.
		Moreover, \cref{lemma:hessian-est-properties}(b) gives Lipschitz continuity of $\nabla_\theta^2V_0(\cdot;\alpha)$ in $\theta$ and $\sup_{\alpha\in[0,\alpha_\epsilon]}\lVert \tilde\theta_n^\star(\alpha)-\theta_0^\star(\alpha) \rVert=o_p(1)$; hence, the second term is uniformly $o_p(1)$ as well.
		Consequently,
		$$\sup_{\alpha\in[0,\alpha_\epsilon]} \lVert R_n(\alpha)\rVert \le o_p(1) \sup_{\alpha\in[0,\alpha_\epsilon]}\lVert \Delta_n(\alpha)\rVert\mathperiod$$
		Returning to \cref{eqn:theta-n-taylor-2} and taking norms we have
		\begin{align*}
			\sup_{\alpha\in[0,\alpha_\epsilon]}\lVert \Delta_n(\alpha)\rVert &\le \sup_{\alpha\in[0,\alpha_\epsilon]}\lVert H_0(\alpha)^{-1}\rVert O_p(n^{-1/2}) + o_p(1)    \sup_{\alpha\in[0,\alpha_\epsilon]}\lVert \Delta_n(\alpha)\rVert\mathperiod \\
			&= O_p(n^{-1/2}) + o_p(1)\sup_{\alpha\in[0,\alpha_\epsilon]}\lVert \Delta_n(\alpha)\rVert\mathcomma
		\end{align*}
		noting that $\lVert H_0(\alpha)^{-1}\rVert$ is uniformly bounded over $[0,\alpha_\epsilon]$ by \cref{assumption:v0-regularity-uniform}.
		By parallel logic to \cref{theorem:opt-theta-asym}, we obtain $\sup_{\alpha\in[0,\alpha_\epsilon]}\lVert \Delta_n(\alpha)\rVert=O_p(n^{-1/2})$ and $\sup_{\alpha\in[0,\alpha_\epsilon]}\lVert R_n(\alpha)\rVert=o_p(n^{-1/2})$, establishing claim (a).
		
		We continue with claim (b). 
		Let 
		$$g_\alpha(O)\coloneq H_0(\alpha)^{-1}\phi_{\nabla\psi}\{O;\eta_0,\eta_0',\theta_0^\star(\alpha),\alpha\}\mathperiod$$
		We show that for any $\alpha_n=o_p(1)$, $(\P_n-\P_0)\{g_{\alpha_n}-g_0\}=o_p(n^{-1/2})$.
		For any $\delta\in(0,\alpha_\epsilon]$, consider the function class
		$$\mathcal G_\delta \coloneq \{g_\alpha-g_0 : \alpha \in [0,\delta]\}\mathperiod$$
		By the same smoothness-in-$\alpha$ arguments used in \cref{lemma:uipw-asymp-uniform-alpha} together with uniform curvature/invertibility from \cref{assumption:v0-regularity-uniform}, the mapping $\alpha\mapsto\theta_0^\star(\alpha)$ is locally Lipschitz (see the proof of part (c)), $\alpha\mapsto H_0(\alpha)^{-1}$ is uniformly bounded and locally Lipschitz, and $(\theta,\alpha)\mapsto \phi_{\nabla\psi}(\cdot;\theta,\alpha)$ is Lipschitz in $L_2(\P_0)$ under moment/smoothness conditions invoked in \cref{theorem:uipw-asym} and \cref{lemma:uipw-asymp-uniform-alpha}.
		Consequently, there exists some finite $L$ such that $\lVert g_\alpha -g_{\alpha'}\rVert_{2,\P_0}\le L\lvert \alpha-\alpha'\rvert$ for all $\alpha,\alpha'\in[0,\alpha_\epsilon]$, and hence $\sup_{\alpha\in[0,\delta]}\lVert g_\alpha-g_0\rVert_{2,\P_0} \le L\delta$.
		Taking $\delta=\alpha_n=o_p(1)$, \cref{corollary:maximal-inequality-shortcut} yields
		$$\sup_{\alpha\in[0,\alpha_n]}\lVert (\P_n-\P_0)(g_\alpha-g_0)\rVert = o_p(n^{-1/2})\mathcomma$$
		and hence
		$$(\P_n-\P_0)(g_{\alpha_n}-g_0) = o_p(n^{-1/2})\mathperiod$$
		
		We conclude with claim (c). Under the first order conditions we have
		$$\nabla_\theta V_0\{\theta_0^\star(\alpha);\alpha\}-\nabla_\theta V_0\{\theta_0^\star(0);\alpha\}=-\big[\nabla_\theta V_0\{\theta_0^\star(0);\alpha\}-\nabla_\theta V_0\{\theta_0^\star(0);0\} \big]\mathperiod$$
		By the Mean Value Theorem, there exists some $\tilde\theta_\alpha\in\mathcal L(\alpha)$ (i.e., on the line segment between $\theta_0^\star(0)$ and $\theta_0^\star(\alpha)$) such that
		$$\nabla V_0\{\theta_0^\star(\alpha);\alpha\}-\nabla_\theta V_0\{\theta_0^\star(0);\alpha\}=\nabla^2_\theta V_0(\tilde\theta_\alpha;\alpha)\{\theta_0^\star(\alpha)-\theta_0^\star(0)\}\mathperiod$$
		Combining and taking norms,
		$$\lVert \nabla^2_\theta V_0(\tilde\theta_\alpha;\alpha)\{\theta_0^\star(\alpha)-\theta_0^\star(0)\} \rVert = \lVert \nabla_\theta V_0\{\theta_0^\star(0);\alpha\}-\nabla_\theta V_0\{\theta_0^\star(0);0\}\rVert\mathperiod$$
		Then, by \cref{assumption:v0-regularity-uniform}, the Hessian is negative definite with eigenvalues uniformly bounded away from $0$ and there exists some $m>0$ such that 
		$$\sup_{\alpha,\theta\in[0,\alpha_\epsilon]\times \mathcal N_r(\Theta_\star)}\lVert\{\nabla_\theta^2 V_0(\theta;\alpha)\}^{-1} \rVert\le m^{-1}\mathperiod$$
		It follows that
		$$\lVert \theta_0^\star(\alpha)-\theta_0^\star(0)\rVert \le m^{-1} \lVert \nabla_\theta V_0\{\theta_0^\star(0);\alpha\}-\nabla_\theta V_0\{\theta_0^\star(0);0\}\rVert\mathperiod$$
		We proceed by bounding the second term.
		Note that, as the penalty $\mathcal P(\theta)$ does not depend on $\alpha$, the $\alpha$-dependence of $V_0(\theta;\alpha)=\psi_0(\theta;\alpha)-\mathcal P(\theta)$ is inherited entirely from $\psi_0(\theta;\alpha)$.
		By similar arguments used in \cref{theorem:uipw-asym} and \cref{lemma:uipw-asymp-uniform-alpha}, we have
		$$\sup_{\theta\in\Theta}\lVert \nabla_\theta V_0(\theta;\alpha)-\nabla_\theta V_0(\theta;0) \rVert \le L \alpha$$
		for all $\alpha\in[0,\alpha_\epsilon]$, where $L$ is a finite constant.
		Evaluating at $\theta=\theta_0^\star(\alpha)$ and combining with the previous result yields
		$\lVert \theta_0^\star(\alpha)-\theta_0^\star(0)\rVert \le m^{-1}L\alpha$.
		Hence, for $C\coloneq m^{-1}L<\infty$ we have 
		$$\lVert \theta_0^\star(\alpha)-\theta_0^\star(0)\rVert \le C\alpha$$
		for all $\alpha\in[0,\alpha_\epsilon]$, establishing the drift bound.
	\end{proof}

	\begin{proof}[Proof of \cref{corollary:opt-theta-asym-alpha}]
		Consider the decomposition
		\begin{equation*}
			\theta_n^\star(\alpha_n)-\theta_0^\star(0) = \underbrace{\big\{\theta_n^\star(\alpha_n)-\theta_0^\star(\alpha_n)\big\}}_{\text{estimation error}} + \underbrace{\big\{\theta_0^\star(\alpha_n)-\theta_0^\star(0)\big\}}_{\text{estimand drift}}\mathperiod
		\end{equation*}
		
		We begin with the estimation error.
		\cref{theorem:opt-theta-asym} gives, for any fixed $\alpha\ge0$,
		\begin{equation*}
			\theta_n^\star(\alpha)-\theta_0^\star(\alpha) = (\P_n-\P_0) H_0(\alpha)^{-1}\phi_{\nabla\psi}\{\cdot;\theta_0^\star(\alpha),\alpha\} + o_p(n^{-1/2})\mathperiod
		\end{equation*}
		The first result of \ref{lemma:opt-theta-uniform-alpha} gives
		$$ \sup_{\alpha\in[0,\alpha_\epsilon]}\Big\lVert \{\theta_n^\star(\alpha)-\theta_0^\star(\alpha)\} - (\P_n-\P_0)\big[H_0(\alpha)^{-1}\phi_{\nabla\psi}\{\cdot;\theta_0^\star(\alpha),\alpha\}\big] \Big\rVert=o_p(n^{-1/2})\mathperiod$$
		Given that $\alpha_n\to0$, we have $\alpha_n\in [0,\alpha_\epsilon]$ with high probability and hence,
		\begin{equation*}
			\theta_n^\star(\alpha_n)-\theta_0^\star(\alpha_n) = (\P_n-\P_0) H_0(\alpha_n)^{-1}\phi_{\nabla\psi}\{\cdot;\theta_0^\star(\alpha_n),\alpha_n\} + o_p(n^{-1/2})\mathperiod
		\end{equation*}
		Next, \cref{lemma:opt-theta-uniform-alpha}(b) gives stochastic equicontinuity of the leading empirical process term at $\alpha=0$. Thus,
		\begin{equation*}
			(\P_n-\P_0) \big[ H_0(\alpha_n)^{-1}\phi_{\nabla\psi}\{\cdot;\theta_0^\star(\alpha_n),\alpha_n \} - H_0(0)^{-1}\phi_{\nabla\psi}\{\cdot;\theta_0^\star(0),0 \} \big]=o_p(n^{-1/2})\mathperiod
		\end{equation*}
		Combining the two previous displays yields
		\begin{equation*}
			\theta_n^\star(\alpha_n)-\theta_0^\star(\alpha_n) = (\P_n-\P_0)H_0(0)^{-1}\phi_{\nabla\psi}\{\cdot;\theta_0^\star(0),0 \}+o_p(n^{-1/2})\mathperiod
		\end{equation*}
		It remains to show that the estimand drift is negligible.
		\cref{lemma:opt-theta-uniform-alpha}(c) gives the bound
		$$\lVert \theta_0^\star(\alpha)-\theta_0^\star(0)\rVert \le C\alpha$$
		for all $\alpha\in[0,\alpha_\epsilon]$, for $C<\infty$.
		Evaluating at $\alpha=\alpha_n=o_p(n^{-1/2})$ gives 
		$$\lVert \theta_0^\star(\alpha_n)-\theta_0^\star(0)\rVert \le C\alpha_n=o_p(n^{-1/2})\mathperiod$$
		Therefore, the estimand drift does not contribute to the first-order asymptotic expansion.
		Returning to the decomposition, we conclude that 
		\begin{equation*}
			\theta_n^\star(\alpha_n)-\theta_0^\star(\alpha_n) = (\P_n-\P_0) H_0(0)^{-1}\phi_{\nabla\psi}\{\cdot;\theta_0^\star(0),0\} + o_p(n^{-1/2})\mathperiod
		\end{equation*}
		The final result follows from the Central Limit Theorem.
	\end{proof}
	
	\section{Proofs Regarding Data-Adaptive Tilting}
	\label{sec:proof-alpha-tilting}
	
	\begin{proof}[Proof of \cref{lemma:alpha-clip-op}]
		Let $L_n(\alpha)\coloneq \bar V_n(\alpha)-\kappa n^{-1/2}\bar\sigma_n(\alpha)$ so that $\alpha_n=\argmax_{\alpha\ge0} L_n(\alpha)$.
		We first show that profiling over $\theta$ preserves the relevant convergence rate.
		For each $\alpha\in[0,\delta]$, 
		\begin{align*}
			\lvert\bar{V}_n(\alpha) - \bar V_0(\alpha) \rvert = \lvert \sup_{\theta\in\Theta}V_n(\theta;\alpha)-\sup_{\theta\in\Theta} V_0(\theta;\alpha) \rvert \le \sup_{\theta\in\Theta} \lvert V_n(\theta;\alpha)-V_0(\theta;\alpha) \rvert\mathperiod
		\end{align*}
		Moreover, since $V_n(\theta;\alpha)-V_0(\theta;\alpha)=\psi_n(\theta;\alpha)-\psi_0(\theta;\alpha)$, it follows that 
		\begin{equation*}
			\sup_{\alpha\ge 0}\lvert \bar V_n(\alpha) - \bar V_0(\alpha)\rvert \le \sup_{\alpha\ge0}\sup_{\theta\in\Theta}\lvert \psi_n(\theta;\alpha)-\psi_0(\theta;\alpha)\rvert\mathperiod
		\end{equation*}
		By the global uniform consistency of $\psi_n(\theta;\alpha)$ we have that
		\begin{equation}
			\label{eqn:profile-bound-global}
			\sup_{\alpha\ge0}\lvert \bar V_n(\alpha)-\bar V_0(\alpha)\rvert = o_p(1)\mathperiod
		\end{equation}
		Moreover, by \cref{lemma:uipw-asymp-uniform-alpha} we have the sharper rate over $\alpha\in[0,\delta]$:
		\begin{equation}
			\label{eqn:profile-bound-delta}
			\sup_{\alpha\in[0,\delta]}\lvert \bar V_n(\alpha) - \bar V_0(\alpha)\rvert \le \sup_{(\theta,\alpha)\in\Theta\times[0,\delta]}\lvert \psi_n(\theta;\alpha)-\psi_0(\theta;\alpha)\rvert=O_p(n^{-1/2})\mathperiod
		\end{equation}
		
		We next show that the maximizer $\alpha_n$ lies in the interval $[0,\delta]$ with probability tending to $1$.
		By the global separation condition, there exists some $\epsilon_\delta>0$ such that 
		$$\inf_{\alpha>\delta}\{\bar V_0(0)-\bar V_0(\alpha)\} \ge \epsilon_\delta\mathperiod$$
		Then, because $\kappa\ge0$ and $\bar\sigma_n(\alpha)\ge0$, $L_n(\alpha)=\bar V_n(\alpha)-\kappa n^{-1/2}\bar\sigma_n(\alpha)\le \bar V_n(\alpha)$ for every $\alpha\ge0$.
		Hence,
		$$\sup_{\alpha>\delta} L_n(\alpha)\le \sup_{\alpha>\delta} \bar V_n(\alpha)\mathperiod$$
		Using \cref{eqn:profile-bound-global} and global separation,
		\begin{equation}
			\label{eqn:Ln-global}
			\sup_{\alpha>\delta} L_n(\alpha) \le \sup_{\alpha>\delta} \bar V_0(\alpha) + \sup_{\alpha\ge0}\lvert\bar V_n(\alpha) - \bar V_0(\alpha) \rvert \le \bar V_0(0) - \epsilon_\delta + o_p(1)\mathperiod
		\end{equation}
		On the other hand, $L_n(0)=\bar V_n(0) - \kappa n^{-1/2}\bar\sigma_n(0)$.
		By \cref{eqn:profile-bound-delta} taken at $\alpha=0$ and the assumed standard error bound we have $\bar V_n(0)-\bar V_0(0)=O_p(n^{-1/2})$ and $n^{-1/2} \bar\sigma_n(0)=O_p(n^{-1/2})$.
		Therefore
		\begin{equation}
			\label{eqn:Ln-0}
			L_n(0) = \bar V_0(0) + O_p(n^{-1/2})\mathperiod
		\end{equation}
		Combining \cref{eqn:Ln-global,,eqn:Ln-0}, we obtain $\P_0\{ \sup_{\alpha>\delta} L_n(\alpha) < L_n(0)\}\to1$ and hence
		\begin{equation*}
			\P_0(\alpha_n\in[0,\delta])\to 1\mathperiod
		\end{equation*}
		
		Condition on the high probability event, $\{\alpha_n\in[0,\delta]\}$. 
		It remains to show the local rate
		Define $r_n\coloneq \sup_{\alpha\in[0,\delta]} \lvert L_n(\alpha) - \bar V_0(\alpha)\rvert$.
		By \cref{eqn:profile-bound-delta} and $\sup_{\alpha\in[0,\delta]} n^{-1/2}\bar\sigma_n(\alpha)=O_p(n^{-1/2})$, we have $r_n=O_p(n^{-1/2})$.
		On the event $\{\alpha_n\in[0,\delta]\}$, we have $L_n(\alpha)\ge L_n(0)$ and
		\begin{align*}
			\bar V_0(0) - \bar V_0(\alpha_n) &= \{\bar V_0(0) - L_n(0)\} + \{L_n(0)-L_n(\alpha_n)\} + \{L_n(\alpha_n)-\bar V_0(\alpha_n)\} \\
			&\le r_n + 0 + r_n \\
			&= 2r_n\mathperiod
		\end{align*}
		By local linear separation, $\bar V_0(0) - \bar V_0(\alpha)\ge c\alpha$ for all $\alpha\in[0,\delta]$.
		Given that $\alpha_n \in [0,\delta]$, we obtain
		\begin{equation*}
			c \alpha_n\le \bar V_0(0) - \bar V_0(\alpha_n) \le 2r_n\mathperiod
		\end{equation*}
		Hence, $\alpha \le 2c^{-1}r_n=O_p(n^{-1/2})$ on an event with probability tending toward $1$.
		It follows that
		$$\alpha_n=O_p(n^{-1/2})\mathperiod$$
		
		Finally, since $\alpha_n=O_p(n^{-1/2})$, $P(\alpha_n < n^{-1/4})\to 1$.
		On this event, the clipped selector satisfies $\tilde\alpha_n=0$.
		Therefore $P(\tilde\alpha_n=0)\to1$. Because $\tilde\alpha_n\ge0$ this implies $\tilde\alpha_n=o_p(n^{-1/2})$.
	\end{proof}
	
	\section{Potential for Efficiency Gains in Deterministic Versus Stochastic Regimes}
	\label{sec:ipw-aipw-efficiency}
	
	Our proposed UIPW framework targets stochastic policy classes and uses tilting to stabilize longitudinal inverse-probability weights in high–decision-point settings.  
	A recurring empirical pattern in our simulations is that, relative to oracle or correctly specified parametric IPW benchmarks, the incremental efficiency gains from undersmoothing can be modest. 
	A key reason is structural: stochastic intervention targets are typically closer to the observed treatment mechanism, so the corresponding change-of-measure weights are less variable and often bounded, leaving less variance for augmentation to reduce.
	
	To isolate this phenomenon in the simplest possible setting, we conducted a point-exposure simulation comparing a deterministic estimand, $\P_0\{Y(1)\}$, to a continuum of incremental propensity score (IPS; \citealp{kennedy_nonparametric_2019}) stochastic estimands $\P_0[Y\{Q(\delta)\}]$ indexed by an odds-multiplier $\delta>1$.  
	IPS interventions provide a natural bridge between stochastic and deterministic regimes: as $\delta$ increases, the intervention shifts treatment probabilities toward $1$ and approaches an always-treat intervention wherever $\pi(x)>0$, so the IPS targets become increasingly near deterministic.  
	This setup therefore directly probes how the potential efficiency gain changes as we move from stochastic to near-deterministic regimes.
	
	We use the \cite{kang_demystifying_2007} design employed in \cite{kennedy_nonparametric_2019}'s simulation study for IPS interventions:
	$X=(X_1,\ldots,X_4)\sim N(0,I_4)$,
	$\P_0(A=1\mid X) = \expit(-X_1+0.5X_2-0.25X_3-0.1X_4)$,
	$(Y\mid A,X)\sim N\{\mu(X,A),1\}$, and
	$\mu( x,a)=200+a\{10+13.7(2x_1 + x_2 + x_3 + x_4)\}$.
	This design is known to induce variable propensities and highlight instability in deterministic weighting estimands.
	Of note, the original design includes an intercept of $200$, which makes the IPW estimator incredibly variable; we additional perform simulations with an intercept of $-10$ such that $\P_0\{\mu(X,1)\}=0$.
	
	We compare two estimands: $\psi_0^\text{Det} \coloneq \P_0\{Y(1)\}$, and  $\psi_0^{\text{IPS}}(\delta)\coloneq \P_0[Y\{Q(\delta)\}]$, where the IPS intervention replaces $\pi(X)$ with the odds-shifted treatment probability $q_\delta(X)=\delta\pi(X)/\{\delta\pi(X)+1-\pi(X)\}$.
	
	To focus on efficiency, we consider correctly specified parametric models for both the propensity score and outcome regression throughout our simulations. 
	We consider the following estimators. 
	For $\psi_0^{\text{Det}}$ we consider (i) standard IPW using $\pi_n$, and (ii) augmented IPW (AIPW) that additionally uses $\mu_n$ to debias.
	For $\psi_0^{\text{IPW}}(\delta)$ we consider: (i) IPS-IPW using a weighted mean representation: $\P_n[(\delta A + 1-A)/\{\delta\pi_n(X)+1-\pi_n(X)\}Y]$, and (ii) an EIF-based estimator based on \citet[Corollary 2]{kennedy_nonparametric_2019}.
	
	We performed 1000 Monte Carlo simulations for each $\delta$ with $\delta=2^1,2^{1.5},\ldots,2^{5}$.
	We summarized the Monte Carlo variance of each estimator and, for each estimand, reported the relative efficiency of the AIPW estimator to the IPW benchmark (\cref{fig:ips-aipw-ipw-efficiency}).
	Our results suggest that there is a larger potential for efficiency gains when targeting deterministic, rather than stochastic, regimes.
	
	\begin{figure}
		\centering
		\includegraphics[width=\linewidth]{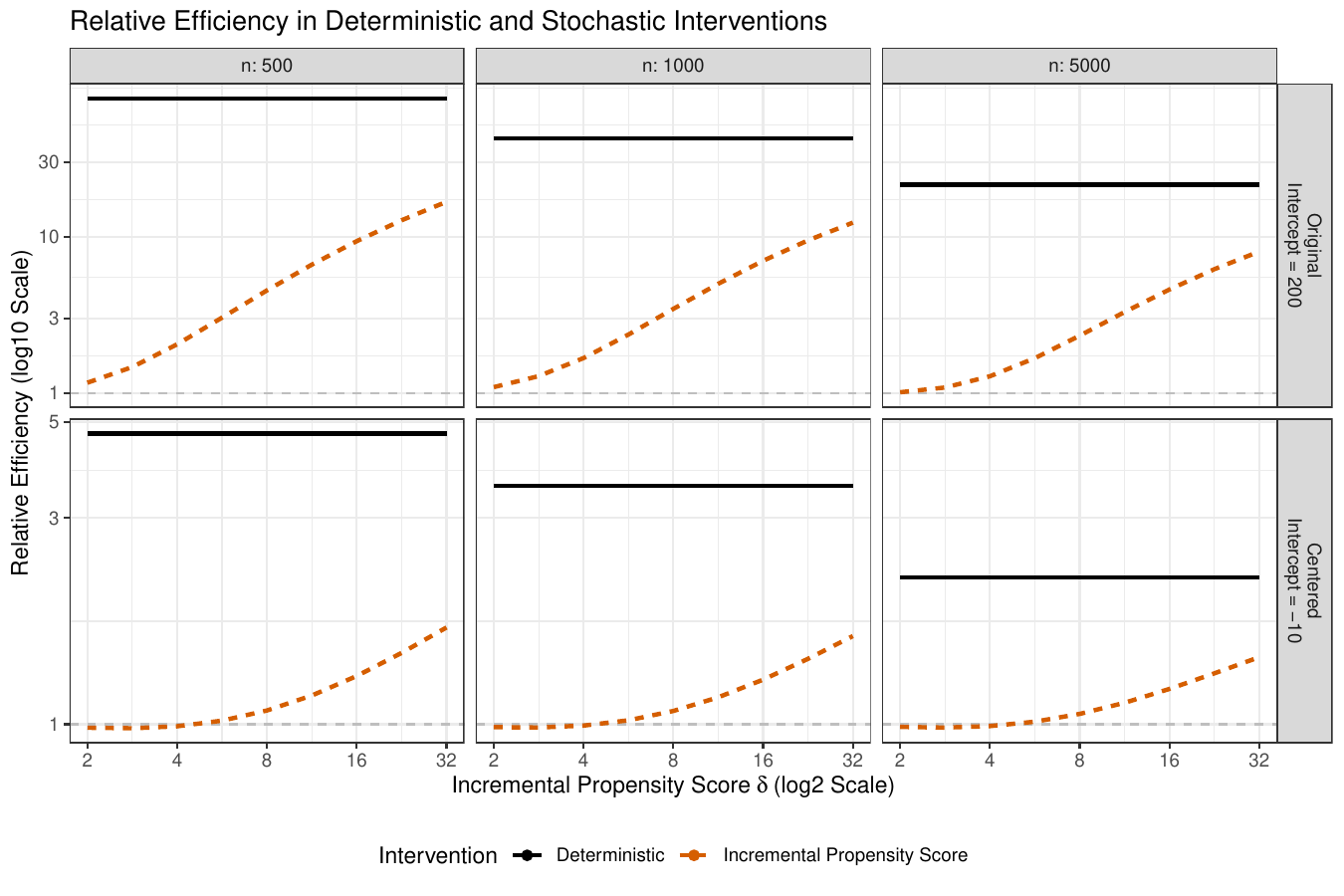}
		\caption{Relative efficiency, defined as $\operatorname{Var}(\psi_n^{\text{IPW}})/\operatorname{Var}(\psi_n^{\text{AIPW}})$, plotted versus the IPS odds-multiplier $\delta$.}
		\label{fig:ips-aipw-ipw-efficiency}
	\end{figure}
	
	At multiple decision points, the same qualitative mechanism typically amplifies: deterministic regime evaluation produces products of inverse propensity terms over time, which can yield highly variable (and often heavy-tailed) weights even when each single-step propensity is reasonable. 
	In contrast, stochastic regimes that remain close to the observed treatment mechanism tend to keep each step’s Radon–Nikodym factor nearer to $1$ (and often bounded), so the longitudinal weight product is substantially more stable. 
	As a result, baseline IPW estimators under stochastic regimes are already closer to their efficient limits, leaving less variance for EIF/AIPW-style corrections (or undersmoothing refinements) to remove.
	
\end{document}